\documentclass[aps,
    reprint,
    nofootinbib,
]{revtex4-2}
\usepackage{xprintlen}
\usepackage[T1]{fontenc}
\usepackage{times}

\usepackage[table,usenames,dvipsnames]{xcolor}

\usepackage{hyperref}
\hypersetup{
  colorlinks=true,
  hypertexnames=false,
  linktocpage,
  colorlinks=true, 
  urlcolor=magenta!90!black,    
  linkcolor=blue!60!black, 
  citecolor=black!60 
}

\usepackage{physics,siunitx}
\usepackage{bbm,mathtools}
\usepackage{amsthm,amssymb,amsmath}
\usepackage{MnSymbol} 
\usepackage{youngtab}

\usepackage{booktabs}
\usepackage{enumitem}

\usepackage[capitalize,compress]{cleveref}

\usepackage{complexity}

\usepackage{pgfplots}
\usepackage{tikz}
\usepackage{pgffor}
\usepackage{algorithm}
\usepackage{algpseudocode}
\algrenewcommand\algorithmicrequire{\textbf{Input:}}
\algrenewcommand\algorithmicensure{\textbf{Output:}}

\usepgfplotslibrary{
  fillbetween,
  }

\usetikzlibrary{
  pgfplots.groupplots,
  arrows.meta,
  backgrounds,
  chains,
  decorations.pathreplacing,
  decorations.pathmorphing,
  fit,
  intersections,
  through,
  positioning,
  shapes,
  shapes.geometric,
  shapes.arrows,
  calc,
  3d
  }

\newclass{\stoqma}{StoqMA}
\newclass{\classP}{P}
\newclass{\bqp}{BQP}
\newclass{\qcam}{QCAM}
\newclass{\postbqp}{postBQP}
\newclass{\posta}{postA}
\newclass{\postiqp}{postIQP}
\newclass{\classa}{A}
\newclass{\bpp}{BPP}
\newclass{\fbpp}{FBPP}
\newclass{\pp}{PP}
\newclass{\cocp}{coC_=P}
\newclass{\ph}{PH}
\newclass{\np}{NP}
\newclass{\conp}{coNP}
\newclass{\gapp}{GapP}
\newclass{\approxclass}{Apx}
\newclass{\gapclass}{Gap}
\newclass{\sharpP}{\#P}
\newclass{\ma}{MA}
\newclass{\am}{AM}
\newclass{\qma}{QMA}

\newclass{\hog}{HOG}
\newclass{\quath}{QUATH}
\newclass{\bog}{BOG}
\newclass{\xeb}{XEB}
\newclass{\xhog}{XHOG}
\newclass{\xquath}{XQUATH}
\newclass{\maxcut}{MAXCUT}
\newclass{\sat}{SAT}
\newclass{\maxtwosat}{MAX2SAT}
\newclass{\twosat}{2SAT}
\newclass{\threesat}{3SAT}
\newclass{\sharpsat}{\#SAT}
\newclass{\se}{Sign Easing}
\newclass{\classx}{X}

\newtheorem{theorem}{Theorem}

\newtheorem{lemma}[theorem]{Lemma}

\DeclareMathOperator{\sign}{sign}

\DeclareMathOperator{\spn}{span}
\DeclareMathOperator{\rad}{rad}
\DeclareMathOperator{\vecz}{vec}

\newcommand{\mc}{\mathcal}
\newcommand{\mb}{\mathbb}

\newcommand{\id}{\mathbbm{1}}
\newcommand{\ee}{\mathrm{e}}
\newcommand{\ii}{\mathrm{i}}

\newcommand{\bin}{\{0,1\}}

\newcommand{\proj}[1]{\ket{#1}\bra{#1}}

\DeclareMathOperator{\acc}{Acc}
\DeclareMathOperator{\rej}{Rej}
\DeclareMathOperator{\gap}{gap}

\usepackage{bibunits}
\defaultbibliography{references,suppmat}
\defaultbibliographystyle{./myapsrev4-1}

\makeatletter 
 \hypersetup{
  pdftitle = {Bell sampling from quantum circuits
    },
  pdfauthor = {Dominik Hangleiter, Michael Gullans},
  pdfkeywords = {
    quantum computing, 
    verification, 
    quantum supremacy, 
    quantum advantage, 
    Bell sampling, 
    stabilizer states, 
    entanglement,
    quantum circuit,
    classical shadow
    }
  }
\makeatother

\newcommand{\QuICS}{
Joint Center for Quantum Information and Computer Science, NIST/University of Maryland, College Park, Maryland 20742, USA}

\usepackage{xprintlen}

\begin{document}

\title{Bell sampling from quantum circuits}
\author{Dominik Hangleiter}
\email{mail@dhangleiter.eu}
\author{Michael J. Gullans}
\email{mgullans@umd.edu}
\affiliation{\QuICS}
\date{\today}

\begin{abstract}
A  central challenge in the verification of quantum computers is benchmarking their performance as a whole and demonstrating their computational capabilities.
In this work, we find a universal model of quantum computation, \emph{Bell sampling}, that can be used for both of those tasks and thus provides an ideal stepping stone towards fault-tolerance. 
In Bell sampling, we measure two copies of a state prepared by a quantum circuit in the transversal Bell basis. 
We show that the Bell samples are classically intractable to produce and at the same time constitute what we call a \emph{circuit shadow}: 
from the Bell samples we can efficiently extract  information about the quantum circuit preparing the state, as well as  diagnose circuit errors. 
In addition to known properties that can be efficiently extracted from Bell samples, we give several new and efficient protocols: an estimator of state fidelity, a test for the depth of the circuit and an algorithm to estimate a lower bound to the number of $T$ gates in the circuit. 
With some additional measurements, our algorithm learns a full description of states prepared by circuits with low $T$-count.
\end{abstract}

\maketitle
\begin{bibunit}


\paragraph*{Introduction. }
As technological progress on fault-tolerant quantum processors continues, a central challenge is to demonstrate their computational advantage and to benchmark their performance as a whole.
Quantum random sampling experiments serve this double purpose \cite{boixo_characterizing_2018,liu_benchmarking_2022,heinrich_general_2022,hangleiter_computational_2023-1} and have arguably surpassed the threshold of quantum advantage \cite{arute_quantum_2019,zhong_quantum_2020,zhu_quantum_2022,madsen_quantum_2022,morvan_phase_2023,deng_gaussian_2023}.
However, this approach currently suffers several drawbacks. 
Most importantly, it can only serve its central goals---benchmarking and certification of quantum advantage---in the classically simulable regime. 
This deficiency arises because evaluating the performance benchmark, the \emph{cross-entropy benchmark}, requires a classical simulation of the ideal quantum computation.
What is more, the cross-entropy benchmark suffers from various problems related to the specific nature of the physical noise in the quantum processor \cite{gao_limitations_2021,morvan_phase_2023,ware_sharp_2023}, and  yields limited information about the underlying quantum state. 
More generally, in near-term quantum computing without error correction, we lack many tools for validating a given quantum computation just using its output samples.

In this work, we consider \emph{Bell sampling}, a model of quantum computation in which two identical copies of a state prepared by a quantum circuit are measured in the transversal Bell basis, see \cref{fig:bell sampling}. 
We show that this model is universal for quantum computation, that the output samples yield a variety of diagnostic information about the underlying quantum state and that they allow for detecting and correcting certain errors in the state preparation. 
We may thus think of the Bell samples as classical \emph{circuit shadows}, in analogy to the notion of state shadows coined by \textcite{aaronson_learnability_2007,huang_predicting_2020}, since we can efficiently extract specific information about the generating circuit or a family of generating circuits from them. 
Bell sampling also serves as a stepping stone towards quantum fault-tolerance: 
not only can we naturally detect certain errors from the Bell samples, but the protocol is also compatible with stabilizer codes---the Bell measurement between code blocks is transversal for such codes and allows for the fault-tolerant extraction of all error syndromes.  
As a concrete application, we demonstrate that Bell sampling from universal quantum circuits exhibits quantum advantage that can be efficiently validated on near-term quantum processors. 

Technically, we make the following contributions.
We show that Bell sampling is universal and provide complexity-theoretic evidence for the classical intractability of Bell sampling from random universal quantum circuits, following an established hardness argument~\cite{aaronson_computational_2013,bremner_average-case_2016,hangleiter_computational_2023-1}. 
We give two new diagnostic primitives based on noiseless Bell samples. 
First, we introduce a new test to verify the depth of quantum circuits. 
Here, we make use of the  fact that from the Bell basis samples one can compute correlation properties of the two copies and in particular a swap test on any subsystem.
Second, we show that the Bell samples can be used to efficiently measure the stabilizer nullity---a magic monotone \cite{beverland_lower_2020}---and give a protocol to efficiently learn a full description of any quantum state that can be prepared by a circuit with low $T$-count. 
Here, we build on a result by \textcite{montanaro_learning_2017}, who has shown that stabilizer states can be learned from Bell samples. 
In the setting of noisy state preparations, we analytically show that the Bell samples can be used to estimate the fidelity of state preparations and demonstrate the feasibility numerically.
We also give a protocol for efficiently detecting errors in the state preparation based only on the properties of the Bell samples.

\begin{figure}[b]
  \includegraphics{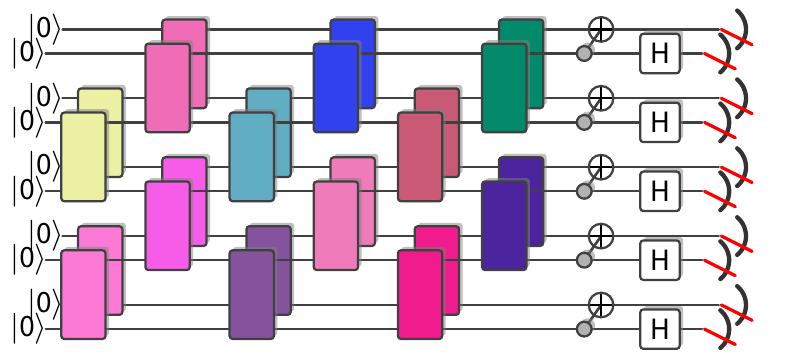}
  \caption{\label{fig:bell sampling} \textbf{The Bell sampling protocol.} In the Bell sampling protocol we prepare  the quantum state $C \ket {0^n}\otimes C \ket {0^n}$ using a quantum circuit $C$, and measure all qubits transversally in the Bell basis across the bipartition of the system. }
\end{figure}
Of course, the idea to sample in the Bell basis to learn about properties of quantum states is as old as the theory of quantum information itself and has found many applications in quantum computing, including learning stabilizer states \cite{montanaro_learning_2017}, testing stabilizerness \cite{gross_schurweyl_2021}, measuring magic \cite{haug_scalable_2023,haug_efficient_2023}, and quantum machine learning \cite{huang_information-theoretic_2021}. 
The novelty of our approach is to view Bell sampling as a computational model. 
We then ask the question: 
What can we learn from the Bell samples about the circuit preparing the underlying quantum state? 



\paragraph*{Bell sampling. }

We begin by defining the Bell sampling protocol and noting some simple properties that will be useful in the remainder of this work. 
Consider a quantum circuit $C$ acting on $n$ qubits, and define the Bell basis of two qubits as 
\begin{align}
\label{eq:bell basis}
  \ket{\sigma_r } &= (\sigma_r \otimes \id) \ket{\Phi^+},\text{ where }
  \ket{\Phi^+} = (\ket {00 } + \ket {11} )/\sqrt 2, 
\end{align}
and for $r \in \bin^2$ we identify
\begin{align}
   \sigma_{00} = \id, \quad \sigma_{01} = X, \quad \sigma_{10} = Z, \quad \sigma_{11} = \ii \,\sigma_{01}\sigma_{10} = Y . 
 \end{align} 
The Bell sampling protocol proceeds as follows, see \cref{fig:bell sampling}. 
\begin{enumerate}
  \item Prepare $\ket{\mc C} \coloneqq \ket C \otimes \ket C  \coloneqq C \ket{0^n} \otimes C \ket{0^n} $.
  \item Measure all qubit pairs $(i,i+n)$  for $i \in [n] \coloneqq \{1, 2, \ldots, n \}$ in the Bell basis, yielding an outcome $r \in \bin^{2n}$. 
\end{enumerate}

It is easy to see that the distribution of the outcomes $r$ can be written as
\begin{align}
\label{eq:output probs}
 P_C(r) & = \frac 1 {2^n} \left|\bra C\sigma_r  \ket{\overline C}\right|^2
\end{align}
where $\sigma_r = \sigma_{r_1 r_{n+1}} \otimes \sigma_{r_2 r_{n+2}} \otimes \cdots \otimes \sigma_{r_n r_{2n}}$ is the $n$-qubit Pauli matrix corresponding to the outcome $r = (r_1, r_2, \ldots, r_{2n})$, and $\overline C$ denotes complex conjugation of $C$. 
In order to perform the measurement in the Bell basis, we need to apply a depth-1 quantum circuit consisting of $n$ transversal \texttt{cnot}-gates followed by Hadamard gates on the control qubits and a measurement of all qubits in the computational basis. 



\paragraph*{Computational complexity. }
We first show that Bell sampling is a universal model of quantum computation. 
To show this, we observe that we can estimate both the sign and the magnitude of $\bra C Z \ket C$ for any quantum circuit $C$ from Bell samples from a circuit $C'(C)$ in which we use a variant of Ramsey interferometery with a single ancilla qubit in each copy of the circuit, see the Supplementary Material (SM) \cite{suppmat}.
We then show that approximately sampling from the Bell sampling distribution $P_C$ is classically intractable on average for universal random quantum circuits $C$ with $\Omega(n^2)$ gates in a brickwork architecture (as depicted in \cref{fig:bell sampling}), assuming certain complexity-theoretic conjectures are satisfied, via a standard proof technique, see the SM~\cite{suppmat} for details.
The argument puts the complexity-theoretic evidence for the hardness of Bell sampling from random quantum circuits on a par with that for standard universal circuit sampling~\cite{boixo_characterizing_2018,bouland_complexity_2019,movassagh_quantum_2020,kondo_quantum_2022,krovi_average-case_2022,arute_quantum_2019}.



\paragraph*{Bell samples as classical circuit shadows. }
Samples in the computational basis---while difficult to produce for random quantum circuits---yield very little information about the underlying quantum state. 
In particular, the problem of verification is essentially unsolved since the currently used methods require exponential computing time. 
In contrast, from the Bell samples, we can \emph{efficiently} infer many properties of the quantum state preparation $\ket C \otimes \ket C$. 
Known examples include the overlap $\tr[\rho \sigma]$ of a state preparation $\rho \otimes \sigma$ via a swap test, 
the magic of the state $\ket C$ \cite{haug_scalable_2023}, and the outcome of measuring any Pauli operator $P \otimes P$ \cite{gottesman_introduction_2009}.
Here, we add new properties to this family. 
We give efficient protocols for estimating the fidelity, testing the depth of low-depth quantum circuits, for testing its magic, and for learning quantum states that can be prepared by a circuit with low $T$-count.

Let us begin by recapping how a swap test can be performed using the Bell samples, and observing some properties that are useful in the context of benchmarking random quantum circuits. 
To this end, write the two-qubit swap operator $\mb S  = P_{\bigvee^2} - P_{\bigwedge^2}$
as the difference between the projectors onto the symmetric subspace $P_{\bigvee^2} = \proj{\sigma_{00}} + \proj{\sigma_{01}} + \proj{\sigma_{10}}$ and the antisymmetric subspace $P_{\bigwedge^2} = \proj{\sigma_{11}}$. 
The overlap 
$\tr[\rho \sigma] = \tr[(\rho \otimes \sigma) \mb S] $ can then be directly estimated up to error $\epsilon$ from $M \in O(1/\epsilon^{2})$ Bell samples as 
\begin{align}
  \frac1 M  \left(|\{r : \pi_Y(r) = 0 \}| - |\{r : \pi_Y(r) = 1 \}|\right). 
\end{align}

\begin{figure}
  \includegraphics{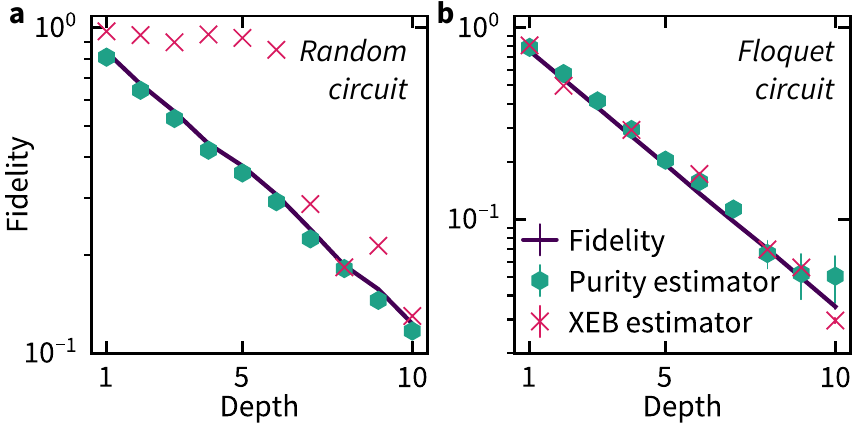}
  \caption{\textbf{Fidelity estimation based on noisy Bell sampling.} We simulate noisy Bell sampling and XEB including noisy measurements using $10^6$ samples and compute the fidelity (lines), and the purity (hexagons), and XEB (crosses) based estimators of fidelity for 
  \emph{(a)} typical Clifford circuits with two-qubit random gates in an all-to-all connected architecture with XEB $1$ on $n = 20$ qubits and Pauli $(X,Y,Z)$ error probabilies $p =  0.005 \cdot (1, 1/3,1/10)$, and 
  \emph{(b)} crystalline Floquet Clifford circuits that are scrambling \cite{sommers_crystalline_2023} on $18$ qubits in 1D with (depolarized) two-qubit gate fidelity $0.98$. Missing XEB points are due to ideal XEB values $0$. Error bars represent one standard deviation.
  \label{fig:noisy purity and fidelity}}
\end{figure}
For noisy quantum state preparations $\rho \otimes \rho$, we can thus estimate the purity $P = \tr [\rho^2]$ of $\rho$. 
We now argue that the purity is a good estimator for the fidelity $F=\bra C \rho \ket C$ of the state preparation if noise is local circuit-level Pauli noise. 
Pauli channels naturally appear as the effective noise after repeated rounds of syndrome extraction in stabilizer codes~\cite{beale_quantum_2018,huang_performance_2019}.
We can also ``force'' the noise into Pauli noise using randomized compiling implemented independently on two copies of a fixed circuit; 
this converts experimentally relevant noise into Pauli channels~\cite{wallman_noise_2016,winick_concepts_2022}. 
We show that the average fidelity of random circuits with local Pauli noise relates to the average purity as $\mb E_C F = \sqrt{\mb E_C P}$ if the Bell measurement is noise-free \cite{liu_benchmarking_2022,ware_sharp_2023}.
The root purity is thus a good estimator of the average fidelity from the Bell samples irrespective of the noise rate and depth of the circuit ensemble with precision $\sim 1/(F\sqrt M)$ in the number of Bell samples $M$. 
Assuming that the local Pauli noise channels are independently identically distributed, we can extend this estimator to noisy Bell measurements as  
\begin{align}
\label{eq:modified purity estimator}
   \overline F = (\mb E_C P)^{\frac E {2 (E + 2/3)}},
\end{align}
where $E$ is the number of error locations in the circuit prior to the measurement. 
We expect this estimator to be accurate even for typical and non-random circuits and give numerical evidence that this is the case in \cref{fig:noisy purity and fidelity} for  random Clifford circuits and non-random Floquet Clifford circuits.

How does the purity estimator compare to other means of estimating the fidelity of a quantum state? 
A widely used method is cross-entropy benchmarking (XEB), which is obtained from classical samples in the computational basis~\cite{arute_quantum_2019,choi_preparing_2023}. 
XEB is sample-efficient for random circuits, but requires computing ideal output probabilities of $C$, making it infeasible for already moderate numbers of qubits and non-Clifford gates. 
The XEB is a good estimator of the fidelity in the regime of low local error probabilities $\eta \lesssim 1/n$ and for depths $d \in \Omega(\log n)$, but not outside of those regimes~\cite{ware_sharp_2023,morvan_phase_2023} as witnessed in \cref{fig:noisy purity and fidelity}(a). 
In contrast, Bell sampling is computationally and sample efficient independently of the circuit, and the root purity estimator of fidelity is accurate in both the regimes of high noise and of low depths.  
Correlated coherent errors on both copies na\"ively ruins the correspondence between fidelity and root purity, but independent randomized compiling on the two copies recovers it by removing these correlations.

In the SM~\cite{suppmat}, we show these results, elaborate the various estimators and also discuss the relation of Bell sampling to different means of verifying quantum computations more generally.  
From now on, we will assume that the purity is close to unity. 

\paragraph*{Depth test.}
We now describe a Bell sampling protocol to measure the depth of a quantum circuit $C$ which is promised to be implemented in a fixed architecture, i.e., with gates applied in layers according to a certain pattern. 
The basic idea underlying the depth test is to use swap tests on subsystems of different sizes in order to obtain estimates of subsystem purities. 
For a subsystem $A$ of $[n]$, the subsystem purity is given by $ P_A(\rho)  = \tr[\rho_A^2]$,
where $\rho_A = \tr_{A^c}[\rho]$ is the reduced density matrix on subsystem $A \subset [n]$. 
It can be estimated from the fraction of outcome strings with even $Y$-parity $\pi_Y(r_A)$ on the substrings $r_A = (r_i,r_{n+i})_{i \in A}$. 

\begin{figure}
  \includegraphics{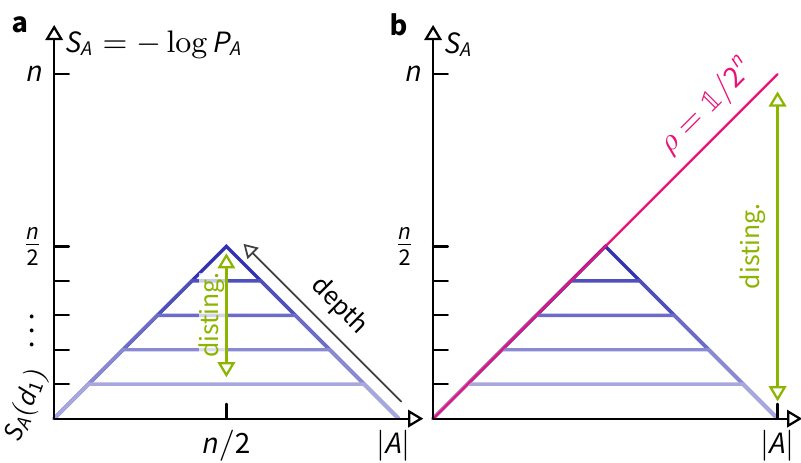}
  \caption{\label{fig:page curves} \textbf{Depth-dependent Page curves.} (a) The maximal subsystem entanglement entropy depends on the circuit architecture and depth (shades of blue) until the half-cut entanglement reaches its maximal value given by $n/2$. We measure the subsystem entropy at half-cuts to obtain the maximal sensitivity to different circuit depths. (b) We detect errors in the Bell samples by detecting strings that lead to a non-zero estimate of the purity of $\rho$. }
\end{figure}

Our test is based on the observation that the amount of entanglement generated by quantum circuits on half-cuts reaches a depth-dependent maximal value until it saturates at a circuit depth that depends on the dimensionality of the circuit architecture, see \cref{fig:page curves}(a) for an illustration. 
In order to lower-bound the depth of a circuit family we choose a subsystem size at which the distinguishability between different depths is maximal. 
This is typically the case at half-cuts, where the R\'enyi-$2$ entanglement entropy $S_A(\rho) = - \log P_A(\rho)$ can be at most $n/2$. 
At the same time, the entanglement entropy is bounded as a function of depth $S_A(d) \leq d|\partial A|$, where $\partial A$ is the number of gates applied across the boundary of $A$ in every layer of the circuit. 
We now compute an empirical estimate $\hat S_{n/2}$ of $S_A(\proj C)$ for a size-$n/2$ subsystem $A$ using the Bell samples and then compute the maximum $d$ such that $\hat S_{n/2} - \epsilon \geq d \cdot |\partial A|$ up to an error tolerance $\epsilon$ depending on the number of Bell samples. 
We can further refine this test for random quantum circuits by exploiting their average subsystem entanglement properties, known as the \emph{Page curve}~\cite{page_average_1993}. 
Depth-dependent Page curves have been computed analytically \cite{nahum_quantum_2017} and numerically \cite{sommers_crystalline_2023} for a few circuit architectures and random ensembles. 

We remark that these entanglement-based tests rely on universal features of quantum chaotic dynamics.  As a result, they are also expected to be applicable to generic Hamiltonian dynamics, similar to how ideas for standard quantum random sampling have recently been extended to this case~\cite{Mark23,Shaw23}.

\paragraph*{Magic test and Clifford+$T$ learning algorithm.}
Another primitive that can be exploited in property tests of quantum states using the Bell samples is the fact that for stabilizer states $\ket S$, the Bell distribution is supported on a coset of the stabilizer group of $\ket S$ \cite{montanaro_learning_2017}. 
Leveraging this property allows for efficiently learning stabilizer states~\cite{montanaro_learning_2017}, testing stabilizerness~\cite{gross_schurweyl_2021}, learning circuits with a single layer of $T$-gates~\cite{lai_learning_2022} and estimating measures of magic~\cite{haug_scalable_2023,haug_efficient_2023}. 
Here, we describe a simple, new protocol that, from the Bell samples, allows us to efficiently estimate the stabilizer nullity, a magic monotone~\cite{beverland_lower_2020}, and learn states that can be prepared by quantum circuits with $t \in O(\log n)$ $T$-gates.

Our learning algorithm proceeds in two steps.
In the first step, we find a compression of the non-Clifford part of the circuit, similarly to Refs.~\cite{arunachalam_parameterized_2022,leone_learning_2023-1}. 
To achieve this, using Bell difference sampling \cite{gross_schurweyl_2021}, we find a Clifford unitary $U_{C}$ corresponding to a subspace $C \subset \mb F_2^{2n}$ such that $U_{C} \ket \psi$ has high fidelity with $\ket x \ket \varphi$ for some computational-basis state $\ket x$  on the first $\dim(C)$ qubits, and a state $\ket \varphi$ on the remaining qubits containing the non-Clifford information. 
The dimension of $C$ satisfies $\dim(C) \geq n- t$.
The number of $T$-gates~$t$ required to prepare $\ket \psi$ is therefore lower-bounded by the stabilizer nullity $M(\ket \psi) \coloneqq n- \dim(C)$, which is a magic monotone \cite{beverland_lower_2020}. 
We show that only $O(n/\epsilon)$ Bell samples are sufficient to ensure that $\ket \psi$ is $\epsilon$-close to a state with exact stabilizer nullity given by the estimate $\hat M$ of $M(\ket \psi)$. 
To the best of our knowledge this is the most efficient way of measuring the magic of a quantum state to date.

In the second step of the learning algorithm, we characterize the state $\ket \varphi$ on the remaining $n - \dim(C) \leq t $ qubits using pure-state tomography, for example via the scheme of Ref.~\cite{guta_fast_2018}, giving an estimate $\ket {\hat \varphi}$. 
The output of the algorithm is a classical description of $\ket {\hat \psi} = U_{C} \ket x \ket {\hat \varphi}$.
The learning algorithm runs in polynomial time and succeeds with high probability in learning an $\epsilon$-approximation to $\ket \psi$ in fidelity using $O(n /\epsilon)$ Bell samples and $O(2^{t}/\epsilon^2)$ measurements to perform tomography of $\ket x \ket {\hat \varphi}$. 

Using Clifford+$T$ simulators \cite[e.g.][]{bravyi_simulation_2019-1,pashayan_fast_2022,campbell_unified_2017} we can now produce samples from and compute outcome probabilities of  $\ket {\hat \psi}$ in time $O(2^t)$. 
We note that the exponential scaling in $t$ is asymptotically optimal since the description of a state with stabilizer nullity $t$ has $2^t + n-t$ real  parameters. 
Our algorithm generalizes to arbitrary non-Clifford gates.

To summarize, we have given efficient ways to extract properties of the circuit $C$---its depth and an efficient circuit description for circuits with low $T$-count---using only a small number of Bell samples.
Further properties of $\ket C$ that can be efficiently extracted from the Bell samples include the expectation values of any diagonal two-copy observables $A = \sum_{r} a_{r} \proj {\sigma_{r}}$ and different measures of magic~\cite{haug_scalable_2023}.
The Bell samples thus serve as an efficient classical shadow of $C$.


\paragraph*{Error Detection and Correction}
In the last part of this paper, we discuss another appealing feature of Bell samples: we can perform error detection and correction.
The idea that redundantly encoding quantum information in many copies of a quantum state allows error detection goes back to the early days of quantum computing. 
Already in 1996, \textcite{barenco_stabilisation_1996} have shown that errors can be reduced by symmetrizing many copies of a noisy quantum state.
More recently Refs.~\cite{cotler_quantum_2019,koczor_exponential_2021,huggins_virtual_2021} used measurements on multiple copies to suppress errors in expectation value estimation. 
In our two-copy setting, some simple \emph{single-sample} error detection properties follow immediately from the tests in the previous section. 

First, we observe that an outcome in the antisymmetric subspace, i.e., an outcome $r$ with $\pi_Y(r) = 1$, is certainly due to an error.
We can thus reduce the error in the sampled distribution by discarding such outcomes. 
We show in the SM \cite{suppmat} that such error detection reduces the error rate of a white-noise model by approximately a factor of $2$. 
Quantum computations in the Bell sampling model with error detection can thus achieve comparable fidelities to circuit model computations, where no error detection is possible, in spite of the factor of 2 in qubit overhead.  

Second, we note that Bell samples are compatible with stabilizer codes. 
For such codes, the Bell measurement between code blocks is a transversal measurement, and allows to extract the syndrome $\sigma \otimes \sigma$ for $\sigma \in \mathsf P_n$ in the stabilizer of the code~\cite{gottesman_introduction_2009}. 
If a detectable/correctable error occurred in one of the code blocks, this syndrome detects/identifies that error up to stabilizer equivalence.
The fact that the Bell measurement is transversal implies that an error in the Bell measurement does not spread, so that local error channels or coherent errors in the entangling \texttt{cnot} gates in the Bell measurement reduce the overall measurement fidelity by $(1-\epsilon)^{n}$, where $\epsilon$ is the error rate per Bell pair. 
Bell sampling is thus accurate in the regime of $\epsilon n \ll 1 $. 
We also note that antisymmetric errors in the Bell measurement are  detectable. 

Finally, we observe that quadratic error suppression is possible for estimating the expectation values of diagonal two-copy observable $A$, through the estimate $\tr[A \rho^{\otimes 2}]/\tr[\mb S\rho^{\otimes2}]$,  similar to virtual distillation \cite{cotler_quantum_2019,huggins_virtual_2021,koczor_exponential_2021}.
Specifically, this is true for estimating the expectation $\bra {E_0} P \ket {E_0}^2$ of a Pauli observable $P$ in the ground state $\ket {E_0} $ of a gapped Hamiltonian by choosing $A = (P \otimes P ) \mb S$, see the SM \cite{suppmat} for details.


\paragraph*{Discussion and outlook. }
In this work, we have proposed and studied Bell sampling as a model of quantum computation. 
We have shown that many properties of the quantum circuit preparing the underlying state can be extracted efficiently, and that in particular certain errors in the state preparation can be detected from single shots.  
Based on this, we have argued that the Bell samples are classical circuit shadows.
Since Bell sampling is universal this allows us to perform universal quantum computations whose outputs also yield information about the quantum circuit. 
This makes Bell sampling an interesting computational model, and our main focus in this work is to establish this fact. 

Bell sampling is not only interesting conceptually, however.
It is also realistic. Since the Bell basis measurement requires only transversal \texttt{cnot} and single-qubit gates, it can be naturally implemented in unit depth on various quantum processor architectures with long-range connectivity. 
These include in particular ion traps \cite{bermudez_assessing_2017} and Rydberg atoms in optical tweezers \cite{bluvstein_quantum_2022}. 
It is more challenging to implement Bell sampling in geometrically local architectures such as superconducting qubits \cite{arute_quantum_2019}. 
In such architectures, one can interleave the two copies in a geometrically local manner such that the Bell measurement is a local circuit; however, this comes at the cost of additional layers of SWAP gates for every unit of circuit depth. Alternatively, one can use looped pipeline architectures to implement the Bell measurement \cite{Cai2022}.

But is Bell sampling also practical in the near term? Initial experimental results indicate that it is practical for logical qubit architectures \cite{bluvstein_logical_2023}, but to more fully answer this question various sources of noise need to be analyzed in more detail. 
For the fidelity estimation protocols we have analyzed the setting of Pauli noise relevant, e.g., to circuits implemented with independent randomized compiling on the two copies or to large-scale fault-tolerant circuits with repeated rounds of error correction. 
It remains an important open question whether we can develop noise-robust versions of the depth and magic tests.
While we have exploited the purity of the state $\ket C$ in our error detection protocol, it is in an interesting question whether it is possible to detect additional errors from the Bell samples. 
For instance, it might be possible to exploit the fact that the subsystem purity of the target state is low for large subsystems, see \cref{fig:page curves}. 

We have shown that classically simulating the Bell sampling protocol with universal random circuits is classically intractable. 
An exciting question in this context is whether the complexity of Bell sampling might be more noise robust than computational-basis sampling in the asymptotic scenario.
For universal circuit sampling in the computational basis \textcite{gao_efficient_2018,aharonov_polynomial-time_2022} developed an algorithm that simulates sufficiently deep random circuits with a constant noise rate in polynomial time. 
In the Supplementary Material~\cite{suppmat} we give some initial evidence that this simulation algorithm fails for Bell measurements.
If the hardness of Bell sampling indeed turns out to be robust to large amounts of circuit noise, we face the exciting prospect of a scalable quantum advantage demonstration with classical validation and error mitigation.

\smallskip

\emph{Note:} While finalizing this work, we became aware of Refs.~\cite{grewal_efficient_2023-2,leone_learning_2023}, where the authors independently report algorithms similar to the one we present above for learning quantum states generated by circuits with low $T$-count.
After this work was completed, we collaborated on the physical implementation of Bell sampling in a logical qubit processor, illustrating the feasibility of our results to near-term devices \cite{bluvstein_logical_2023}. 

\paragraph*{Acknowledgements}
D.H. warmly thanks Abhinav Deshpande and Ingo Roth for helpful discussions that aided in the proofs of \cref{thm:sharpp hardness probs} and \cref{lem:subspace}, respectively.  
We are also grateful to Dolev Bluvstein, Maddie Cain, Bill Fefferman, Xun Gao, Soumik Ghosh, Alexey Gorshkov, Vojt\v ech Havl\'i\v cek, Markus Heinrich, Marcel Hinsche, Marios Ioannou, Marcin Kalinowski, Mikhail Lukin and Brayden Ware for discussions. This research was supported in part by NSF QLCI grant
OMA-2120757 and Grant No.\ NSF PHY-1748958 through the KITP program on ``Quantum Many-Body Dynamics and Noisy Intermediate-Scale Quantum Systems.''
D.H.\ acknowledges funding from the US Department of Defense through a QuICS Hartree fellowship.

\let\oldaddcontentsline\addcontentsline
\renewcommand{\addcontentsline}[3]{}
\putbib
\let\addcontentsline\oldaddcontentsline%
\end{bibunit}

\onecolumngrid
\cleardoublepage
\setcounter{section}{0}
\setcounter{page}{1}
\setcounter{equation}{0}
\thispagestyle{empty}
\begin{center}
\textbf{\large Supplementary Material for ``Bell sampling from quantum circuits''}\\
\vspace{2ex}
Dominik Hangleiter and Michael J. Gullans
\vspace{2ex}
\end{center}

\begin{bibunit}

\twocolumngrid
\renewcommand\thesection{S\arabic{section}}

\pagenumbering{roman}
\renewcommand\thefigure{S\arabic{figure}}
\renewcommand\theequation{S\arabic{equation}}

\tableofcontents
\section{\bqp-completeness of Bell sampling}
In this section, we show that Bell sampling is \bqp-complete as a model of quantum computation in spite of the fact that we are restricting to circuits of the form $C \otimes C$ and measurements in the transversal Bell basis.

\begin{lemma}[\bqp-completeness]
  Bell sampling is \bqp-complete.
\end{lemma}
\begin{proof}
Consider an arbitrary quantum circuit $C$. Then estimating the probability of measuring the first qubit of $C$ in the $\ket1$ state up to additive precision is \bqp-complete.
The idea of the proof is to design a circuit $C'$ such that we can infer $p_1 \coloneqq \tr[(\proj{1} \otimes \id_{n-1}) C \proj {0^n}C^\dagger]$ from Bell samples from $C'$.  
Let $\rho = \tr_{[n]\setminus \{1\}}[C \proj {0^n}C^\dagger] $ be the reduced density matrix of the first qubit of the state $C \ket 0 $ before the measurement.
Then our task is to estimate $p_1 = \tr [\rho \proj 1 ]$ from Bell sampling. 

To achieve this, we make use of the fact that we can infer the square of any Pauli expectation value $\bra C P \ket C^2$ up to additive precision from Bell samples from $\ket C \otimes \ket C$ for any circuit $C$. 
Starting from an arbitrary quantum circuit $C$ we add an ancillary qubit (labeled by $0$) before the first qubit and run the circuit $C' = \ee^{- \ii \frac \pi 8  (Z_0 Z_1 + Z_0) } (H \otimes C )$, where $H$ denotes the Hadamard gate.
The resulting marginal state on the ancillary qubit  is then given by 
\begin{multline}
 \tr_{[n]}[\ee^{- \ii \frac \pi 8  (Z_0 Z_1 + Z_0) }(\proj + \otimes \rho )\ee^{\ii \frac \pi 8  (Z_0 Z_1 + Z_0)} \\= \frac 1 2 \left( \id + \left(\frac 12 - \frac 12\bra C Z_1 \ket C \right) X + \left( \frac 12 + \frac 12\bra C Z_1 \ket C  \right) Y\right). 
\end{multline}
From Bell samples from $\ket{C'} \otimes \ket{C'}$ we can now estimate $|\bra {C'} X_0 \ket {C'}| = |1 - \bra C Z_1 \ket C|/2 = p_1$.
\end{proof}

\section{Classical hardness of Bell sampling}

In this section, we provide complexity-theoretic evidence that sampling from the Bell distribution $P_C$ up to constant total-variation distance (TVD) error is classically intractable. 

In order to show this, we follow a standard proof strategy, which has three main ingredients. 
\begin{enumerate}
  \item \textbf{Hiding}:  The distributions over outcomes and circuit instances are interchangeable.
  \item \textbf{Average-case \gapp\ hardness} of approximating the output probabilities. While we cannot prove this in any instance, we typically provide evidence for it using three ingredients:
  \begin{enumerate}
    \item  Worst-case \gapp\ hardness of approximation,
    \item  Average-case \gapp\ hardness of near-exact computation, and
    \item Anticoncentration.
  \end{enumerate} 
\end{enumerate}
We defer the reader to Section III of Ref.~\cite{hangleiter_computational_2023-1} for a detailed exposition of the proof strategy.

\subsection{Hiding}

Consider the Bell sampling distribution 
\begin{align}
  P_C(r) & = \frac{1}{2^n}|\bra{ \sigma_r }C \otimes C \ket {0^{2n}}|^2\\
  & = \frac{1}{2^n}| \bra {\Phi^+} (\sigma_r \otimes \id) (C \otimes C) \ket {0^{2n}}|^2.
\end{align}
Then, using that $\ket{\Phi^+} = \vecz (\id)/{\sqrt 2}$, we find that 
\begin{align}
   (\sigma_r \otimes \id) \ket {\Phi^+} = (\sqrt{\sigma_r} \otimes \sqrt{\sigma_r}^T)\ket {\Phi^+},
\end{align}
and hence
\begin{align}
  P_{C}(r) = | \bra {\Phi^+} \sqrt{\sigma_r}C  \otimes \sqrt{\sigma_r}^T C \ket {0^{2n}}|^2, 
\end{align}
so that for $x$ such that $\sqrt{\sigma_r} \neq \sqrt{\sigma_r}^T$, we cannot write $P_C(r) = P_C'(0^n)$ for some $C' = C'(C,r)$.
And indeed, while $\sqrt X ^T = \sqrt X $ and $\sqrt Z ^T = \sqrt Z$, we have $\sqrt Y^T \neq \sqrt Y$.

This means that hiding does not hold on the full output distribution. 
However, we can restrict ourselves to the \textit{a priori} known support of the Bell sampling distribution, which is given by the symmetric subspace, characterized by $\pi_Y(r) =0$. 
Indeed, we can view the fact that $\sqrt Y^T \neq \sqrt Y$ as a signature of the fact that the corresponding Bell state is not symmetric.
Conversely, consider an even number of $Y$-outcomes. 
Specifically, for a two-qubit state with two $Y$-outcomes, we can explicitly check that $
\sqrt{Y \otimes Y}^T = \sqrt{Y \otimes Y}$. 

Hence, hiding holds on the symmetric subspace for any architecture that  includes a layer of single-qubit gates that is invariant under $\sqrt{X}, \sqrt{Y}, \sqrt{Z}$ at the end of the circuit. 

Furthermore, defining $C_r = \sqrt {\sigma_r} C $ we observe that 
\begin{align}
  |\bra {\sigma_r} C_s \ket {0^{2n}}|^2 = |\bra {\sigma_{r \oplus s}} C \ket {0^{2n}} |^2,
\end{align}
since $\sqrt{\sigma_r} \sigma_s \sqrt{\sigma_r} \in \sf P_{n}$, where $\sf P_n$ is the Pauli group with phases in $\{\pm 1, \pm \ii\}$. 

\subsection{Average-case \gapp\ hardness of approximating outcome probabilities}

\subsubsection{Worst-case \gapp\ hardness of approximating outcome probabilities}
\begin{lemma}
\label{thm:sharpp hardness probs}
  Given a quantum circuit $C$, it is \gapp-hard to compute $2^{-n}|\bra C\sigma_r \ket{\overline C}|^2$ up to relative error $<1/2$.
\end{lemma}
\begin{proof}[Proof]
In order to perform the Bell basis measurement we apply $\texttt{cnot}(H \otimes \id)$ across all pairs of qubits. 
Writing $ \ket C = \sum_x c_x \ket x$, the pre-measurement state then transforms as 
\begin{align}
  \ket C \otimes \ket C =&   \sum_{xy} c_x c_y \ket x \ket{y} \\
   \xrightarrow{\texttt{cnot}^{\otimes n}} & \sum_{xy} c_x c_y \ket x \ket{x \oplus y}\\
    \xrightarrow{(H \otimes \id)^{\otimes n}} & \sum_{xyz} (-1)^{x\cdot z} c_x c_y \ket z \ket{x \oplus y} \eqqcolon \ket{\mc C}. 
\end{align}

Consider a state $\ket C$ on $n+1$ qubits and the $(0^n 1,0^{n+1})$-Bell amplitude 
\begin{align}
\label{eq:10 amplitude}
 \bra {0^n 1} \bra{0^{n+1}} \ket {\mc C } &= \sum_{xyz} (-1)^{x \cdot z} c_x c_y \braket {0^{n}1}{z}\braket {0^{n+1}}{x \oplus y}\\ & = \sum_x(-1)^{x \cdot 0^{n}1} c_x^2 \\
 & = \sum_{x: x_n = 0} c_x^2 - \sum_{x: x_n = 1} c_x^2 .
\end{align}
Let us now specify $\ket C \coloneqq \ket{C_f} \propto \sum_x \ket x \ket{f(x)}$ up to normalization, 
where $f: \bin^{n} \rightarrow \bin$ is an efficiently computable Boolean function. 
Observe that for $b \in \bin$
\begin{align}
  (\id \otimes \bra b) \sum_x \ket x \ket{f(x)} = \sum_{x: f(x) = b} \ket x. 
  \label{eq:stateb}
\end{align}
For $b=1$ the state \eqref{eq:stateb} is normalized by the square root of the number of accepting inputs of $f$, given by $ \#f \coloneqq | \acc(f)| \equiv |\{x: f(x)=1\}|$, and for $b=0$ by $\sqrt{2^n - \#f}$ . 
Hence, the coefficients of $\ket {C_f} = \sum_{x \in \bin ^{n+1}} c_x \ket x$ are given by 
\begin{align}
  c_{yb} \equiv (\bra y \otimes \bra b) \ket {C_f} =  \begin{cases} 1/\sqrt{\#f},  & f(y)=b = 1\\
   1/\sqrt{2^n - \#f},  & f(y)=b= 0\\
  0  & \text{otherwise}. \end{cases}
\end{align} 

Let us now add another qubit so that we have $n+2$ qubits in total. 
Given an arbitrary efficiently computable Boolean function $g: \{0,1\}^n \rightarrow \{0,1\}$ let us define
\begin{align}
  f_g &: \{0,1\}^{n+1} \rightarrow \{0,1\}\\
  & f_g(y) = \begin{cases}
    1 & \text{ if } y = (x, g(x))\\
    0 & \text{ if } y = (x, \neg g(x)). 
  \end{cases}
\end{align}
We now reversibly compute the function $f_g$ obtaining $n+2$-bit outcome strings $(y,f_g(y))$ with the property that $f_g(y) = 1$ if $y = (x,g(x))$. So the $(n+1)$st qubit encodes the outcome of $g$ while the last (i.e., $(n+2)$nd) qubit encodes the outcome of~$f_g$. 

$f_g$ can be efficiently computed: Let $D$ be the circuit that maps $\ket x \ket b \rightarrow \ket x \ket {g(x) \oplus b}$. Then the quantum circuit $C = (\id_n \otimes \texttt{cnot} (X \otimes \id)) (D^\dagger \otimes \id)$ computes $\ket y \ket 0 \mapsto \ket y \ket {f_g(y)}$. 

We observe that
\begin{align}
\label{eq:equality fg}
  | \acc(f_g)| =|\rej(f_g)| = 2^{n+1}/2
\end{align}
and, moreover, we have\footnote{We start labelling indices at $0$ and hence $ x_n$ is the $n+1$st bit of $x$.}
\begin{align}
\label{eq:equality g}
  |\acc(g)| = |\{ y \in \acc(f_g), y_{n} = 1 \}| &  =
  |\{ y \in \rej(f_g), y_{n} = 0 \}|, \\
  |\rej(g)| = |\{ y \in \acc(f_g), y_{n} = 0 \}| &  =
  |\{ y \in \rej(f_g), y_{n} = 1 \}|.
\end{align}
To see the latter, just observe that 
\begin{multline}
   |\{ y \in \acc(f_g), y_{n} = 1 \}| = 
  |\{ (x, g(x)): g(x) = 1 \}|\\
  =  |\{ (x, \neg g(x)): \neg g(x) = 0 \}| 
   = |\acc(g)|.
\end{multline}

Let us consider the outcome string $(0^n11, 0^{n+2})$ and compute the corresponding amplitude of the state $\ket{\mc C_{f_g}} = ((H \otimes \id) \cdot \texttt{cnot})^{\otimes (n+2)} \ket C \otimes \ket C$ as 
\begin{widetext}
\begin{align}
\label{eq:11 amplitude}
 \bra {0^n 11,0^{n+2}} \ket {\mc C_{f_g}}&= \sum_{xyz} (-1)^{x \cdot z} c_x c_y \braket {  0^{n} 11}{z}\braket {0^{n+2}}{x \oplus y}\\ & = \sum_x(-1)^{x \cdot (0^{n}11)} c_x^2 \\
 & = \sum_{x: x_n = 0, x_{n+1} = 0} c_x^2 
 - \sum_{x: x_n = 1, x_{n+1} = 0} c_x^2  
 - \sum_{x: x_n = 0, x_{n+1} = 1} c_x^2 
 + \sum_{x: x_n = 1, x_{n+1} = 1} c_x^2\\
  \stackrel{\eqref{eq:equality fg}}{=}& \frac 1 {2^{n}}\left(|\{y \in \rej(f_g): y_n = 0 \}|- |\{y \in \rej(f_g): y_n = 1 \}|\right)\\
  & \quad + \frac 1 {2^{n}} \left(|\{y \in \acc(f_g): y_n = 1 \}| -|\{y \in \acc(f_g): y_n = 0 \}| \right)  \\ 
\stackrel{\eqref{eq:equality g}}{=} &\frac{2}{2^{n}} \left(|\acc(g)| - |\rej(g)| \right)  \equiv \frac 1 {2^{n-1}} \gap(g)
\end{align}
\end{widetext}
This shows that the output amplitudes of Bell sampling from universal quantum circuits can encode the gap of any \sharpP-function. 
Finally, we reduce the outcome to the all-zero outcome---and by the hiding property any outcome in the symmetric subspace---by observing that the output string corresponds to the $\ket \id^{\otimes n} \ket Z^{\otimes 2} $ outcome and hence we can define 
$\tilde C_{f_g} = (\id^{\otimes n} \otimes (Z^{1/2})^{ \otimes 2}) C_{f_g}$ to show that $\bra {0^{2(n+2)}}\ket{\mc C_{f_g}}$ is \gapp-hard to compute.

By Proposition 8 of \textcite{bremner_average-case_2016} (see also Lemma~8 of Ref.~\cite{hangleiter_computational_2023-1}), approximating $|\bra {0^{2n}} \ket {\tilde {\mc C}_{f_g}}|^2$ up to any relative error $< 1/2$ or additive error $1/2^{2(n+1)}$ is \gapp-hard. 
\end{proof}

Notice that this argument also proves that computing certain Pauli coefficients of an $n$-qubit quantum circuit is \gapp\ hard up to relative error $<1/2$ since the circuit $C_{f_g}$ we have used in the encoding of the gap of a \sharpP\ function is real and $|\bra C \sigma_r \ket {\overline C}|^2 = |\tr[\proj C \sigma_r]|^2$ for a real circuit $C$. 
Notice, however, that we cannot reduce to the all-zero string in this case. 
This is also easy to see since the probability of the all-zero outcome of a real circuit is always one. 

Conversely, if we include the $\sqrt Z \otimes \sqrt Z $ gate at the end of $C_{f_g}$, this shows that computing the overlap $|\bra C \ket {\overline C}|^2$ is \gapp-hard in general.
\medskip 

The next step in applying the Stockmeyer argument is to see if approximate average-case hardness is plausible. To this end we can/need to make two arguments. 

\subsubsection{Near-exact average-case hardness}

First, we need to show (near-)exact average-case hardness, see \cite[Sec.~IV.D.5]{hangleiter_computational_2023-1} for the available techniques. 

Consider a circuit $C$ with some Haar-random 2-qubit gates. 
Let us follow the strategy by \textcite{krovi_average-case_2022}. 
Analogously, we interpolate every 2-qubit gate $G_i$ in the worst-case circuit $C$ to a Haar random gate $H_i G_i$ with $H_i$ drawn from the 2-qubit Haar measure. 
\begin{align}
G_i(\theta)& = \exp(\ii (1-\frac\theta{2m}) \log H_i) G_i\\
& = V_i^\dagger \sum_{k_i=1}^4 \ee^{\ii ( 1-\frac\theta{2m})\phi_{k_i}} \proj{\psi_{k_i}}V_i G_i\\
& \eqqcolon \sum_{k_j}  \ee^{\ii ( 1- \frac\theta{2m})\phi_{k_i}} \tilde G_{k_j}
\end{align}
where $V_i$ diagonalizes $H_i$ into eigenvectors $\ket {\psi_{k_i}}$ and eigenvalues $\exp(\ii \phi_{k_i})$. 
Notice that compared to \textcite{krovi_average-case_2022}, we interpolate only by an angle $\theta/2m$ instead of $\theta/m$.
Now, we consider the circuit $C(\theta)$ defined by replacing all the gates $G_i$ of $C$ with $G_i(\theta)$. 
We can write the output probability as 
\begin{multline}
  |\bra {0^{2n}} \ket {\mc C(\theta)}|^2 = | \bra {\Phi^+}^{\otimes n} C(\theta ) \otimes C(\theta) \ket {0^{2n}}\\
  = \left|\sum_{k_{ij}: i \in [m], j \in \bin} \ee^{\ii (1- \frac\theta{2m}) \sum_{ij} \phi_{k_{ij}}}
  \bra {\Phi^+}^{\otimes n} 
 \tilde G_{k_{00}} \cdots G_{k_{d1}}
  \ket {0^{2n}} \right|^2 \\
   =  \sum_{k,k'} \ee^{\ii (1- \frac\theta{2m}) \Delta \phi_{k,k'}} \bra {\Phi^+} \tilde G_k \ket 0 \bra 0 \tilde G_k^\dagger \ket {\Phi^+},
\end{multline}
where $k = (k_{00}, k_{01}, \ldots, k_{m1})$, $\tilde G_k = \prod_{ij} G_{k_{ij}}$ and $\Delta \phi_{k,k'} = \sum_{ij} (\phi_{k_{ij}} - \phi_{k'_{ij}})$. 
Since $|\Delta \Phi_{k,k'}| /2m \in O(1)$, we can now follow the argument of Krovi, replacing $m \leftarrow 2m$, to construct a polynomial of degree $d \in O(m / \log(m)^2)$ that approximates the output probabilities of Bell sampling. 
Using this polynomial, we can run robust polynomial interpolation to obtain near-exact average-case hardness. 

\subsubsection{Anticoncentration}
Finally, the question arises whether the output distribution anticoncentrates. Recall  that anticoncentration is defined as 
\begin{align}
  \Pr_{C \sim \mc C} \left[ P_C(S) \geq \frac 1 {| \Omega|} \right] \geq \gamma, 
\end{align}
for constant $\gamma$ and the sample space $\Omega$. 
Our standard tool for showing anticoncentration is the Payley-Zygmund inequality which lower bounds the anticoncentrating fraction in terms of second moments as 
\begin{align}
   \Pr_{C \sim \mc C} \left[Z  \geq \alpha \mb E[Z]\right] \geq (1- \alpha)^2 \frac { \mb E[Z]^2}{\mb E[Z^2]}, 
\end{align}
for a random variable $Z$ with $ 0 \leq Z \leq 1$ which we take to be $Z = P_C(S)$. 

This problem gets interesting for the Bell sampling circuits since already a single copy of the Bell sampling state has two copies of the circuit, so we will need second and fourth moments of the circuit family to show anticoncentration with the standard technique. 
Recall that the Bell sampling output distribution is given by 
\begin{align}
  P_C(r) = |\bra{ \sigma_{x} }C \otimes C \ket {0^{2n}}|^2, 
\end{align}
where $\sigma_r = \prod_{i=1}^N \sigma_{x_i}$.
Let's assume that we have a four-design and begin with the first moment. 
\begin{align}
  \mb E[P_C(r)] & = \bra{\sigma_r}  \mb E \left[C \otimes C \ket {0^{2n}} \bra{0^{2n}} C^\dagger \otimes C ^\dagger  \right] \ket{\sigma_r}\\
  & =\frac {1}{D_{[2]}} \bra{\sigma_r} P_{[2]}\ket{\sigma_r},
\end{align}
where $D_{[t]} = \binom{d + t -1}{t}$ and $P_{[t]}$ is the projector onto the symmetric subspace of $t$ copies (see \cite{zhu_clifford_2016} for a nice intro). 
We can explicitly compute the expression as 
\begin{align}
\bra{\sigma_r}& P_{[2]}\ket{\sigma_r}  = \frac 12 \bra{\sigma_r} (\id + \mb S) \ket{\sigma_r} \\
& = \bra {\Phi^+} (\sigma_r^2 \otimes \id) \ket{\Phi^+} + \bra {\Phi^+} (\sigma_r \sigma_r^T \otimes \id) \ket{\Phi^+}\\
& = 2^{n} (1 + (-1)^ {\pi_Y(r)}),  
\end{align}
where we in abuse of notation we write $\ket {\Phi^+} = \ket{\Phi^+}^{\otimes n} $ and observe that $\tr[XX^T] =\tr[ZZ^T] = \tr[\id]= 2^n  $ and $\tr[YY^T]=-2^n$. 
We find---as expected---that the Bell state with an odd number of singlet states (coresponding to $\sigma_r= Y$) is in the antisymmetric subspace and therefore the projector onto the symmetric subspace evaluates to zero in that case: 
\begin{align}
  \mb E [P_C(r)] = \begin{cases}
    D_{[2]}^{-1} = 2/(2^{2n}(1+ 2^{-n})) & \text{if } \pi_Y(r) \text{ even}\\
    0 & \text{if } \pi_Y(r) \text{ odd}
  \end{cases}.
\end{align}
The output distribution is thus supported on the even $Y$-parity sector on which all outcomes have equal expectation value given by $D_{[2]}^{-1} = 2/(2^n(2^n+1))$. 
The size of the even-$Y$-parity sector should be exactly given by this number since these strings correspond to a basis of the symmetric subspace.

The second moment is more complicated. 
It reads
\begin{multline}
  \mb E [\bra{\sigma_r}^{\otimes 2} C^{\otimes 4} \proj 0( C^{\dagger})^{\otimes 4} \ket {\sigma^r}^{\otimes2}] \\
  =   \frac 1 {D_{[4]}} \bra{\sigma_r}^{\otimes 2} P_{[4]}\ket {\sigma_r}^{\otimes2}
\end{multline}
To compute this overlap on the symmetric subspace on which $\sigma_r = \sigma_r^T$, we write $P_{[f]} = 4!^{-1} \sum_{\sigma \in S_4} P_\sigma $, where $P_\sigma$ is the permutation matrix corresponding to the element $\sigma$ of the symmetric group $S_4$. 
We observe that for a $c=1/3$-fraction of the permutations in the definition of $P_{[4]}$ the overlap evaluates to $\tr[\sigma_r^2]^2/2^{2n} = \tr[\id]^2/2^{2n} = 1$, while for the other $1-c$ fraction, it evaluates to $\tr[\sigma_r^4]/2^{2n} = \tr[\id]/2^{2n} = 1/2^n$. 
Viewing the Bell state $\ket {\Phi^+}$ as a vectorization of the identity matrix, these correspond exactly to the cases in which there are one versus two connected components in the resulting graph.

Consequently, we find 
\begin{multline}
  \mb E [\bra{\sigma_r}^{\otimes 2} C^{\otimes 4} \proj 0( C^{\dagger})^{\otimes 4} \ket {\sigma^r}^{\otimes2}] \\= \frac 1 { D_{[4]}} \left( c + \frac {1-c}{2^n}\right).
\end{multline}
for all even $Y$-parity $r$ (satisfying that $\pi_Y(r)$ is even). 

From this we find that if $\mc C$ generates a state $4$-design, for all $r$ in the symmetric subspace 
\begin{align}
  \Pr[P_C(r)  \geq \frac \alpha {D_{[2]}}] &\geq (1-\alpha)^2  \frac{1}{c + \frac {1-c} {2^{n}}}\frac {D_{[4]}}{D_{[2]}^2 } \\
  &= \frac{(1-\alpha)^2} {3! \cdot (c + \frac {1-c} {2^{n}})} \frac{(2^n + 3)(2^n + 2)}{(2^n+1)2^n}\\
  & \geq \frac{(1-\alpha)^2} {6c}
\end{align}

By the result of \textcite{brandao_local_2016}, this means that the Bell sampling distribution anticoncentrates in linear depth. 
Given the fact that we just copied two random circuits and qualitatively found the same results as for a single copy, we  conjecture that the Bell sampling distribution also anticoncentrates in $\log$-depth. 
To check this, we would have to directly compute the moments of the output distribution, but this becomes more complicated than for the single-copy case. 
This is because the local degrees of freedom the standard mapping to a statistical-mechanical model \cite{zhou_emergent_2019,hunter-jones_unitary_2019,barak_spoofing_2021,dalzell_random_2022}, increase from 2 to 24 permutations.

\section{Learning circuit properties from Bell sampling}

In the following, we will detail the tests that can be performed \emph{using just the samples} from the Bell distribution.
Success on those tests---while falling short of loophole-free verification---increases our confidence in the correctness of the experiment.

\subsection{Measuring purity and fidelity}

Our first observation is that, given some noisy state preparation $ \rho \otimes \sigma$, we can estimate $\tr [\rho \sigma]$ from Bell measurements on $\rho\otimes \sigma$. 
To see this, observe that we can express the single-qubit swap operator in the Bell basis as 

\begin{align}
  \mb S &= \proj{00} + \proj{11} + \ket{01}\bra{10} + \ket{10}\bra{01}\\
  & = \underbrace{\proj{\sigma_{00}} + \proj{\sigma_{01}} + \proj{\sigma_{10}}}_{P_{\bigvee^2}} - \underbrace{\proj{\sigma_{11}}}_{P_{\bigwedge^2}}
\end{align} 
Hence, we can estimate the overlap 
\begin{align}
 O(\rho, \sigma) = \tr[\rho \sigma] = \tr[(\rho \otimes \sigma) \mb S] 
\end{align}
by taking the difference between the frequency of outcomes with even parity of $11$-outcomes and odd parity of $11$-outcomes
\begin{align}
  \hat O = \frac 1 M \left(|\{r: \pi_Y(r) = 0 \}| - |\{r: \pi_Y(r) = 1\}|\right), 
\end{align}
where $M$ is the total number of measurements.
In case $\rho = \sigma$, 
\begin{align}
  O(\rho, \rho)  = \tr[\rho^2] = P(\rho)
\end{align}
is just the purity of $\rho$.

\subsubsection{Estimating the fidelity using the purity}

In this section, we will detail how an estimate of the purity can be used in order to estimate the fidelity 
\begin{align}
  F(\rho, \proj C ) = \bra C \rho \ket C,
\end{align}
of a state preparation $\rho $ compared to a target state $\ket C$. 

We model noisy state preparation $\rho_C$ of the state $\proj C$ prepared by a circuit $C$ as circuit level noise after each gate in $C$. 
Using randomized compiling implemented independently on each copy of the Bell sampling circuit, the effective noise channel for the Bell samples can then be approximately reduced to a Pauli channel \cite{wallman_noise_2016}.  
Strictly speaking, randomized compiling works well in the experimentally relevant setting in which there are some `easy' gates, say, Pauli gates, on which the noise is approximately gate-independent and some `hard' gates which have gate-dependent errors.
In that setting, the gate-dependent noise on the hard gates, except for those in the last layer of the circuit, can be reduced to a Pauli channel, and the pre-Bell-measurement state can be written as $\rho_C \otimes \rho_C$.

Let us begin by analyzing rigorously the estimation of the average fidelity 
\begin{align}
  \mb E_C F(\rho_C, \proj C)
\end{align}
over random quantum circuits $C$ from. 
Following an argument in Ref.~\cite{ware_sharp_2023}, consider circuits~$C$ constructed in such a way that every---potentially fixed---gate is followed by random single-qubit gates which form a unitary 2-design. 
For simplicity, consider iid. noise given by a single-qubit Pauli noise channel $\mc N =\sum_{i \in \{0,\ldots, 3\}} p_i \sigma_i \cdot \sigma_i $ acting on every qubit after the application of a two-qubit gate. 

\begin{lemma}[Translating fidelity and purity]
Let $\rho_C(p)$ with be the output state of a random noisy circuit $C$ with single-qubit two-design gates following every noisy gate, and Pauli noise probabilities $p = (p_0,\ldots, p_3)$. 
Then
\begin{align}
  \mb E_C \rho_C(p)^2  = \mb E_C F(\rho_C(q), \rho_C(0)),
\end{align}
where $q_0 = p_0^2 + \sum_{i \neq 0} p_i^2 $ and  for $k \neq 0$, $q_k = 2 p_k p_0 + \sum_{ij \neq 0 } |\epsilon_{ijk}| p_i p_j$ with the Levi-Civita symbol $\epsilon_{ijk}$. 
 \end{lemma} 
Before we prove the lemma, note that the iid. assumption is not necessary since the translation of noise rates holds for every noise location individually. 

\begin{proof}
Let $\mc U = |U \rrangle \llangle U^\dagger| \equiv \mathrm{vec}(U \cdot U^\dagger)$ be the matricization of the adjoint action of the two-qubit gate $U$, where we denote the vectorization of a matrix $A$ by $|A \rrangle$. 
The full circuit is then composed of two copies of noisy unitary two-qubit channels given by $\mc N^{\otimes 4}{\mc U^{\otimes 2}}$, where the cut is across the bipartition of the Bell measurement, and let the noisy state be $\rho(\epsilon)$. 

Then, we can express the fidelity between a state with noise rate $\epsilon$ and a state with noise rate $0$, as the purity of a state with noise rate $\epsilon'$.
To see this, we evaluate the average over gates in a single layer of the circuit. 
Then 
\begin{align}
  \mb E_U \mc N^{\otimes 4} \mc U^{\otimes 2} =  \sum_{\pi, \sigma} \mc N^{\otimes 4} | \pi \rrangle \llangle \sigma |,  
\end{align}
and we can evaluate
\begin{align}
   \mc N^{\otimes 2} | \id \rrangle & = \mc N' \otimes \id | \id \rrangle  = | \id \rrangle, \\
 \mc N^{\otimes 2} (\mb S )& = 
 \sum_{ijk \in \{0,\ldots, 3\}} p_i p_j (-1)^{s(k)} \sigma_i \otimes \sigma_j \ket {\sigma_k} \bra {\sigma_k}\sigma_i \otimes \sigma_j \\
 & = 
 \sum_{ij \in \{0,\ldots, 3\}} p_i p_j (\sigma_i \sigma_j \otimes \id) \mb S(\sigma_j \sigma_i \otimes \id)\\
 & = \sum_k q_k (\sigma_k \otimes \id)\mb S( \sigma_k \otimes \id)\\
 & \equiv (\mc N'\otimes \id) (\mb S).
 \end{align} 
Here, we have defined $s(k) = 1$ for $k = 2$ (corresponding to $\sigma_k = Y$) and $0$ otherwise. 
\end{proof}

Thus, we have moved noise from one copy of the circuit to the other copy of the circuit, and thereby related the quantities for fidelity (one noisy, one ideal copy) and purity (two noisy copies). 
For the depolarizing channel with depolarizing parameter $\epsilon$, we have $p_0 = 1-3\epsilon/4, p_1 = p_2 = p_3 = \epsilon/4$, so $q_0 = (1-3\epsilon/4)^2 + 3(\epsilon/4)^2$, $q_1 = q_2 = q_3 = \epsilon/2 + 2 (\epsilon/4)^2$.
Hence, in this case the type of noise channel even remains the same, and purity with local depolarizing strength $\epsilon$ simply corresponds to fidelity with a different local depolarizing rate given by $2\epsilon$: 
\begin{align}
\label{eq:moving epsilon}
  \mb E_C P(\rho(\epsilon)) = \mb E_C F(\rho_C(2\epsilon)), \proj C).
\end{align}
For more general types of Pauli noise, the noise channels $\mc N$ and $\mc N'$ will be different. 
The mapping from purity to fidelity carries over through appropriate redefinition of the noise rate.  

The average fidelity of random quantum circuits is up to small corrections given by the probability of no error in the circuit, that is, 
\begin{equation}
  \mb E_C F(\rho_C(\eta), \proj C) =  (1 - \eta)^E. 
  \label{eq:fidelity decay}
\end{equation}
An analytical argument that this behaviour is accurate irrespective of the noise rate is given in Ref.~\cite{ware_sharp_2023}, where it is shown that the average fidelity is approximately given by 
\begin{align}
  \mb E_C F \approx 2^{-n } + (1 - \eta)^{nd} + C \Lambda^d 2^{-n}, 
\end{align}
where $d$ is the circuit depth so that $E = nd$, and $\Lambda< 1, C> 0 $.
This expression shows that the decay of the fidelity with circuit volume continues until $\mb E_C F \approx 2^{-n}$, where it plateaus for arbitrary circuit depth. 

Since $(1-\eta)^E \approx \ee^{-\eta E}$, \cref{eq:moving epsilon,eq:fidelity decay} imply that 
\begin{align}
   \mb E_C P(\rho(\epsilon)) \approx [\mb E_C F(\rho_C(\epsilon), \proj C)]^2, 
\end{align}
and therefore that 
\begin{align}
   \overline F_{\mathrm{purity}} = \sqrt {\mb E_C\tr[\rho_C^2]}
   \label{eq:supp root purity}
\end{align}
is a good estimator of the average fidelity. 

We note that the decay behaviour according to \eqref{eq:fidelity decay} has been observed numerically to high precision for arbitrary noise rates~\cite{ware_sharp_2023,morvan_phase_2023}. 
Thus, we expect the root average purity to be an accurate average fidelity estimator for arbitrary noise rates. 
At the same time, the average XEB is known to fail for high noise rates, as well as in the regime of very low depth, before anticoncentration sets in \cite{ware_sharp_2023}. 

\subsubsection{Estimating the fidelity of fixed circuits}

The analytical argument above is valid for the average fidelity and purity of random circuits with circuit-level Pauli noise. 
However, we also expect the fidelity decay with the circuit volume to be true of non-random circuits, and the root purity estimator \eqref{eq:supp root purity} to be accurate in such cases as well. 
Na\"ively, we can understand this expectation as follows. 
If the local noise is given by Pauli noise channels---which can be enforced using randomized compiling---the effective density matrix on each copy of a fixed state preparation can be written as 
\begin{align}
\label{eq:randcomp}
  \rho_C(\eta) \coloneqq (1-\gamma )\proj {C} + \gamma \,\sigma_C ,
\end{align}
for another density matrix $\sigma$. 
This is because for a Pauli noise channel, there is a probability of no error occurring, which is given at least by $(1- \gamma) \geq (1- \eta)^E$. 
Generically, we expect the overlap $\bra C \sigma_C \ket C \sim 2^{-n}$ as well as the purity of $\sigma_C$ $\tr[\rho_C^2] \sim 2^{-n}$ to be exponentially small since after randomized compiling the errors in the circuit are incoherent Pauli errors and the deviation $\sigma$ from the ideal state $\proj C$ is generically uncorrelated with $\proj C$.
This implies that the purity of $\rho_C$ is given by 
\begin{align}
  P(\rho_C) &= (1-\gamma)^2 + 2 \gamma(1-\gamma) \bra C \sigma \ket C + \gamma^2 \tr[\sigma_C^2]\\
  & \sim (1-\gamma)^2 + \frac 1 {2^n} (2 \gamma (1- \gamma ) + 1),
\end{align}
and similarly for the fidelity 
\begin{align}
  F(\rho_C, \proj C) & = (1- \gamma)^2 + \gamma \bra C \rho_C \ket C\\
  & = (1- \gamma)^2 + \frac \gamma {2^n},
\end{align}
and therefore, up to exponentially small corrections 
\begin{align}
   F(\rho_C, \proj{\mc C}) \approx \sqrt{\tr[\rho_C(\eta)^2]}.
 \end{align} 
We also provide numerical evidence that the root purity is an accurate estimator of the fidelity for fixed circuits in the following.

\begin{figure*}
  \includegraphics{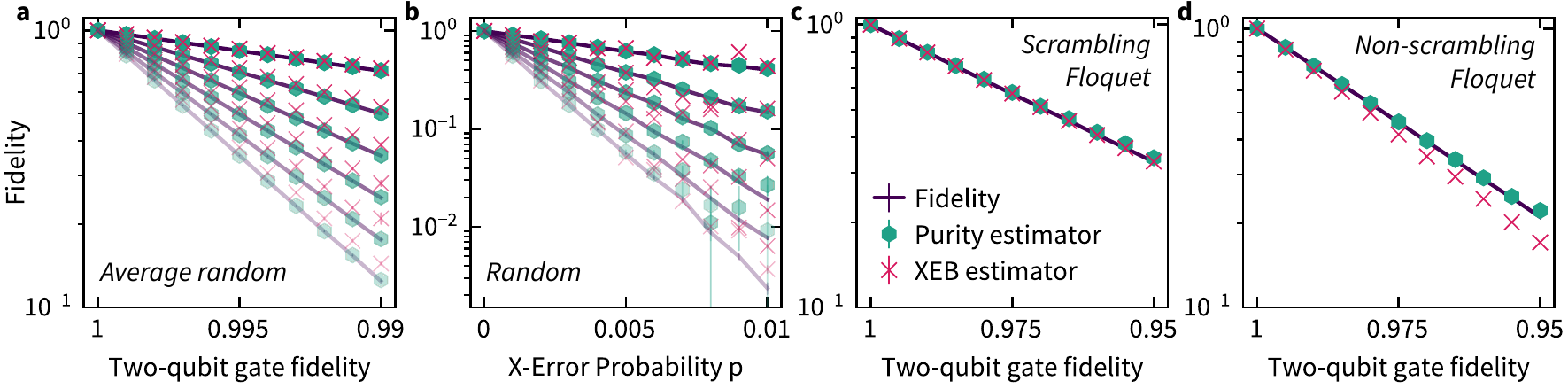}
  \caption{\textbf{Performance of the Bell sampling fidelity estimator.} To assess the purity estimator, we simulate random and non-random Clifford circuits and perform direct fidelity estimation, Bell sampling, and standard-basis sampling to compute the fidelity (green lines), root purity (blue hexagons), and normalized linear XEB (pink crosses), respectively, each obtained from $10^6$ samples per circuit. 
  For the purpose of this plot, the Bell measurement is assumed to be noise-free.
  Error bars represent one standard deviation. 
  \emph{(a)} Average fidelity over 100 random two-qubit Clifford circuits for $n = 6,12,18,24,30,36$ qubits (in decreasing opacity) in an all-to-all connected architecture with $12$ gate layers with local depolarizing noise. 
  \emph{(b)} Single-instance of a random circuit for $n = 6,12,18,24,30,36$ qubits (decreasing opacity) with $12$ gate layers in an all-to-all connected architecture with local Pauli noise with $(X,Y,Z)$ error probabilites $(p, p/3,p/10)$. 
  \emph{(c)} Non-scrambling 1D, depth-$2$ crystalline circuit with iSWAP entangling gatesfrom Ref.~\cite{sommers_crystalline_2023} on 18 qubits with no single-qubit gates. 
  \emph{(d)} Scrambling 1D, depth-$2$ crystalline circuit with iSWAP entangling gatesfrom Ref.~\cite{sommers_crystalline_2023} on 18 qubits with single-qubit $\sqrt{X}$ gates following every two-qubit gate. 
  \label{fig:purity_xeb_fidelity}}
\end{figure*}

\subsubsection{Accounting for measurement noise} 
\label{app:measurement noise}

When using Bell samples in order to estimate the purity, the measurement will of course not be ideal, but as noisy as the other gates in the circuit.
This distorts the estimator of $\tr[\rho^2]$. 
When the noise is accurately modelled by iid.\ single-qubit Pauli noise following the two-qubit gates in the circuit, we can accurately account for the error, however. 

To this end, observe that the Pauli errors at the end of the circuit will either not affect the Bell outcome (for $X$ and $Z$ errors on the first and second copy, respectively) or flip the measurement outcome. 
This implies that exactly $2n$ error locations are added to the circuit, albeit with probability of an error of $2\eta/3$. 
Hence, the purity $P_{\mathrm{Bell}}$ estimated from noisy Bell measurements will be given by
\begin{align}
    P_{\mathrm{Bell}} \approx \ee^{- 2E  \eta - 2n \eta/3} = \ee^{- E \eta \cdot  (E/n +2/3 ) \cdot 2 n /E } = \overline F^{ (E/n +2/3 ) \cdot 2 n /E }, 
\end{align}
with $\overline F = \ee^{- \eta E }$.
If the number of two-qubit gates in the circuit is given by $m$ and $E = 2m$, this suggests a fidelity estimator given by 
\begin{align}
  \label{eq:noisy fidelity estimator purity}
  \overline F_{\text{purity}}=  P_{\mathrm{Bell}}^{m/(n (2m/n +2/3 )) }. 
\end{align}

\subsubsection{Numerical simulations}

In order to assess the fidelity estimate obtained from Bell sampling in practice, we numerically simulate quantum circuits subject to Pauli noise, see \cref{fig:purity_xeb_fidelity}, and compare the fidelity estimator from the Bell samples---given by the root purity $\sqrt{\tr[\rho^2]}$---with fidelity estimators based on linear cross-entropy benchmarking (XEB) \cite{dalzell_random_2021,choi_preparing_2023,ringbauer_verifiable_2023}. 
However, computing the true fidelity and estimating the linear XEB requires exponential runtime. 
This is why we resort to Clifford circuit simulations in which the comparison can be done efficiently. 
In particular, the true fidelity of the noisy state can be efficiently estimated using direct fidelity estimation (DFE) \cite{flammia_direct_2011}.

In DFE, random elements $s \in \mc S$ of the stabilizer group $\mc S$ of a Clifford state $\ket {\mc S}$ are measured, yielding outcomes $\sigma\in \pm 1$ and the results averaged. 
From the expansion of the fidelity as
\begin{align}
   F(\rho, \proj{\mc S}) = \frac 1 {2^n} \sum_{s \in \mc S} \sum_{\sigma = \pm 1} \tr[\rho \pi_s^\sigma] \cdot \sigma,
\end{align}
where we have written $s = \pi_s^+ - \pi_s^-$, this procedure yields an estimator of the fidelity. 
For the linear XEB fidelity estimator, we estimate the ideal XEB of a state $\ket \psi$ with output distribution in the computational basis $p(x) = |\braket{x}{\psi}|^2$ 
\begin{align}
  \chi_{\mathrm {ideal}} = 2^n \sum_x p(x)^2 - 1, 
\end{align}
as well as the noisy XEB of a state preparation $\rho$ with output distribution $q(x) = \bra x \rho \ket x$ 
\begin{align}
  \chi = 2^n \sum_x q(x) p(x) -1, 
\end{align}
using samples from $q$.
The fidelity estimator is then given by~\cite{ringbauer_verifiable_2023}
\begin{align}
   F_{\mathrm{XEB}} = \frac \chi {\chi_{\mathrm {ideal}}}. 
\end{align}
In order to compute circuit-averaged fidelities, we average the DFE and root purity estimator over many random circuits; for the XEB estimator we compute the average noisy and ideal XEB values $\overline \chi$ and $\overline \chi_{\mathrm{ideal}}$, respectively, and estimate the average fidelity $\overline F_{\mathrm{XEB}} = \overline \chi/\overline \chi_{\mathrm{ideal}} $.

We assess the quality of the XEB and root purity estimators both for averaged and typical random circuits and non-random circuits. 
An interesting example of non-random Clifford circuits are the crystalline Floquet circuits studied by \textcite{sommers_crystalline_2023}, since these circuits have been classified in terms of their scrambling properties. 

In \cref{fig:purity_xeb_fidelity}(a,b) we show the three estimators for random circuits with uniformly random two-qubit Clifford gates applied to random pairs of qubits. 
In (a), we show the average fidelity for circuits subject to local depolarizing noise following each two-qubit gate. 
In (b), we show the single-instance fidelity for typical circuits subject to non-depolarizing Pauli noise. 
In \cref{fig:purity_xeb_fidelity}(c,d) we turn to non-random circuits, namely, the crystalline circuits of Ref.~\cite{sommers_crystalline_2023}. 
We consider a 1D brickwork circuit with iSWAP entangling gates followed by single-qubit $\sqrt X $ gates (d), and no single-qubit gates (c). 
These two variants of the circuit are fast scrambling (d) and non-scrambling (c), respectively. 
In all variants of the circuit, we find that the root-purity is an accurate estimator of the true fidelity, while the XEB estimator is significantly less accurate. 
Importantly, \cref{fig:purity_xeb_fidelity} does not include measurement noise.

\begin{figure}
  \includegraphics{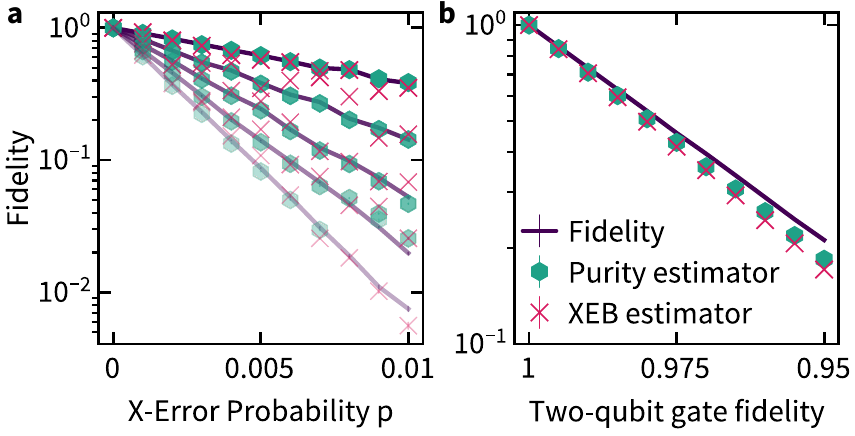}
  \caption{\textbf{Fidelity estimation based on noisy Bell sampling.} We simulate noisy Bell sampling and XEB including noisy measurements using $10^7$ samples. 
  \emph{(a)} Single-instance fidelity (lines), purity estimator (hexagons), and XEB (crosses) of typical random Clifford circuits for $n = 6,12,18,24,30$ (in decreasing opacity) with depth $12$ and Pauli $(X,Y,Z)$ error probabilies $p (1,1/3,1/10)$. 
  \emph{(b)} Fidelity of non-random scrambling crystalline circuits \cite{sommers_crystalline_2023} of depth $2$ acting on $18$ qubits in 1D.
  \label{fig:noisy purity and fidelity suppmat}}
\end{figure}

In order to account for measurement noise, we use the modified estimator \eqref{eq:noisy fidelity estimator purity}, and show the results in \cref{fig:noisy purity and fidelity,fig:noisy purity and fidelity suppmat}. 
We find that the root-purity estimator consistently outperforms the XEB, notably in regimes in which the XEB dramatically fails, such as low depths. 

One drawback of the root-purity as a fidelity estimator becomes apparent when considering the curves for a larger number of qubits \cref{fig:purity_xeb_fidelity}(b) and \cref{fig:noisy purity and fidelity suppmat}(a). 
This is due to the fact that the estimation error of the root purity scales as $\epsilon \sim 1/(F \sqrt M)$ in the number of samples $M$ and the fidelity $F$. 
In contrast, for the XEB and DFE, we obtain an error scaling $1/\sqrt M$ independently of the value of the fidelity. 
In order to resolve low fidelities we then need an estimation error $\epsilon \sim F$, so that for XEB and DFE we have a scaling of the required number of samples as $M \sim 1/F^2$, whereas for the purity we require $M \sim 1/F^4$. 

\subsubsection{Relation to verified quantum computation}

The results in this section imply that arbitrary  quantum computations in the Bell sampling model (with randomized compiling or in an encoded setting) can be verified efficiently just from the classical Bell samples under specific assumptions about the noise in the physical device. 
Let us briefly pause and contextualize this point in the context of protocols for the verification of quantum computations in different models of computation. 
Efficiently verifying a quantum computation in full generality is known to be possible only using comparably complicated schemes that require large overheads compared to a bare quantum circuit. 

The specific schemes we have in mind are blind verified quantum computation \cite{fitzsimons_unconditionally_2017}, verification using two spatially separated entangled quantum computers \cite{reichardt_classical_2013}, and classically verifiable quantum computation using post-quantum cryptography \cite{mahadev_classical_2018-1}. 
All of these settings make strong assumptions about the capabilities of a quantum device. 
Blind verified computation requires a large space overhead compared to the circuit model since it makes use of measurement-based quantum computation and requires the ability to prepare single-qubits perfectly. 
The classical leash protocol of \textcite{reichardt_classical_2013} requires two spatially separated quantum computers which share a large number of Bell pairs, since it relies on the rigidity of CHSH inequalities. It also involves a space-time mapping similar to measurement-based quantum computation and therefore incurs a high space-overhead.
Finally, the protocol of \textcite{mahadev_classical_2018-1} requires a high overhead since the quantum computation needs to be performed in a post-quantum secure homomorphic encryption scheme and is based on the assumption that such schemes exist. 
 
Bell sampling is a model of quantum computation that allows one to efficiently glean information about a quantum state prepared by an arbitrary circuit from measurements in a single basis only. 
It is a setting that is realistic in the near and intermediate term, namely one where two copies of a circuit are entangled on a single device. 
The price to be paid for the small overhead (just a factor of 2) is that the obtained information is conditional on assumptions about the noise occurring in the physical device. 
For instance, when we use randomized compiling, that assumption boils down to an assumption about single-qubit Pauli gates and the last circuit layer being nearly noise-free, while the remainder of the noise on the other gates is not context-dependent. 
We think of this setting as comparable to blind-verified computation, which can be viewed as imposing an assumption on the noise in the single-qubit state preparation step. 
For verification of random circuits on average (through the average state fidelity) the assumption we have used (but which could be weakened) is local gate-independent noise. 

While we have made an effort to obtain a good initial understanding of the types and amounts of noise required for the verification protocol to work, an exhaustive characterization would go far beyond the scope of this letter. 
We hope that our discussion of Bell sampling as a validated model of quantum computation will motivate future research into specific instantiations.

\subsection{Testing depth}
Another property that we can efficiently check from the Bell samples is subsystem purity and thereby the entanglement structure of the state. 
To this end, given a subsystem $A \subset [n]$ we simply consider the substrings $r_A = (r_i: i \in A) \cup (r_{n+i}: i \in A)$ corresponding to this subsystem and then run the purity test. 

Now, we can use the maximal entanglement achievable by circuits in a certain architecture, in order to test for their depth. 
For geometrically local circuits in large local dimension with entangling two-qubit gates, the maximal entanglement which can be generated in any subsystem is simple: 
it is just given by the number of entangled pairs which can be generated by the circuit across the boundary of the considered region. 
Thus, for depth-$d$ circuits in a one-dimensional geometry (with closed boundary conditions), the maximal entanglement entropy of a contiguous subsystem $A$ is given by $E_A(d) = \min\{2d, |A|, n - |A|, n/2\}$.
In higher dimensions, for a subsystem $A$ this increases to $E_A(d) = \min\{|\partial A| d, |A|, n - |A|, n/2 \}$, where $\partial A$ is the number of edges in the interaction graph protruding out of the subsystem $A$. 
The maximal entanglement $E_A(d)$ is thus a depth-dependent property.

This picture provides us with a simple way to test the depth of a circuit from the Bell samples. 
We estimate $P_A(\rho) =\tr[\rho_A^2]$ for random contiguous subsystems of size $|A| = n/2$ and compare the results to the maximal achievable entanglement in each subsystem in a given architecture with circuit depth $d$. 
The minimal $d$ for which $ P_A(\rho) \geq E_A(d)$ is our lower bound for the circuit depth.

The estimate of the subsystem purity is obtained from the substrings $r_A$ as 
\begin{align}
\label{eq:hatp_A}
  \hat P_A = \frac 1 M  \left(|\{r: \pi_Y(r_A) = 0 \}| - |\{r: \pi_Y(r_A) = 1 \}|\right), 
\end{align}
and $|\hat P_A - \tr[\rho_A^2]| \leq \sqrt{2\log(\delta/2)/M}$ with probability at least $1-\delta$. 
Hence, in order to obtain an estimate of the R\'enyi-2-entanglement entropy for a circuit of depth $d\leq n/2$---given by $2^{-O(d)}$---$2^{O(d)}$ samples are required. This means that log-depth circuits can be efficiently validated, while larger depth requires superpolynomially many Bell samples.

\begin{algorithm}[H]
 \caption{Depth test from maximal entanglement \label{alg:depth max}}
 \begin{algorithmic}[1]
   \Require Bell samples $b^0, \ldots, b^M \leftarrow P_\psi$, error tolerance $\epsilon > 0 $.  
   \State Estimate $\hat P_A(\proj \psi)$ as defined in \cref{eq:hatp_A} for a contiguous subsystem $A$ with $|A| = n/2$ obtaining an estimate $\overline P_{n/2}$, for example using a median-of-means estimator.

   \Ensure $d_l = \max \left\{d: - \log (\overline P_{n/2} + \epsilon) \geq E_{n/2}(d) \right\}$

 \end{algorithmic}
\end{algorithm}

We can further refine the depth test by using  properties of \emph{random} quantum circuits. 
Consider first a Haar-random pure state. 
It has been known since the seminal work of \textcite{page_average_1993} that the average entanglement of subsystems of increasing size obeys the now so-called \emph{Page curve}: as the subsystem size increases to $n/2$, its average entanglement entropy increases, as it increases beyond $n/2$, it decreases back to $0$ at subsystem size $n$. 
The precise shape of the Page curve including the so-called Page correction---the deviation of the entanglement at subsystem size $n/2$ from maximal entanglement---is known for various families of random quantum states \cite{garnerone_typicality_2010,collins_matrix_2013,fukuda_typical_2019,iosue_page_2023}.
Notably, \textcite{garnerone_typicality_2010} show that random MPSs exhibit typical entanglement whenever the bond dimension scales at least as a square of the system size. 

However, depth-dependent Page curves have to the best of our knowledge not been computed analytically yet. Nonetheless, empirically one finds that for various circuit families, the bounds $E_A(d)$ can be nearly exhausted, see for example Ref.~\cite{sommers_crystalline_2023}. 

Rather than comparing to the maximal achievable entanglement $E_A(d)$, if we have a depth-dependent Page curve, we can compare the result to the anticipated value of the Page curve $T_A(d) =- \log( \mb E_{C \sim \mc C_d} \tr[( \tr_{A^c}\proj C )^2])$ for a size-$n/2$ subsystem $A$ for the circuit family  $\mc C_d $ of a given depth $d$. 

\begin{algorithm}[H]
 \caption{Depth test from average entanglement \label{alg:depth average}}
 \begin{algorithmic}[1]
   \Require Bell samples $b^{C_1}_0, \ldots, b^{C_L}_M \leftarrow P_C$ for $C_1, \ldots, C_L \leftarrow \mc C$, error tolerance $\epsilon > 0 $.  
   \State Estimate $\frac 1 K \sum_{i=1}^L \hat P_A(\proj \psi)$ with $\hat P_A$ defined as in \cref{eq:hatp_A}, from choosing a contiguous subsystem with $|A| = n/2$ obtaining an estimate $\overline P_{n/2}$, for example using a median-of-means estimator.

   \Ensure $d_l = \max \left\{d: - \log (\overline P_{n/2} + \epsilon) \geq t_{A}(d)\right\}$

 \end{algorithmic}
\end{algorithm}

\paragraph*{Small amounts of noise}

The presence of a small amount of noise only slightly distorts our estimate of the purity and subsystem purity. 
We can always write the noisy state as 
\begin{align}
  \rho = (1-\eta) \proj \psi + \eta \chi, 
\end{align}
where $\chi \perp \proj \psi$.
The deviation of its purity from unity is then given by 
\begin{align}
  |1 - \tr \rho^2| &= 2\eta + O(\eta^2) 
\end{align}
and likewise for the subsystem purity
\begin{align}
  \big|\tr &\proj \psi_A^2 - \tr \rho_A^2\big|  \\
  &\leq 2 \eta \left|\tr \proj \psi_A^2 - \tr \proj \psi_A \chi_A\right| + O(\eta^2)\\
  & \leq 2 \eta (\tr \proj \psi_A^2 + 1)+ O(\eta^2)\\
  & \leq  \eta (3 +2^{- O(d)}) + O(\eta^2). 
\end{align}
Hence to verify depth we need states with exponentially small impurity $2^{- O(d)}$ in the circuit depth $d$.

\paragraph*{Larger amounts of noise}

The discrepancy of the Page curve for a pure state and the Page curve of a mixed state can be intuitively understood: 
For a maximally mixed state the R\'enyi-2 entropy of subsystems of size $k$ is exactly given by $-\log 2^{-k}$, and hence, as a function of subsystem size, it is just given by the subsystem size. 
The entanglement entropy of a white-noise state $\rho = (1- \eta) \proj \psi +  \eta \id/2^n$ is thus given by 
\begin{multline}
  - \log \tr\rho_A^2 = \\-\log \big[ (1- \eta)^2\tr(\proj \psi_A^2) +  ( 2\eta( 1- \eta) + \eta^2) 2^{-k}\big]  \\
  = - \log \big[(1- \eta)^2\tr(\proj \psi_A^2)\big]\\+ \frac{( 2\eta( 1- \eta) + \eta^2) 2^{-k}}{(1- \eta)^2\tr(\proj \psi_A^2)} + O(2^{-2k}),
\end{multline}
which yields a good approximation for $\tr ( \proj \psi_A^2 ) \gg 2^{-k}$.

\subsection{Learning a Clifford + $T$ circuit}

In this section, we show that a quantum state prepared by a circuit with few non-Clifford gates can be learned efficiently. 
We separate the proof into several steps.
First, we derive the expansion of the operator $\ket \psi \bra {\overline \psi}$ corresponding to a quantum state $\ket \psi$ generated   by a Clifford circuit with few $T$ gates. 
This operator determines the Bell sampling distribution, and in deriving it, we will already see the central concepts we will use in the learning algorithm.
Motivated by this decomposition, we elaborate the algorithm. 
Finally, we analyze statistical errors when running the algorithm and show the output of the algorithm is close to $\ket \psi$ in fidelity.

Before we describe the protocol, let us recap some simple properties of the Bell samples from the stabilizer state $\proj S =2^{-n} \sum_{\sigma \in \mc S}\sigma$ with $n$-dimensional stabilizer group $\mc S \subset \sf P_n$, i.e., a commuting subgroup of the $n$-qubit Pauli group $\sf P_n$. 
For stabilizer states $\ket S$, the complex conjugation $\ket {\overline S} = \sigma_k \ket S$ is described by a Pauli operator $\sigma_k$ that depends on $\ket S$ \cite{montanaro_learning_2017}. 
Let us denote by roman letters the binary symplectic subspace $S\subset \mb F_2^{2n}$ corresponding to a subgroup $\mc S$ of $\sf P_n$, which includes all but the phase information about $\mc S$. 
For $a,b \in \mb F_2^{2n}$, the symplectic inner product is given by $\omega(a,b) \coloneqq (\sum_{i=1}^n a_i b_{n+i} - b_i a_{n+i}) \mod 2$. 

From \cref{eq:output probs} it immediately follows that the output distribution of Bell sampling from $\ket S \otimes \ket S$ is supported on the affine space $S \oplus k\coloneqq \{s \oplus k: s \in S\}$. 
We can therefore learn $S$ from differences of the Bell samples $b^i \oplus b^j \in S $, and the missing phases of the stabilizers from a measurement of the corresponding stabilizer operators \cite{montanaro_learning_2017}.

\subsubsection{The Bell distribution of low $T$-count quantum states}

Consider a state $\ket \psi = C_t T_{x_t} C_{t-1} T_{x_{t-1}} \cdots T_{x_1} C_0  \ket 0 $ generated by a circuit comprising $t+1$ Clifford layers $C_i$ and $t$ $T$ gates at positions $x_i \in [n]$ for $i\in [t]$. 
Then, we can shift the $T$-gates to the end of the circuit as
\begin{align}
  \ket \psi & = \tilde T_t \cdots \tilde T_1 C_t \cdots C_0 \ket 0 \eqqcolon \tilde T_t \cdots \tilde T_1 \ket S, 
\end{align}
where we have defined $\tilde T_i = C_t \cdots C_i T_{x_i} C_i^\dagger \cdots C_t^\dagger$.
Hence, we can write 
\begin{align}
  \ket{\psi}\bra{\overline \psi}  = \prod_{i=1}^t (\alpha \id + \ii \beta P_i ) \proj S  \prod _{j=1}^t (\alpha \id - \ii \beta P_j) K^\dagger, 
\end{align}
where $P_i = C_t \cdots C_i Z_{x_i} C_i^\dagger \cdots C_t^\dagger$, $\alpha = \cos (\pi/8)$ and $\beta = \sin (\pi/8)$, and $K: \ket \psi \mapsto \ket{\overline \psi}$ is the complex-conjugation operator. 
In as next step, we  rewrite $\ket \psi\bra{\overline \psi}$ in terms of the stabilizer groups $\mc S$ which stabilizes $\ket S$ and $\mc G = \langle P_1, \ldots, P_d \rangle$ as 
\begin{align}
\ket \psi\bra{\overline \psi} & = \frac 1 {2^n} \left(\sum_{\sigma_{g } \in \mc G } \alpha_g \sigma_{g } \right) \left(\sum_{\sigma_{s } \in \mc S} \sigma_{s } \right) \left(\sum_{\sigma_{g'}  \in \mc G} \overline \alpha_{g'} \sigma_{g'}  \right) K^\dagger,
\end{align}
where the prefactors are given by 
\begin{align}
  \alpha_g = \sum_{x \in \bin^d: \prod_i P_i^{x_i} = \sign(x) \sigma_g} \alpha^{d-|x|} (\ii\beta)^{|x|} \sign(x). 
\end{align}
Now, we make use of the fact that the complex conjugation operator $K$acting on the stabilizer state $\ket S$ can be written as a Pauli-$Z$ matrix $\sigma_k \coloneqq \sigma_{10}^{\otimes s} $ for some $s \in \bin^n$ depending on $\ket S$ as $K \ket S = \sigma_k \ket S$. 
This gives us 
\begin{align}
\ket \psi &\bra{\overline \psi} = \frac 1 {2^n} \left(\sum_{\sigma_{g} \in \mc G } \alpha_g \sigma_{g} \right)\left(\sum_{\sigma_{s} \in \mc S} \sigma_{s} \right)  \sigma_{k} \left(\sum_{\sigma_{g'} \in \mc G}  \beta_{g'} \sigma_{g'} \right)\nonumber \\
& = \sum_{\sigma_g,\sigma_{g'} \in \mc G } \alpha_g\beta_{g'} \sigma_g \Pi_{\mc S} \sigma_k \sigma_{g'} \nonumber\\
& = \sum_{\sigma_g,\sigma_{g'} \in \mc G } \alpha_g\beta_{g'} (-1)^{\langle g',k\rangle} \sigma_g \Pi_{\mc S}  \sigma_{g'}\sigma_k, 
\label{eq:decomposition psipsibar}
\end{align}
where $\Pi_{\mc S} = \sum_{\sigma_s \in \mc S} \sigma_s/2^n \equiv \proj S$ denotes the projector onto the ground space of $\mc S$. 
Moreover, we let $\beta_g = \overline \alpha_g \cdot (-1)^{\pi_Y(g)}$.

Let us denote by $\langle \mc G, \mc S \rangle = \{\sigma_g\sigma_s : \sigma_g \in  \mc G, \sigma_s \in \mc S\}$ the group generated by $\mc G$ and $\mc S$, and by $\mc G \oplus \sigma = \{ \sigma_g \sigma: \sigma_g \in \mc G \}$ the shift of $\mc G$ by $\sigma$. 
We notice that all Pauli operators appearing in the sum are elements of $\langle \mc G,  \mc S\rangle \oplus \sigma_k$ with dimension $\dim(\langle \mc G, \mc S \rangle ) \leq n + t$. 
Let us decompose $\langle\mc G , \mc S \rangle  = \langle \mc C,  \mc L \rangle  $ into a maximally commuting subgroup $\mc C$, i.e., the maximal subgroup with the property that $[\sigma_s, \sigma_l] = 0,\, \forall \sigma_s \in \mc C, \sigma_l \in  \langle\mc G ,\mc S \rangle $ (the stabilizer group), and a `logical' subgroup $\mc L$ (which is ambiguous because it can be shifted by any element of $\mc C$. 
On the level of the corresponding symplectic vector spaces, $C = H \cap H^\perp$, where $H^\perp = \{s \in \mb F_2^{2n}: \omega(s, h) = 0\, \forall h \in H\}$, and $H \coloneqq \spn(G,S)$. 
To find $C$ one thus simply needs to solve a linear system of equations. 
Thus, $\dim(\mc C) \geq n-t$ while $\dim(\mc L) \leq 2 t$. 
Then we can decompose every element $\sigma$ of $\langle \mc G, S\rangle $ as $\sigma_g = \sigma_l(g) \sigma_c(g) = \sigma_c(g)  \sigma_l(g) $ for $\sigma_l(g) \in \mc L$ and $\sigma_c(g) \in \mc C$.

We can then rewrite \cref{eq:decomposition psipsibar} as 
\begin{align}
  \ket{\psi}&\bra {\overline \psi} \\
  &= \sum_{\sigma_g, \sigma_{g'} \in \mc G} \alpha_g  \beta_{g'}
  \sigma_l(g) \sigma_c(g)  \Pi_{\mc S\setminus \mc C}\Pi_{\mc C} \sigma_c(g') \sigma_l(g')  \sigma_k\\
  & = \sum_{\sigma_g, \sigma_{g'} \in \mc G} \alpha_g  \beta_{g'}
  \sigma_l(g)   \Pi_{\mc S\setminus \mc C} \underbrace{\sigma_c(g) \Pi_{\mc C} \sigma_c(g')}_{= \Pi_{\mc C}} \sigma_l(g')  \sigma_k\\
  & = \sum_{\sigma_g, \sigma_{g'} \in \mc G}\sum_{\sigma \in \mc S\setminus \mc S'} \frac{ \alpha_g  \beta_{g'}}{2^{\dim (\mc L)}}
  \underbrace{\sigma_l(g)  \sigma\sigma_l(g')}_{\in \mc L} \Pi_{\mc C} \sigma_k\\
& \eqqcolon \sum_{l \in  L} \lambda_l \sigma_l \Pi_{\mc C} \sigma_k  \eqqcolon \sum_{h \in \spn(L,C)\oplus k} q(h) \sigma_h \label{eq:pauli decomposition low t count}. 
\end{align}
In the first equality we have used that $\Pi_{\mc S } = \Pi_{\mc C} \Pi_{\mc S\setminus \mc C}$, and in the last equality we have grouped the operators in $\mc L$ and defined its coefficients as $\lambda_l \in \mb C$, $l = 1, \ldots, 4^{t}$. 
Fixing an $L$ such that $\spn(L,C)= H$,
and given $l\in L$ we fix the phase of the corresponding $\sigma_l\in \mc L$ in the logical (non-commuting) subgroup $\mc L$ to be $+1$, making the coefficients $\lambda_l$ unique.
Since $L$ decomposes $C$ into $2^{\dim(L)}$ disjoint cosets, for $h = l \oplus s \oplus k$, we also define $q(h) = \lambda_l/2^{\dim(C)}$.

When we perform Bell sampling from $\ket\psi \otimes \ket \psi $, a Bell sample $b$ will then be distributed as 
\begin{align}
\label{eq:bell distribution low t -count}
  P_\psi(b) = |\tr[\ket \psi \bra {\overline \psi } \sigma_b]|^2 /2^n = 2^n|q(b)|^2. 
\end{align}
In particular the distribution is supported on $  \spn( C, L)  \oplus k$. 

A subspace $S\subset \bin^n$ is \emph{isotropic} if for all $s,t \in S$, $\omega(s,t) = 0$. An isotropic subspace $S$  thus corresponds to a commuting subgroup $\mc S$, since $\omega(s,t) =0$ iff $[\sigma_s,\sigma_t] = 0$.
The maximal isotropic subspace $C$ of $H$ is called the \emph{radical} of $H$, which is given by $\rad H = C \cap C^\perp$.

\subsubsection{The learning algorithm}

Our learning algorithm is based on the decomposition \eqref{eq:pauli decomposition low t count} and the resulting Bell distribution \eqref{eq:bell distribution low t -count}. 
In the algorithm, we assume access to state preparations of $\ket \psi$

 \begin{algorithm}[H]
   \caption{Magic estimation \label{alg:magic}}
   \begin{algorithmic}[1]
     \Require $M \in \mb N$
      \State 
      Perform Bell sampling from $\ket \psi$, obtaining samples $b^0, \ldots, b^M \leftarrow P_\psi$.  
      \label{algstep:sample}

      \State Compute all Bell differences $b^{(i,j)} = b^j \oplus b^i$. \label{algstep:differences}
       \State Define $G' = \spn(\{b^{(i,j)} \}_{i,j})$, and $\hat t = \dim(G') - n$. \label{algstep:basis} 
   \end{algorithmic}
 \end{algorithm}

 \begin{algorithm}[H]
   \caption{Clifford+$T$ learning algorithm\label{alg:Clifford+T}}
   \begin{algorithmic}[1]
     \Require Error threshold $\epsilon$, failure probability $\delta$.
      \State Run \cref{alg:magic} with $M =  2n \log(1/{\delta})/\epsilon$, yielding $\hat t, G'$. 
      \State Find a set of generators $s^1, \ldots, s^{\hat t}$ of the radical $C' = \rad G'$ of $G'$ by solving a linear system of equations.

      \State
      Find a Clifford unitary $U_{C'}$ such that $U_{C'} \sigma_{s^i} U_{C'}^\dagger = \pm Z_i$ for all $i = 1, \ldots, \hat t$. \label{algstep:Clifford finding} 

    \State
      Let $M' = O(2^k \log(1/\delta)/\epsilon^2)$. 
      Prepare $M'$ copies of $U_C' \ket \psi$. 
      Measure qubits $1, \ldots, n- \hat t $ in the computational basis, and qubits $n-\hat t + 1, \ldots, n$ according to the pure-state tomography scheme of Ref.~\cite{guta_fast_2020}, yielding data sets $X = \{x^1, \ldots, x^{M'}\},D = \{d^1, \ldots, d^{M'}\}$. Let $x \in \{0,1\}^{n - \hat t} $ be the  majority outcome of the computational-basis measurements. \label{algstep:tomography}
  \State 
    Remove those elements $d^k$ for which $x^k \neq x$ from $D$.
    \label{algstep:tomography2}
  \State 
    Run the recovery algorithm of Ref.~\cite{guta_fast_2020} on $D$. 
    \label{algstep:tomography3}
 \Ensure 
  $U_{C'}, x, \ket {\hat \varphi}$.
   \end{algorithmic}
 \end{algorithm}

\subsubsection{Correctness of \cref{alg:Clifford+T}}

To show the correctness of the algorithm, we proceed in several steps. 
In the first step, we show that 
the state $U_{C'} \ket \psi$ is close to a product state $\ket x \otimes \ket \varphi$, since all elements of $\mc C'$ are close to stabilizers of $\ket \psi$. 
In the second step, we show that the tomographic estimates we obtain from measuring the first $n- \hat t$ qubits in the computational basis, and performing state tomography on others, yield an overall reconstruction $\ket{\hat \psi}$ with fidelity at least $1- \epsilon$ with the target state $\ket \psi$.

\begin{lemma}
\label{lem:fidelity stabilizer}
  Let $C'$ be the radical of the subspace spanned by $M \in O(n \log(1/\delta) /\epsilon)$ Bell samples from a pure state $\ket \psi$, and let $U_{C'}$ be the associated Clifford unitary. 
  Then with probability at least $1- \delta$ over the Bell samples, there is a bit string $x \in \bin ^{\dim(C')}$ and a pure state $\ket \varphi \in (\mb C^2)^{\otimes n - \dim (C')}$ such that $|\bra x \bra\varphi  \psi \rangle|^2 \geq 1 - \epsilon$. 
\end{lemma}
\begin{proof}
We begin by observing that all elements of $\mc C'$ are approximate stabilizers of $\ket \psi$. 
To see this, we observe that for any Pauli operator $P$ we can measure $\bra \psi P \ket \psi^2$ using the Bell samples $B = \{b^1, \ldots b^M\}$ as \cite{huang_information-theoretic_2021}
\begin{multline}
  \bra \psi P \ket \psi^2 \approx \hat e(P) \coloneqq  \frac 1M\big( |\{b \in B:  P \otimes P \ket {\sigma_b} = +1 \ket {\sigma_b}\}| \\-  |\{b \in B:  P \otimes P \ket {\sigma_b} = -1 \ket {\sigma_b}\}|\big).
\end{multline}
Let us write the Pauli operators corresponding to the Bell samples $\sigma_{b^i} = \sigma_{g^i} \sigma_k$ for ${g^i} \in G'$ and $\sigma_k$ the Pauli operator corresponding to complex conjugation. 
Then, for all $c \in C'$ and $i \in [M]$, $[\sigma_c, \sigma_{g^i}] = 0 $, since $C'$ is the radical of $G'$. 
Hence, 
\begin{align}
  \sigma_c \otimes \sigma_c \ket {\sigma_{b^i}} & = \sigma_c \sigma_{b^i} \otimes \sigma_c \ket{\Phi^+}\\
  & = \pm( \sigma_c \sigma_{g^i} \sigma_c \sigma_k \otimes \id) \ket{\Phi^+}\\
  & =  \pm(  \sigma_{g^i} \sigma_k \otimes \id) \ket{\Phi^+} = \pm \ket {\sigma_{b^i}},
\end{align}
where the sign depends on whether $\sigma_c^T = \pm \sigma_c$ and $\sigma_k \sigma_c = \pm \sigma_c \sigma_k$, but not on the Bell sample $b^i$.
Hence, all estimated expectation values $e(\sigma_c)=  1$. 
This implies that our estimate for $\bra\psi U_{C'}^\dagger Z_1^{z_1} \cdots Z_{n-\hat t}^{z_{n-\hat t}} U_{C'}\ket \psi ^2 $ equals $1$ for all $z \in \bin^{n- \hat t }$, and therefore the first $n- \hat t$ qubits of $ U_{C'}\ket \psi$ are close to a computational-basis state.

We now bound the distance of $\rho_L = \tr_{[n-\hat t]^c}[\proj \psi] $ from a computational-basis state. 
By the union bound with failure probability $\delta$, $M \geq2  \log (1/\delta)/\epsilon$  Bell samples are sufficient to ensure that $|e(\sigma_c) - \bra \psi \sigma_c \ket \psi^2| < \epsilon$ and hence $\bra \psi \sigma_c \ket \psi^2| > 1-\epsilon$ \cite{mnih_empirical_2008}. 
By another union bound, this imples that $M \geq2  n \log (1/\delta)/\epsilon$ samples are sufficient to ensure that this is the case for \emph{all} $c \in C'$. 
Let $s(\sigma_c) = \sign(\bra \psi \sigma_c \ket \psi)$. 
Then can use direct fidelity estimation \cite{flammia_direct_2011} to compute the fidelity with the computational-basis state $\ket x$ that satisfies $\bra x Z_c \ket x = s(\sigma_c)$, where $Z_c = U_{C'} \sigma_c U_{C'}^\dagger$. 
We find
\begin{align}
  \tr[\rho_L \proj x] & = \frac 1 {2^{n - \hat t } }\sum_{c \in C'}[ s(\sigma_c)\tr [\rho_L Z_c]
  \\& \geq \frac 1 {2^{n - \hat t }} \sum_{c \in C'}[ s(\sigma_c) s(\sigma_c) ( 1- \epsilon)] \\
  &= 1 - \epsilon,
\end{align}
since $\sqrt{1 - \epsilon} \geq 1 - \epsilon$.

What remains to be shown is that the state $U_{C'} \ket \psi$ is $\epsilon$-close to the product state $\ket x \ket \varphi$ with some $\ket \varphi$ in fidelity. 
To see this, we observe that we can write $U_{C'} \ket {\psi} = \sum_y a_y \ket {y}\ket{\varphi_y}$ for some post-measurement states $\ket \varphi_y = \bra y \otimes \id  \ket \psi/ \norm{(\bra y \otimes \id) \ket \psi}$. 
We find $|\alpha_x |^2 \ge 1-\epsilon$, and setting $\ket \varphi \coloneqq \ket \varphi_x$ proves the claim. 
\end{proof}

The tomography steps \ref{algstep:tomography} returns $x$ with exponentially low failure probability. 
By discarding those elements from $D$ for which $x^i \neq x$ in step \ref{algstep:tomography2}, we ensure that the remaining elements will yield an estimate $\ket {\hat \varphi}$ that has fidelity at least $1- \epsilon$ with $\ket \varphi_x$, which yields the claim. 

The overall runtime of the algorithm is polynomial in $n$ and $2^t$, since finding $U_{C'}$ is achieved by solving a linear system of equations, and we perform quantum state tomography on at most $t$ qubits in the last step. 

\subsubsection{Correctness of \cref{alg:magic}}

The estimate $\hat t$ of the stabilizer nullity $t$ of $\ket \psi$ clearly satisfies $\hat t \leq t$.
It immediately follows from \cref{lem:fidelity stabilizer} that there is a state $\ket \phi = U_{C'}^\dagger \ket x \ket \varphi $ with fidelity $| \braket{\phi}{\psi}|^2 \geq 1 -\epsilon$ and stabilizer nullity exactly $\hat t$. 
To see this, observe that $M(\ket \phi) = M(\ket x ) + M (\ket \varphi) = 0 + \hat t$, since the stabilizer nullity is additive for product states. \\

We also note that we the subspace $G'$ carries probability weight $1- \epsilon$ of both the Bell distribution $P_\psi $ and the characteristic distribution $C_\psi$ of $\ket \psi$. 
The characteristic distribution of $ \ket \psi$ is defined as $C_\psi(b) = 2^n | p(b)|^2$, where we write 
\begin{align}
  \proj \psi = \sum_b p(b) \sigma_b. 
\end{align}

To show this directly, we formulate the following Lemma, which  generalizes a well-known result of \textcite{erdos_probabilistic_1965} to nonuniform distributions over subspaces of $\mb F_2^n$.

\begin{lemma}[Weighted subspace generation]
\label{lem:subspace}
Let $S \subset \bin^{n}$ be a binary subspace of dimension at most $n$ and $P: S \rightarrow [0,1]$ be a probability distribution over that subspace.
Then the subspace $S' \coloneqq \spn( s^1, \ldots,s^M) \subset S$ spanned by $M$ samples $s^i \leftarrow P$ with label $i = 1, \ldots, M$ has the property that $\sum_{s \in S\setminus S'}P(s) < \epsilon$ with probability at least $1-\delta$ whenever
\begin{align}
   M \geq 2 n \log (n)\log\left(2/\delta\right) /\epsilon. 
 \end{align} 
\end{lemma}
\begin{proof}
  We construct the subspace iteratively, observing that the first $i$ samples which have been drawn generate a subspace $S_i = \spn(s^1, \ldots, s^i)$. 
  The $(i+1)$th sample $s^{i+1}$ now either lies in $S_i$ or in its complement. 
  If it lies in the complement the dimension of $S_{i+1}= \spn( S,s^{i+1}) $ is increased by $1$ compared to $S_i$, if it lies in $S_i$ its dimension remains unchanged. 
  Suppose the complement of $S_i$ has probability weight at least $\epsilon$. 
  Then the probability that after drawing $k$ additional samples, none of them lies in the complement of $S_i$ is given by $1 - ( 1- \epsilon)^k$ and hence $k = \log \delta/\log(1- \epsilon) \geq 2 \log(1/\delta)/\epsilon$ samples are sufficient that $\dim(S_{i+k}) \geq \dim(S_i)+1 $ with probability at least $1- \delta $. 
  Repeating this argument $n$ times, the total success probability is given by $(1-\delta)^n$ unless $P(S\setminus S_{i}) < \epsilon$ for some $i \in [nk]$. 
  Choosing the failure probability in every step as $\delta/n$ and observing that $(1- \delta/n)^n \geq 1 - 2 \delta$ proves the claim.
\end{proof}

Now, we have that $\spn(\{b^i\}_i) \oplus k \subset G'$, since  for $M > n + t $ the Bell samples are linearly dependent and hence at least one sample $b^{j_0}$ is in the span of all others. 
Let $b^M$ be such a sample. Then $\spn(\{b^{(i,M)}\}_i \in \spn(\{b^i\}_i)$ and therefore $\spn(\{b^{(i,j)} \}_{i}) \supset \spn( \{b^i\}_i)$.
But this implies that $P_\psi(G' \oplus k )\geq 1- \epsilon$. 

Let $Q_\psi(a) = \sum_b P_\psi(b) P_\psi(b+a)$ be the distribution of the Bell difference samples. 
Then the differences $b^{(2i,2i+1)}$ for $i \in \{0, \ldots, (M-1)/2\}$ are distributed according to $Q_\psi$. 
It follows from \cref{lem:subspace} that if $M \geq 4n \log(n) \log(2/\delta)/\epsilon$, $Q_\psi(G') \geq 1 - \epsilon$ with probability at least $1- \delta$.
It follows from Proposition 10 of Ref.~\cite{grewal_efficient_2023} that $C_\psi(G')\geq 1 - \epsilon$.

\section{Error detection and correction}
\label{sec:error detection}
In this section, we will outline some details of the error detection procedure, and discuss potential issues that arise due to noise in the Bell measurement itself. 

\subsection{Error reduction by error detection}

In this section, we will calculate the amount of error reduction that it is possible by error detection using the global symmetry. 
Recall that our error detection procedure runs as follows.

 \begin{algorithm}[H]
   \caption{Error detection \label{alg:error detection}}
   \begin{algorithmic}[1]
    \Require Bell sample $r \leftarrow P_\rho$. 
    \If{$\pi_Y(r) = 1$}
    \State Declare an error and abort.
    \EndIf
    \Ensure $r$
    \end{algorithmic}
  \end{algorithm}

Let us now quantify the amount of error reduction that is possible using \cref{alg:error detection}. 
First, let us recall why the error detection is correct: 
Since we know that the purity of the ideal state is unity, all samples which lead to a purity away from unity must have been due to an error. These are exactly the samples from the antisymmetric subspace, i.e., samples $r$ satisfying $\pi_Y(r) = 1$. 

In the next step, we can compute the error reduction capabilities of the algorithm. 
We do so in the simplest possible model of noise: global white noise. 
In this model we write each copy of the state as 
\begin{align}
\rho = \rho_C(\eta)  =  (1- \eta) \proj C + \eta \id/2^n, 
\end{align}
and hence the pre-measurement state as 
\begin{multline}
  \rho \otimes \rho = (1- \eta)^2 \proj C^{\otimes 2} + \frac {\eta( 1- \eta)}{2^n}( \proj C \otimes \id\\ + \id \otimes \proj C) + \frac{\eta^2}{2^{2n}} \id. 
\end{multline}

Again, we start with the case of $s=0$. 
We observe that the Bell distribution for the noisy state can be written as 
\begin{align}
  P_\rho(r) = (1 - \eta)^2 P_C(r) + \frac 1{4^n} P_e(\eta),  
\end{align}
with the error probability $P_e(\eta) = 2 \eta ( 1- \eta) + \eta^2$, i.e., the distribution of the errors is uniform.
Now, we observe that an error falls into the antisymmetric subspace with probability $D_{[2]}/2^{2n} = \frac 12 ( 1 + 1/2^n)$. 
Hence, the probability that an error is detected is given by $P_e(\eta) \cdot \frac 12 ( 1 + 1/2^n)$. 

Now, we can consider the unnormalized postselected distribution $\tilde Q_\rho(r) = (1 - \eta)^2 P_C(r) + \delta_{[2]}(r)/D_{[2]} P_e(\eta)$, and normalize it as $\tilde P_\rho(r) = \tilde Q_\rho(r)/(\sum_r \tilde Q_\rho(r))$, where $\delta_{[2]}(r) = 1 $ if $\pi_Y(r) = 0$ and $0$ otherwise. 
In order to compute the effective error reduction, we find 
\begin{align}
  \arg \min_{\epsilon} \norm{P_{\rho(\epsilon)} - \tilde P_{\rho(\eta)}}_{\ell_1} = \eta/2 + O(\eta^2). 
\end{align}
For $\eta\ll 1 $ we thus find an error reduction by a factor of $2$ from the first stage of the error detection algorithm.
Error detection based just on the entire subsystem therefore recovers the probability of no error compared to running computations on a single copy. 

In spite of doubling the number of qubits compared to standard-basis sampling, we have thus achieved the same overall post-selected fidelity at a given error rate.

\subsection{Noise in the Bell measurement}
\label{ssec:bell measurement noise}

Above, we have considered (local) noise in the state preparation while keeping the Bell measurements themselves error-free.
Of course, in an actual implementation of Bell sampling, the Bell measurements themselves will be noisy as well. 
The Bell measurement is constituted of transversal \texttt{cnot}-entangling gates and single-qubit measurements. The potential sources of error are therefore errors in the entangling gates and errors in the measurement apparatus. 

Since we can move all the errors before the Bell measurement to the state preparation, which we have already discussed, we only have to consider errors \emph{after} the entangling gates. 
First, consider single-qubit noise channels $\mc N$ with noise strength $\epsilon$ after each two-qubit entangling gates before the measurements in the Hadamard and computational basis, i.e., 
\begin{align}
  \proj \Phi^+ \rightarrow \texttt{cnot} \, {\mc N^\dagger}^{\otimes 2} ( \proj + \otimes \proj 0  ) \,\texttt{cnot} ,
\end{align}
and hence, the fidelity of the noisy measurement with the ideal measurement is just given by $\tr[ \mc N (\proj 0) \proj 0 ] \tr[ \mc N (\proj +) \proj + ] \approx (1- \epsilon)^2$. 
The global measurement fidelity is then just $(1-\epsilon)^{2n} $, and hence for noise rate $\epsilon \ll 1/n$, the fidelity is sufficiently high. 
Notice that this is also the regime in which we can meaningfully use the cross-entropy benchmark \cite{ware_sharp_2023}. 
Coherent errors such as \texttt{cnot}-over- or underrotations do not change the overall picture, since all entangling gates are carried out in parallel and hence the measurement fidelity just factorizes into the local fidelities. 

Of course, antisymmetric errors in the measurement will also be detected by our detection procedure (assuming that they do not combine with antisymmetric errors from the state preparation). 
At a high level, these favourable properties of the Bell measurement with regards to their susceptibility to errors and their capabilities to detect them is what makes them fault-tolerant gadgets in stabilizer codes as well. \\

Let us note, however, that this analysis will be drastically different in architectures in which the \texttt{cnot}-entangling gates cannot be carried out in a single circuit layer at the end of the circuit such as geometrically local architectures. 
In such architectures the error contribution from the Bell measurement itself will be significant.

\subsection{Virtual distillation using the Bell samples}

We observe that even further error suppression is possible in the white-noise model for estimation of expectation values of observables that are diagonal in the Bell basis. 
Such observables can be written as 
$A = \sum_{i,j \in \bin} a_{ij} \proj{\sigma_{ij}}$. 

Writing an arbitray noisy state preparation of $\ket \psi$ as 
\begin{align}
  \rho_\psi(\epsilon) = (1- \epsilon) \proj \psi + \epsilon \rho_\perp,
\end{align}
where $\tr[\rho_\perp \proj \psi] = 0 $, we can write the expectation value with respect to $\rho = \rho_\psi(\epsilon)$ as 
\begin{multline}
  \tr [A (\rho \otimes \rho)] = (1- \epsilon)^2 \tr[A \proj\psi^{\otimes 2}]\\
  + \epsilon(1 - \epsilon) \tr [A (\rho_\perp \otimes \proj \psi + \proj \psi \otimes \rho_\perp)] + \epsilon^2\tr [ A \rho_\perp^{\otimes 2}].
\end{multline}
Combined with the purity estimate, we can then estimate the ideal expectation value $\langle A \rangle_\psi = \tr[A \proj \psi^{\otimes 2}]$ from the noisy Bell sampling data from $\rho_\psi(\epsilon)$
\begin{multline}
\hat{\langle A \rangle}_\psi = \frac{\tr[A (\rho \otimes \rho)]}{\tr[\mb S (\rho \otimes \rho)]} 
=  \left( 1 - \tr[\rho_\perp^2] \epsilon^2 \right) \langle A \rangle_\psi + \epsilon^2\tr [A \rho_\perp^{\otimes 2}] \\
+ (\epsilon + \epsilon^2) \tr[ A (\rho_\perp \otimes \proj \psi + \proj \psi \otimes \rho_\perp )] + O(\epsilon^3).
\end{multline}
Here, we have used that $\tr[\rho_\psi(\epsilon)^2] = (1-\epsilon)^2 + \epsilon^2 \tr[\rho_\perp^2]$.

In particular, choosing $A= (P \otimes P)\mb S $, we can estimate $\bra \psi P \ket \psi^2$ with error suppression $\epsilon^2$. 
To see this, observe that $\tr[A (\rho\otimes \rho)] = \tr[P \rho P \rho]$ and 
\begin{multline}
  \hat{\langle P \rangle}_\psi^2 = \frac{\tr[P \rho P\rho)]}{\tr[\rho^2]} = ( 1- \tr[\rho_\perp]^2 \epsilon^2)   \langle P \rangle_\psi^2 \\+ ( \epsilon + \epsilon^2) \bra \psi P \rho_\perp P \ket \psi + \epsilon^2 \tr[P \rho_\perp P \rho_\perp] + O(\epsilon^3). 
\end{multline}
Generically, the term $ \bra \psi P \rho_\perp P \ket \psi$ will be exponentially suppressed and hence we obtain an error suppression from $\epsilon$ to $\min\{\epsilon/2^n, \epsilon^2\}$.

As a concrete example, consider the case where $\rho \propto e^{-\beta H}$ for a local gapped Hamiltonian   with gap $\Delta$ satisfying $\beta\Delta \gg 1$, and we are interested in estimating local Pauli expectation values in the ground state $\ket{E_0}$. 
Using the two-copy observable $A = (P \otimes P) \mb S$ we obtain the estimator $\tr[ P \rho P \rho ]/\tr[\rho^2] = |\bra{E_0} P \ket{E_0}|^2 + O(e^{-2 \beta \Delta})$.
This identity follows because $\bra{E_n} P \ket{E_0}$ for $n \ne 0$ is generically suppressed as an inverse polynomial in $n$ for local gapped Hamiltonians and local operators $P$ \cite{Hastings10}.  As a result, we can use post-processing of the Bell samples to virtually ``cool'' the system to half the temperature of the initial state.

\section{Applying the noisy simulation algorithm to Bell sampling}
\label{sec:aharonov algorithm}

In this section, we consider whether the noisy simulation algorithm of \textcite{gao_efficient_2018,aharonov_polynomial-time_2022} applies to Bell sampling.
Before we start, let us briefly recap the algorithm. 
We will use the notation of Ref.~\cite{aharonov_polynomial-time_2022}.

\subsection{Recap of the algorithm}
The key idea is to write an output probability $P(C,x)$ of a random circuit $C = U_d U_{d-1} \cdots U_1$ with Haar-random two-qubit gates $U_i$ as a path integral
\begin{align}
\label{eq:pauli path integral}
  P(C,x) = &\sum_{s_0, \ldots, s_d \in \mathsf{P}_n} \tr[\proj x s_d] \tr[s_d U_d s_{d-1}U_d^\dagger] \cdots  \nonumber\\
  &\qquad\cdots \tr[s_1 U_1 s_{0}U_1^\dagger] \tr[s_0 \proj {0^n}]\\
  \equiv \sum_{s \in \mathsf{P}_n^{d+1}} &\llangle x | s_d\rrangle \llangle s_d | \mc U_d | s_{d-1} \rrangle \, \cdots \, \llangle s_1 | \mc U_1 | s_{0} \rrangle \llangle s_0 | {0^n} \rrangle\\
  \eqqcolon \sum_{s \in \mathsf{P}_{n}^{d+1}} & f(C,s,x),
\end{align}
where $\mathsf{P}_n$ is the $n$-qubit Pauli group. 
This expression can be easily seen from the fact that the Pauli matrices form a complete operator basis and therefore $\tr[U\rho U^\dagger s] =\sum_{t \in \mathsf{P}_n} \tr[Ut U^\dagger s]  \tr[\rho t]$.
We also write $\mc U \coloneqq U \cdot U^\dagger$ and $\tr[ab] = \llangle a | b\rrangle$. 
Notice that in writing the path integral \eqref{eq:pauli path integral}, we have normalized the Pauli matrices to $\tr[pp']= \delta_{p,p'}$ for $p,p' \in \mathsf P_n$.

We can also think of the Pauli path integral as a Fourier decomposition of the output probabilities. 
In this Fourier representation, the effect of local depolarizing noise can be easily analyzed since it just acts as $\mc E(\rho) = (1 - \epsilon) \rho + \epsilon \tr[\rho] \id / 2^n$. 
The contribution of a Pauli path of a noisy quantum circuit to the total output probability thus decays with the number of non-identity Pauli operators in it (the Hamming weight of $s$) as
\begin{align}
  \tilde P(C,x) = \sum_{s \in \mathsf{P}_{n}^{d+1}} (1- \epsilon)^{|s|} f(C,s,x). 
\end{align}
The algorithm of \textcite{aharonov_polynomial-time_2022} is based on approximating this sum by truncating it to paths with weight $|s| \leq \ell$ for some $\ell \in \mb N$. 
The approximation can then be computed in time $2^{O(\ell)}$. 
Furthermore, since the output string just appears as $\llangle x | s_d \rrangle$, any marginal can be computed at the same complexity. 
Replacing the outcome string $x \in \{0,1\}^n$ by a string $y$ with characters $ \{ 0, 1, \bullet \}^n$, whenever there is a $\bullet$ at position $i$ of $y$ we write $\id_2$ at $i$, and eventually trace over $ \prod_{i: y_i \in \{0,1\}} \proj {y_i}_i \otimes \id$, which is efficient.
The algorithm then samples from the truncated distribution $\tilde p_C^{(\ell)}$ using marginal sampling. 
What remains to be shown is that the total-variation distance $\Delta = \mathsf{TVD}(\tilde p_C^{(\ell)},\tilde p_C^{(\ell)})$ between the truncated and the noisy distribution is sufficiently small in $\ell$ to give an efficient algorithm. To this end, they provide an upper bound on $\mb E_C [\Delta^2]$ using the Cauchy-Schwarz inequality, see Secs.~2 \& 3 of Ref.~\cite{aharonov_polynomial-time_2022}. 

Here, we consider two possible strategies to prove the algorithm remains efficient for Bell sampling. 
We show that the first strategy fails, and give evidence that the second strategy fails.
Together, this provides some evidence that the algorithm does not work, making Bell sampling a compelling candidate for noise-resilient sampling.

\subsection{Strategy 1: Upper bounds on the trace distance}

The first strategy we consider is to adapt the upper bound of \textcite{aharonov_polynomial-time_2022} on TVD to an upper bound on the trace distance. 
An upper bound $\epsilon$ on the trace distance of the sampled quantum state shows that the optimal single-copy measurement distinguishing probability is given by $\epsilon$. 
In Bell sampling we perform a two-copy measurement, and we can bound the two-copy trace distance in terms of the single-copy trace distance as 
\begin{align}
  \norm{\rho\otimes \rho - \sigma \otimes \sigma }_{1} &
  = \norm{\rho\otimes \rho - \rho \otimes \sigma + \rho \otimes \sigma -\sigma \otimes \sigma }_{1}\\
  &\leq \norm{\rho \otimes (\rho - \sigma)}_{1} + \norm{(\rho - \sigma)\otimes \sigma}_{1} \\
  &= 2 \norm{\rho - \sigma}_1
\end{align}

Consider the ideal pre-measurement state in the  path integral formulation
\begin{align}
\label{eq:pauli path integral state}
  \rho(C) = \sum_{s \in \mathsf{P}_n^{d+1}} & | s_d\rrangle \llangle s_d | \mc U_d | s_{d-1} \rrangle \, \cdots \, \llangle s_1 | \mc U_1 | s_{0} \rrangle \llangle s_0 | {0^n} \rrangle\\
  \eqqcolon \sum_{s \in \mathsf{P}_{n}^{d+1}} & g(C,s) | s_d \rrangle . 
\end{align}
Then we can write the noisy state as well as the state which is effectively generated when truncating the noisy path integral to paths of weight $\leq \ell$ as 
\begin{align}
\label{eq:noisy path integral state}
  \tilde \rho(C) &
  \coloneqq \sum_{s \in \mathsf{P}_{n}^{d+1}} (1- \gamma)^{|s|} g(C,s) | s_d \rrangle,\\
  \label{eq:noisy path integral state truncated}
  \tilde \rho^\ell(C) & 
  \coloneqq \sum_{s: |s| \leq \ell} (1- \gamma)^{|s|} g(C,s) | s_d \rrangle\\
  \label{eq:noisy path integral state difference}
  \Delta \tilde \rho(C) & 
  \coloneqq \sum_{s: |s| > \ell} (1- \gamma)^{|s|} g(C,s) | s_d \rrangle.   
\end{align}

Analogously to Eq.~(25) of \textcite{aharonov_polynomial-time_2022}, we can bound the trace distance
\begin{align}
   \norm{\Delta \tilde \rho}_1^2 & \equiv \mb E_C \norm{ \tilde \rho^\ell(C) -  \tilde \rho(C)}^2_1  \\
  \leq & 2^n \mb E_C \norm{ \Delta \tilde \rho(C) }^2_2\\
   = & 2^n \mb E_C \tr[(\Delta \tilde \rho(C))^\dagger\Delta \tilde \rho(C)] \label{eq:2norm bound trace distance}
  \\
   = &  2^n \mb E_C \sum_{s,s': |s'|, |s| > \ell } (1- \gamma)^{2|s|} \llangle s_d'| s_d \rrangle g(C,s) g(C,s')\\
   = & 2^n \sum_{s: |s| > \ell } (1- \gamma)^{2|s|} \mb E_C g(C,s)^2, 
\end{align}
using the Cauchy-Schwarz inequality and orthogonality of the coefficients $g(C,s)$ (note that this holds since it is just a local property of $\mb E_U [\llangle q| \mc U | p \rrangle  \llangle r | \mc U | s \rrangle]$). 

Hence, in order to upper bound $\norm{\Delta \tilde \rho}_1$ we need to upper bound $\sum_s \mb E_C [g(C,s)^2]$ for certain values of $s$.
\textcite{aharonov_polynomial-time_2022} achieve this by upper bounding the total sum over all $s$.

Following their strategy, we define a Fourier weight
\begin{align}
  V_k = 2^n \mb E_C \sum_{s: |s| = k} g(C,s)^2, 
\end{align}
and compute its properties. 

We certainly have 
\begin{align}
  V_0 & = 1\\
  V_k & = 0 \quad \forall 0 < k \leq d.
\end{align}
To see this, we just follow the argument of \textcite{aharonov_polynomial-time_2022}. In particular $V_0 =  2^n \llangle \id | \id \rrangle^{2d} \llangle \id | 0 \rrangle^2 = 1$ since $\llangle \id | 0 \rrangle = 1/\sqrt{2^n}$ and $\llangle \id | \id \rrangle = 1$. 

What remains is to compute $\sum_{k \geq d+1} V_k$. 
To this end, we can compute---using a 2-design assumption on $C$---the total Fourier weight
\begin{align}
  \sum_{k \geq 0} V_k &  = 2^n \mb E_C \sum_{k \geq 0 } \sum_{s:  |s| = k} g(C,s)^2\\
  & = 2^n \mb E_C \sum_{s',s} g(C,s)g(C,s')\\
  & = 2^n\mb E_C (\sum_{s \in  \sf P_n^{d+1}} g(C,s))^2 \label{eq:vk l1} \\
  &  = 2^n \mb E_C(\sum_{p \in \sf P_n } \bra C p \ket C)^2 \label{eq:vk l2} \\
  & = 2^n  \sum_{p,p'}\tr[ \mb E_C \proj C^{\otimes 2} (p \otimes p')] \\
  & \stackrel{C\,  2-\text{des}}{=} \frac{2^n}{2 D_{[2]}} \sum_{p,p'} \tr[(\id + \mb S)(p \otimes p')] 
  \\
  & =  \frac{2^n}{2 D_{[2]}} \sum_{p,p'} (2^n \delta_{p,\id}\delta_{p',\id} + \delta_{p,p'})\\
  & = \frac{2^n}{ D_{[2]}} \frac{2^n + 4^n}{2} = 2^n
\end{align}
Here, we have used orthogonality in reverse, and in lines \eqref{eq:vk l1} and \eqref{eq:vk l2}, we have used that $\rho = \sum_p p \tr[\rho p] = \sum_{s} s_d g(C,s)$ and hence $\tr[\rho p ] = \sum_{s: s_d = p} g(C,s)$.

Putting everything together, analogously to \textcite[][Eq.~(29)]{aharonov_polynomial-time_2022} we find that 
\begin{align}
  \norm{\Delta \tilde \rho}_1^2 \leq 2^n (1-\gamma)^{2 \ell}, 
\end{align}
which remains trivial for $\ell \in o(n)$.

The same argument as in Ref.~\cite{aharonov_polynomial-time_2022} thus cannot be used to show that the algorithm works for any measurement strategy. 
One might wonder if we can tighten the bound. We argue that we cannot.

First, observe that the only strict inequality we have used is to bound the trace distance by the Frobenius norm in \cref{eq:2norm bound trace distance}. 
We can also hardly hope to remove the factor of $2^n$ incurred in this bound, however, because we expect the state to be spread out in Hilbert space and the bound is tight for the uniform distribution. 

Second, observe that the trace distance upper bound is dominated by the sum over all $4^n$ Pauli matrices. 
Compare this to the original algorithm, where the only non-zero contributions to the final measurement outcomes were Pauli paths which ended in a $Z$-type string since the overlap with a computational basis state was computed. 
This is a reduction by precisely the factor of $2^n$ which we find in our upper bound on the trace distance.
Since the Bell distribution overlaps with almost all Pauli strings, we expect the upper bound to be similar when done directly in the Bell basis.  

\subsection{Strategy 2: The argument in the Bell basis}

Indeed, an alternative proof strategy is to directly upper bound the total-variation distance of the truncated Bell distribution. 
To this end we write the Bell-basis Pauli path integral as 
\begin{align}
\label{eq:pauli path integral}
  P(C,r) = &\sum_{s \in \mathsf{P}_{2n}^{d+1}} \llangle \Phi^+ | \mc P_r^\dagger \otimes \id_n| s_d\rrangle \llangle s_d | \mc U_d^{\otimes 2} | s_{d-1} \rrangle \, \cdots \\
  & \qquad\qquad \cdots\, \llangle s_1 | \mc U_1^{\otimes 2} | s_{0} \rrangle \llangle s_0 | {0^{2n}} \rrangle\\
  \eqqcolon & \sum_{s \in \mathsf{P}_{2n}^{d+1}} f(C,s,r)
\end{align}
Writing $s_i = (s_i^0, s_i^1) $ with $s_i^j \in \mathsf{P}_n$, we observe that $\llangle \Phi^+ | \mc P_r^\dagger \otimes \id_n| s_d\rrangle= (-1)^{\langle P_r,s_i^0 \rangle + \delta(s_i^0,Y)} 
\delta (s_i^0, s_i^1)$, where $\langle p,q \rangle = 1$ if $p$ and $q$ anticommute and zero otherwise. 
The boundary condition for the end of the Pauli path is therefore that both branches---copies of the system---must end at the same $n$-qubit Pauli.

Valid paths are therefore of the form 
\begin{align}
  (s_0^0, s_0^1) \rightarrow (s_1^0, s_1^1) \rightarrow \cdots \rightarrow(s_{d-1}^0, s_{d-1}^1) \rightarrow (s_d^0, s_d^0),
\end{align}
and there are $(4^{2n})^{d} \cdot 4^n = 4^{n(2d + 1)}$ of them.

Next when we add noise to the circuit, the sum again transforms to
\begin{align}
\label{eq:noisy path integral bell}
  \tilde P(C,r) = \sum_{s \in \mathsf{P}_{2n}^{d+1}} (1- \gamma)^{|s|} f(C,s,r).
\end{align}

Let us now bound the total-variation distance between an $\ell$-truncated noisy path integral $\tilde P_\ell$  and the non-truncated path integral \eqref{eq:noisy path integral bell} as 
\begin{align}
  \mb E_C [\Delta^2] & =\mb E_C  \left(\sum_r |\tilde P(C,r) - \tilde P_\ell(C,r) | \right)^2 \\
  & \leq 2^{2n} \mb E_C \sum_r (\tilde P(C,r) - \tilde P_\ell(C,r) )^2\\
  & \leq 2^{2n} \mb E_C \sum_r  \left(\sum_{s: |s| > \ell} (1-\gamma)^{|s|} f(C,s,r) \right)^2 \label{eq:bound sum tvd}
\end{align}

\textcite{aharonov_polynomial-time_2022} now go forward to bound their expression analogous to \eqref{eq:bound sum tvd} (up to scaling and replacing $r \leftarrow x $. 
Letting 
\begin{align}
  W_k = 2^{2n} \mb E_C \sum_{s \in \mathsf{P}_{n}^{d+1}: |s| = k} f(C,s,0^n)^2
\end{align}
be the total Fourier weight of a circuit at degree $k$, they use the following properties: 
\begin{itemize}
  \item orthogonality of the Fourier coefficients, i.e., 
  \begin{align}
    \mb E_C [f(C,s,x) f(C,s',x)] = 0, \quad \forall s \neq s'
  \end{align}

  \item Bounds on the total Fourier weight 
  \begin{align}
    W_0 &= 1\\
    W_k &= 0, \quad \forall 0 < k < d, \\
    \sum_{k \geq d+1} W_k &\in O(1),
  \end{align}
  which follow from anticoncentration, 
  \begin{align}
    2^n\mb E_C [p(C,x)^2] \in O(2^{-n}).
  \end{align}
\end{itemize}

Notice that all of  these properties are second-moment properties. In the Bell sampling, these become fourth moment properties.

\paragraph{Orthogonality}

Let us begin by considering the orthogonality property. 

To this end, consider a single gate in the circuit. The exact expression to compute is
\begin{align}
\label{eq:aharonov orthogonal expression}
  \mb E_{U \sim \mu} \llangle p^0 | \mc U |q^0 \rrangle  \llangle r^0 | \mc U | s^0 \rrangle \llangle p^1 | \mc U |q^1 \rrangle  \llangle r^1 | \mc U | s^1 \rrangle ,
\end{align}
 $p,q,r,s \in \mathsf{P}_{2^{4}}$. 
If $\mu$ is the Haar measure, using Weingarten calculus, we can rewrite the expression as
\begin{multline}
  \sum_{\pi, \sigma \in S_4} \text{Wg}(\pi \sigma^{-1}) \tr (W_\pi(p^0 \otimes p^1 \otimes r^0 \otimes r^1))\\\times \tr (W_\sigma(q^0 \otimes q^1 \otimes s^0 \otimes s^1)). 
\end{multline}
which does not obviously simplify.

Fortunately, to show orthogonality in the single-copy case (which involves only second moments), we can make use of the right-invariance of our gate set under multiplication with Pauli matrices, i.e., $\mc G\cdot p = \mc G$ for any $p \in \mathsf{P}_2$.  
Hence, we can insert a expectation value over Pauli matrices, 
\begin{align}
   \mb E_{U \sim \mc G} f(U) = \mb E_{U \sim \mc G} \mb E_{v \sim \mathsf{P}_2} f(Uv).
 \end{align}
 In our (two-copy) case, we therefore get
\begin{align}
  &\mb E_{U \sim \mc H}\llangle p^0 | \mc U |q^0 \rrangle  \llangle r^0 | \mc U | s^0 \rrangle \llangle p^1 | \mc U |q^1 \rrangle  \llangle r^1 | \mc U | s^1 \rrangle \\
  & =\mb E_{U \sim \mc H} \mb E_{v \sim \mathsf{P}_{2}} \llangle p^0 | \mc Uv |q^0 \rrangle  \llangle r^0 | \mc Uv | s^0 \rrangle \llangle p^1 | \mc Uv |q^1 \rrangle  \llangle r^1 | \mc Uv | s^1 \rrangle \\
  & =\mb E_{U \sim \mc H} \mb E_{v \sim \mathsf{P}_{2}} \tr[p^0 Uv q^0 v^\dagger U^\dagger]\tr[r^0 Uv s^0 v^\dagger U^\dagger]\\
  & \qquad\qquad\tr[p^1 Uv q^1 v^\dagger U^\dagger]\tr[r^1 Uv s^1 v^\dagger U^\dagger].
\end{align}

Hence, for orthogonality to hold it suffices that
\begin{align}
  \mb E_{v \in \mathsf{P}_{2}}[ v ^{\otimes 4} (p \otimes q \otimes r \otimes s)(v^{\dagger})^{\otimes 4}] = 0. 
\end{align}
Indeed, 
\begin{align}
  \mb E_{v \in \mathsf{P}_{2}}&[ v ^{\otimes 4} (p \otimes q \otimes r \otimes s)(v^{\dagger})^{\otimes 4}]\\
  & = \frac 1 {16} \sum_{v \in \mathsf{P}_2} (-1)^{\langle v, p\rangle + \langle v, q\rangle + \langle v, r\rangle + \langle v, s\rangle } p \otimes q \otimes r \otimes s\\
  & =  \frac 1 {16} \sum_{v \in \mathsf{P}_2} (-1)^{\langle v, pqrs\rangle } p \otimes q \otimes r \otimes s\\
  & = 0, \quad \forall pq \neq \ii^k rs,
  \label{eq:orthogonality condition bell}
\end{align}

\cref{eq:orthogonality condition bell} gives us the orthogonality property 
\begin{align}
  \mb E_{C \sim \mc D} f(C,q,r)f(C,s,r) = 0, \quad \forall q^0s^0 \neq \ii^k q^1 s^1.
\end{align}
The proof follows straightforwardly from \eqref{eq:orthogonality condition bell}, noting that the last layer of Paulis is always the same across $q$ and $s$. 

While the condition $q \neq s$ that arises in the single-copy case considered by \textcite{aharonov_polynomial-time_2022} reduces summation over $\mathsf{P}_2^2$ to summation over $\mathsf{P}_2$, the condition $pq \neq rs $ we find for the $2$-copy case reduces summation over $\mathsf{P}_4^2$ to summation over $\mc S = \{p,q,r,s \in \mathsf{P}_2: pq = \ii^k rs \} $. 
We have $|\mc S| = 4^{2n} \cdot 4^n= 4^{3n}$ since we can choose $p,q$ freely and then the pair $r,s$ is constrained to $rs = i^k pq $ for some $k$ of which there are $4^n $ choices.
As a result, similar to the trace-distance calculation, we find an additional exponential summation over Pauli strings, which blows up the sum by a factor of $4^n$.

Notice, however, that the orthogonality condition \eqref{eq:orthogonality condition bell} we have found is only a sufficient condition for the expectation \eqref{eq:aharonov orthogonal expression} to vanish, and it could be that the full Haar average has more zero terms. 
We expect, however, that condition \eqref{eq:orthogonality condition bell} is the only condition.
To prove that this is indeed the case, one needs to compute the fourth moments \eqref{eq:aharonov orthogonal expression}---a task that we leave to future work.
\medskip

To summarize, in both strategies we run into an exponential blow-up in the summation that, as of now, results in an exponential upper bound on the TVD between the sampled distribution and the target distribution. 

\putbib
\end{bibunit}


\begin{thebibliography}{0}%
\makeatletter
\providecommand \@ifxundefined [1]{%
 \@ifx{#1\undefined}
}%
\providecommand \@ifnum [1]{%
 \ifnum #1\expandafter \@firstoftwo
 \else \expandafter \@secondoftwo
 \fi
}%
\providecommand \@ifx [1]{%
 \ifx #1\expandafter \@firstoftwo
 \else \expandafter \@secondoftwo
 \fi
}%
\providecommand \natexlab [1]{#1}%
\providecommand \enquote  [1]{``#1''}%
\providecommand \bibnamefont  [1]{#1}%
\providecommand \bibfnamefont [1]{#1}%
\providecommand \citenamefont [1]{#1}%
\providecommand \href@noop [0]{\@secondoftwo}%
\providecommand \href [0]{\begingroup \@sanitize@url \@href}%
\providecommand \@href[1]{\@@startlink{#1}\@@href}%
\providecommand \@@href[1]{\endgroup#1\@@endlink}%
\providecommand \@sanitize@url [0]{\catcode `\\12\catcode `\$12\catcode
  `\&12\catcode `\#12\catcode `\^12\catcode `\_12\catcode `\%12\relax}%
\providecommand \@@startlink[1]{}%
\providecommand \@@endlink[0]{}%
\providecommand \url  [0]{\begingroup\@sanitize@url \@url }%
\providecommand \@url [1]{\endgroup\@href {#1}{\urlprefix }}%
\providecommand \urlprefix  [0]{URL }%
\providecommand \Eprint [0]{\href }%
\providecommand \doibase [0]{http://dx.doi.org/}%
\providecommand \selectlanguage [0]{\@gobble}%
\providecommand \bibinfo  [0]{\@secondoftwo}%
\providecommand \bibfield  [0]{\@secondoftwo}%
\providecommand \translation [1]{[#1]}%
\providecommand \BibitemOpen [0]{}%
\providecommand \bibitemStop [0]{}%
\providecommand \bibitemNoStop [0]{.\EOS\space}%
\providecommand \EOS [0]{\spacefactor3000\relax}%
\providecommand \BibitemShut  [1]{\csname bibitem#1\endcsname}%
\let\auto@bib@innerbib\@empty
\end{thebibliography}%


\begin{thebibliography}{76}%
\makeatletter
\providecommand \@ifxundefined [1]{%
 \@ifx{#1\undefined}
}%
\providecommand \@ifnum [1]{%
 \ifnum #1\expandafter \@firstoftwo
 \else \expandafter \@secondoftwo
 \fi
}%
\providecommand \@ifx [1]{%
 \ifx #1\expandafter \@firstoftwo
 \else \expandafter \@secondoftwo
 \fi
}%
\providecommand \natexlab [1]{#1}%
\providecommand \enquote  [1]{``#1''}%
\providecommand \bibnamefont  [1]{#1}%
\providecommand \bibfnamefont [1]{#1}%
\providecommand \citenamefont [1]{#1}%
\providecommand \href@noop [0]{\@secondoftwo}%
\providecommand \href [0]{\begingroup \@sanitize@url \@href}%
\providecommand \@href[1]{\@@startlink{#1}\@@href}%
\providecommand \@@href[1]{\endgroup#1\@@endlink}%
\providecommand \@sanitize@url [0]{\catcode `\\12\catcode `\$12\catcode
  `\&12\catcode `\#12\catcode `\^12\catcode `\_12\catcode `\%12\relax}%
\providecommand \@@startlink[1]{}%
\providecommand \@@endlink[0]{}%
\providecommand \url  [0]{\begingroup\@sanitize@url \@url }%
\providecommand \@url [1]{\endgroup\@href {#1}{\urlprefix }}%
\providecommand \urlprefix  [0]{URL }%
\providecommand \Eprint [0]{\href }%
\providecommand \doibase [0]{http://dx.doi.org/}%
\providecommand \selectlanguage [0]{\@gobble}%
\providecommand \bibinfo  [0]{\@secondoftwo}%
\providecommand \bibfield  [0]{\@secondoftwo}%
\providecommand \translation [1]{[#1]}%
\providecommand \BibitemOpen [0]{}%
\providecommand \bibitemStop [0]{}%
\providecommand \bibitemNoStop [0]{.\EOS\space}%
\providecommand \EOS [0]{\spacefactor3000\relax}%
\providecommand \BibitemShut  [1]{\csname bibitem#1\endcsname}%
\let\auto@bib@innerbib\@empty
\bibitem [{\citenamefont {Boixo}\ \emph {et~al.}(2018)\citenamefont {Boixo},
  \citenamefont {Isakov}, \citenamefont {Smelyanskiy}, \citenamefont {Babbush},
  \citenamefont {Ding}, \citenamefont {Jiang}, \citenamefont {Bremner},
  \citenamefont {Martinis},\ and\ \citenamefont
  {Neven}}]{boixo_characterizing_2018}%
  \BibitemOpen
  \bibfield  {author} {\bibinfo {author} {\bibfnamefont {S.}~\bibnamefont
  {Boixo}}, \bibinfo {author} {\bibfnamefont {S.~V.}\ \bibnamefont {Isakov}},
  \bibinfo {author} {\bibfnamefont {V.~N.}\ \bibnamefont {Smelyanskiy}},
  \bibinfo {author} {\bibfnamefont {R.}~\bibnamefont {Babbush}}, \bibinfo
  {author} {\bibfnamefont {N.}~\bibnamefont {Ding}}, \bibinfo {author}
  {\bibfnamefont {Z.}~\bibnamefont {Jiang}}, \bibinfo {author} {\bibfnamefont
  {M.~J.}\ \bibnamefont {Bremner}}, \bibinfo {author} {\bibfnamefont {J.~M.}\
  \bibnamefont {Martinis}}, \ and\ \bibinfo {author} {\bibfnamefont
  {H.}~\bibnamefont {Neven}},\ }\bibinfo {title} {\emph {Characterizing Quantum
  Supremacy in Near-Term Devices}},\ \href {\doibase 10.1038/s41567-018-0124-x}
  {\bibfield  {journal} {\bibinfo  {journal} {Nature Phys}\ }\textbf {\bibinfo
  {volume} {14}},\ \bibinfo {pages} {595} (\bibinfo {year} {2018})}\BibitemShut
  {NoStop}%
\bibitem [{\citenamefont {Liu}\ \emph {et~al.}(2022)\citenamefont {Liu},
  \citenamefont {Otten}, \citenamefont {Bassirianjahromi}, \citenamefont
  {Jiang},\ and\ \citenamefont {Fefferman}}]{liu_benchmarking_2022}%
  \BibitemOpen
  \bibfield  {author} {\bibinfo {author} {\bibfnamefont {Y.}~\bibnamefont
  {Liu}}, \bibinfo {author} {\bibfnamefont {M.}~\bibnamefont {Otten}}, \bibinfo
  {author} {\bibfnamefont {R.}~\bibnamefont {Bassirianjahromi}}, \bibinfo
  {author} {\bibfnamefont {L.}~\bibnamefont {Jiang}}, \ and\ \bibinfo {author}
  {\bibfnamefont {B.}~\bibnamefont {Fefferman}},\ }\bibinfo {title} {\emph
  {Benchmarking Near-Term Quantum Computers via Random Circuit Sampling}},\
  \href@noop {} {\  (\bibinfo {year} {2022})},\ \Eprint
  {http://arxiv.org/abs/2105.05232} {arXiv:2105.05232}\BibitemShut {NoStop}%
\bibitem [{\citenamefont {Heinrich}\ \emph {et~al.}(2022)\citenamefont
  {Heinrich}, \citenamefont {Kliesch},\ and\ \citenamefont
  {Roth}}]{heinrich_general_2022}%
  \BibitemOpen
  \bibfield  {author} {\bibinfo {author} {\bibfnamefont {M.}~\bibnamefont
  {Heinrich}}, \bibinfo {author} {\bibfnamefont {M.}~\bibnamefont {Kliesch}}, \
  and\ \bibinfo {author} {\bibfnamefont {I.}~\bibnamefont {Roth}},\ }\bibinfo
  {title} {\emph {General Guarantees for Randomized Benchmarking with Random
  Quantum Circuits}},\ \href@noop {} {\  (\bibinfo {year} {2022})},\ \Eprint
  {http://arxiv.org/abs/2212.06181} {arXiv:2212.06181}\BibitemShut {NoStop}%
\bibitem [{\citenamefont {Hangleiter}\ and\ \citenamefont
  {Eisert}(2023)}]{hangleiter_computational_2023-1}%
  \BibitemOpen
  \bibfield  {author} {\bibinfo {author} {\bibfnamefont {D.}~\bibnamefont
  {Hangleiter}}\ and\ \bibinfo {author} {\bibfnamefont {J.}~\bibnamefont
  {Eisert}},\ }\bibinfo {title} {\emph {Computational Advantage of Quantum
  Random Sampling}},\ \href {\doibase 10.1103/RevModPhys.95.035001} {\bibfield
  {journal} {\bibinfo  {journal} {Rev. Mod. Phys.}\ }\textbf {\bibinfo {volume}
  {95}},\ \bibinfo {pages} {035001} (\bibinfo {year} {2023})}\BibitemShut
  {NoStop}%
\bibitem [{\citenamefont {Arute}\ \emph {et~al.}(2019)\citenamefont {Arute},
  \citenamefont {Arya}, \citenamefont {Babbush}, \citenamefont {Bacon},
  \citenamefont {Bardin}, \citenamefont {Barends}, \citenamefont {Biswas},
  \citenamefont {Boixo}, \citenamefont {Brandao}, \citenamefont {Buell},
  \citenamefont {Burkett}, \citenamefont {Chen}, \citenamefont {Chen},
  \citenamefont {Chiaro}, \citenamefont {Collins}, \citenamefont {Courtney},
  \citenamefont {Dunsworth}, \citenamefont {Farhi}, \citenamefont {Foxen},
  \citenamefont {Fowler},\ and\ \citenamefont
  {\emph{others}}}]{arute_quantum_2019}%
  \BibitemOpen
  \bibfield  {author} {\bibinfo {author} {\bibfnamefont {F.}~\bibnamefont
  {Arute}}, \bibinfo {author} {\bibfnamefont {K.}~\bibnamefont {Arya}},
  \bibinfo {author} {\bibfnamefont {R.}~\bibnamefont {Babbush}}, \bibinfo
  {author} {\bibfnamefont {D.}~\bibnamefont {Bacon}}, \bibinfo {author}
  {\bibfnamefont {J.~C.}\ \bibnamefont {Bardin}}, \bibinfo {author}
  {\bibfnamefont {R.}~\bibnamefont {Barends}}, \bibinfo {author} {\bibfnamefont
  {R.}~\bibnamefont {Biswas}}, \bibinfo {author} {\bibfnamefont
  {S.}~\bibnamefont {Boixo}}, \bibinfo {author} {\bibfnamefont {F.~G. S.~L.}\
  \bibnamefont {Brandao}}, \bibinfo {author} {\bibfnamefont {D.~A.}\
  \bibnamefont {Buell}}, \bibinfo {author} {\bibfnamefont {B.}~\bibnamefont
  {Burkett}}, \bibinfo {author} {\bibfnamefont {Y.}~\bibnamefont {Chen}},
  \bibinfo {author} {\bibfnamefont {Z.}~\bibnamefont {Chen}}, \bibinfo {author}
  {\bibfnamefont {B.}~\bibnamefont {Chiaro}}, \bibinfo {author} {\bibfnamefont
  {R.}~\bibnamefont {Collins}}, \bibinfo {author} {\bibfnamefont
  {W.}~\bibnamefont {Courtney}}, \bibinfo {author} {\bibfnamefont
  {A.}~\bibnamefont {Dunsworth}}, \bibinfo {author} {\bibfnamefont
  {E.}~\bibnamefont {Farhi}}, \bibinfo {author} {\bibfnamefont
  {B.}~\bibnamefont {Foxen}}, \bibinfo {author} {\bibfnamefont
  {A.}~\bibnamefont {Fowler}}, \ and\ \bibinfo {author} {\bibnamefont
  {\emph{others}}},\ }\bibinfo {title} {\emph {Quantum Supremacy Using a
  Programmable Superconducting Processor}},\ \href {\doibase
  10.1038/s41586-019-1666-5} {\bibfield  {journal} {\bibinfo  {journal}
  {Nature}\ }\textbf {\bibinfo {volume} {574}},\ \bibinfo {pages} {505}
  (\bibinfo {year} {2019})}\BibitemShut {NoStop}%
\bibitem [{\citenamefont {Zhong}\ \emph {et~al.}(2020)\citenamefont {Zhong},
  \citenamefont {Wang}, \citenamefont {Deng}, \citenamefont {Chen},
  \citenamefont {Peng}, \citenamefont {Luo}, \citenamefont {Qin}, \citenamefont
  {Wu}, \citenamefont {Ding}, \citenamefont {Hu}, \citenamefont {Hu},
  \citenamefont {Yang}, \citenamefont {Zhang}, \citenamefont {Li},
  \citenamefont {Li}, \citenamefont {Jiang}, \citenamefont {Gan}, \citenamefont
  {Yang}, \citenamefont {You}, \citenamefont {Wang}, \citenamefont {Li},
  \citenamefont {Liu}, \citenamefont {Lu},\ and\ \citenamefont
  {Pan}}]{zhong_quantum_2020}%
  \BibitemOpen
  \bibfield  {author} {\bibinfo {author} {\bibfnamefont {H.-S.}\ \bibnamefont
  {Zhong}}, \bibinfo {author} {\bibfnamefont {H.}~\bibnamefont {Wang}},
  \bibinfo {author} {\bibfnamefont {Y.-H.}\ \bibnamefont {Deng}}, \bibinfo
  {author} {\bibfnamefont {M.-C.}\ \bibnamefont {Chen}}, \bibinfo {author}
  {\bibfnamefont {L.-C.}\ \bibnamefont {Peng}}, \bibinfo {author}
  {\bibfnamefont {Y.-H.}\ \bibnamefont {Luo}}, \bibinfo {author} {\bibfnamefont
  {J.}~\bibnamefont {Qin}}, \bibinfo {author} {\bibfnamefont {D.}~\bibnamefont
  {Wu}}, \bibinfo {author} {\bibfnamefont {X.}~\bibnamefont {Ding}}, \bibinfo
  {author} {\bibfnamefont {Y.}~\bibnamefont {Hu}}, \bibinfo {author}
  {\bibfnamefont {P.}~\bibnamefont {Hu}}, \bibinfo {author} {\bibfnamefont
  {X.-Y.}\ \bibnamefont {Yang}}, \bibinfo {author} {\bibfnamefont {W.-J.}\
  \bibnamefont {Zhang}}, \bibinfo {author} {\bibfnamefont {H.}~\bibnamefont
  {Li}}, \bibinfo {author} {\bibfnamefont {Y.}~\bibnamefont {Li}}, \bibinfo
  {author} {\bibfnamefont {X.}~\bibnamefont {Jiang}}, \bibinfo {author}
  {\bibfnamefont {L.}~\bibnamefont {Gan}}, \bibinfo {author} {\bibfnamefont
  {G.}~\bibnamefont {Yang}}, \bibinfo {author} {\bibfnamefont {L.}~\bibnamefont
  {You}}, \bibinfo {author} {\bibfnamefont {Z.}~\bibnamefont {Wang}}, \bibinfo
  {author} {\bibfnamefont {L.}~\bibnamefont {Li}}, \bibinfo {author}
  {\bibfnamefont {N.-L.}\ \bibnamefont {Liu}}, \bibinfo {author} {\bibfnamefont
  {C.-Y.}\ \bibnamefont {Lu}}, \ and\ \bibinfo {author} {\bibfnamefont {J.-W.}\
  \bibnamefont {Pan}},\ }\bibinfo {title} {\emph {Quantum Computational
  Advantage Using Photons}},\ \href {\doibase 10.1126/science.abe8770}
  {\bibfield  {journal} {\bibinfo  {journal} {Science}\ }\textbf {\bibinfo
  {volume} {370}},\ \bibinfo {pages} {1460} (\bibinfo {year}
  {2020})}\BibitemShut {NoStop}%
\bibitem [{\citenamefont {Zhu}\ \emph {et~al.}(2022)\citenamefont {Zhu},
  \citenamefont {Cao}, \citenamefont {Chen}, \citenamefont {Chen},
  \citenamefont {Chen}, \citenamefont {Chung}, \citenamefont {Deng},
  \citenamefont {Du}, \citenamefont {Fan}, \citenamefont {Gong}, \citenamefont
  {Guo}, \citenamefont {Guo}, \citenamefont {Guo}, \citenamefont {Han},
  \citenamefont {Hong}, \citenamefont {Huang}, \citenamefont {Huo},
  \citenamefont {Li}, \citenamefont {Li},\ and\ \citenamefont
  {\emph{others}}}]{zhu_quantum_2022}%
  \BibitemOpen
  \bibfield  {author} {\bibinfo {author} {\bibfnamefont {Q.}~\bibnamefont
  {Zhu}}, \bibinfo {author} {\bibfnamefont {S.}~\bibnamefont {Cao}}, \bibinfo
  {author} {\bibfnamefont {F.}~\bibnamefont {Chen}}, \bibinfo {author}
  {\bibfnamefont {M.-C.}\ \bibnamefont {Chen}}, \bibinfo {author}
  {\bibfnamefont {X.}~\bibnamefont {Chen}}, \bibinfo {author} {\bibfnamefont
  {T.-H.}\ \bibnamefont {Chung}}, \bibinfo {author} {\bibfnamefont
  {H.}~\bibnamefont {Deng}}, \bibinfo {author} {\bibfnamefont {Y.}~\bibnamefont
  {Du}}, \bibinfo {author} {\bibfnamefont {D.}~\bibnamefont {Fan}}, \bibinfo
  {author} {\bibfnamefont {M.}~\bibnamefont {Gong}}, \bibinfo {author}
  {\bibfnamefont {C.}~\bibnamefont {Guo}}, \bibinfo {author} {\bibfnamefont
  {C.}~\bibnamefont {Guo}}, \bibinfo {author} {\bibfnamefont {S.}~\bibnamefont
  {Guo}}, \bibinfo {author} {\bibfnamefont {L.}~\bibnamefont {Han}}, \bibinfo
  {author} {\bibfnamefont {L.}~\bibnamefont {Hong}}, \bibinfo {author}
  {\bibfnamefont {H.-L.}\ \bibnamefont {Huang}}, \bibinfo {author}
  {\bibfnamefont {Y.-H.}\ \bibnamefont {Huo}}, \bibinfo {author} {\bibfnamefont
  {L.}~\bibnamefont {Li}}, \bibinfo {author} {\bibfnamefont {N.}~\bibnamefont
  {Li}}, \ and\ \bibinfo {author} {\bibnamefont {\emph{others}}},\ }\bibinfo
  {title} {\emph {Quantum Computational Advantage via 60-Qubit 24-Cycle Random
  Circuit Sampling}},\ \href {\doibase 10.1016/j.scib.2021.10.017} {\bibfield
  {journal} {\bibinfo  {journal} {Science Bulletin}\ }\textbf {\bibinfo
  {volume} {67}},\ \bibinfo {pages} {240} (\bibinfo {year} {2022})}\BibitemShut
  {NoStop}%
\bibitem [{\citenamefont {Madsen}\ \emph {et~al.}(2022)\citenamefont {Madsen},
  \citenamefont {Laudenbach}, \citenamefont {Askarani}, \citenamefont
  {Rortais}, \citenamefont {Vincent}, \citenamefont {Bulmer}, \citenamefont
  {Miatto}, \citenamefont {Neuhaus}, \citenamefont {Helt}, \citenamefont
  {Collins}, \citenamefont {Lita}, \citenamefont {Gerrits}, \citenamefont
  {Nam}, \citenamefont {Vaidya}, \citenamefont {Menotti}, \citenamefont
  {Dhand}, \citenamefont {Vernon}, \citenamefont {Quesada},\ and\ \citenamefont
  {Lavoie}}]{madsen_quantum_2022}%
  \BibitemOpen
  \bibfield  {author} {\bibinfo {author} {\bibfnamefont {L.~S.}\ \bibnamefont
  {Madsen}}, \bibinfo {author} {\bibfnamefont {F.}~\bibnamefont {Laudenbach}},
  \bibinfo {author} {\bibfnamefont {M.~F.}\ \bibnamefont {Askarani}}, \bibinfo
  {author} {\bibfnamefont {F.}~\bibnamefont {Rortais}}, \bibinfo {author}
  {\bibfnamefont {T.}~\bibnamefont {Vincent}}, \bibinfo {author} {\bibfnamefont
  {J.~F.~F.}\ \bibnamefont {Bulmer}}, \bibinfo {author} {\bibfnamefont {F.~M.}\
  \bibnamefont {Miatto}}, \bibinfo {author} {\bibfnamefont {L.}~\bibnamefont
  {Neuhaus}}, \bibinfo {author} {\bibfnamefont {L.~G.}\ \bibnamefont {Helt}},
  \bibinfo {author} {\bibfnamefont {M.~J.}\ \bibnamefont {Collins}}, \bibinfo
  {author} {\bibfnamefont {A.~E.}\ \bibnamefont {Lita}}, \bibinfo {author}
  {\bibfnamefont {T.}~\bibnamefont {Gerrits}}, \bibinfo {author} {\bibfnamefont
  {S.~W.}\ \bibnamefont {Nam}}, \bibinfo {author} {\bibfnamefont {V.~D.}\
  \bibnamefont {Vaidya}}, \bibinfo {author} {\bibfnamefont {M.}~\bibnamefont
  {Menotti}}, \bibinfo {author} {\bibfnamefont {I.}~\bibnamefont {Dhand}},
  \bibinfo {author} {\bibfnamefont {Z.}~\bibnamefont {Vernon}}, \bibinfo
  {author} {\bibfnamefont {N.}~\bibnamefont {Quesada}}, \ and\ \bibinfo
  {author} {\bibfnamefont {J.}~\bibnamefont {Lavoie}},\ }\bibinfo {title}
  {\emph {Quantum Computational Advantage with a Programmable Photonic
  Processor}},\ \href {\doibase 10.1038/s41586-022-04725-x} {\bibfield
  {journal} {\bibinfo  {journal} {Nature}\ }\textbf {\bibinfo {volume} {606}},\
  \bibinfo {pages} {75} (\bibinfo {year} {2022})}\BibitemShut {NoStop}%
\bibitem [{\citenamefont {Morvan}\ \emph {et~al.}(2023)\citenamefont {Morvan},
  \citenamefont {Villalonga}, \citenamefont {Mi}, \citenamefont {Mandr{\`a}},
  \citenamefont {Bengtsson}, \citenamefont {Klimov}, \citenamefont {Chen},
  \citenamefont {Hong}, \citenamefont {Erickson}, \citenamefont {Drozdov},
  \citenamefont {Chau}, \citenamefont {Laun}, \citenamefont {Movassagh},
  \citenamefont {Asfaw}, \citenamefont {Brand{\~a}o}, \citenamefont {Peralta},
  \citenamefont {Abanin}, \citenamefont {Acharya}, \citenamefont {Allen},
  \citenamefont {Andersen},\ and\ \citenamefont
  {\emph{others}}}]{morvan_phase_2023}%
  \BibitemOpen
  \bibfield  {author} {\bibinfo {author} {\bibfnamefont {A.}~\bibnamefont
  {Morvan}}, \bibinfo {author} {\bibfnamefont {B.}~\bibnamefont {Villalonga}},
  \bibinfo {author} {\bibfnamefont {X.}~\bibnamefont {Mi}}, \bibinfo {author}
  {\bibfnamefont {S.}~\bibnamefont {Mandr{\`a}}}, \bibinfo {author}
  {\bibfnamefont {A.}~\bibnamefont {Bengtsson}}, \bibinfo {author}
  {\bibfnamefont {P.~V.}\ \bibnamefont {Klimov}}, \bibinfo {author}
  {\bibfnamefont {Z.}~\bibnamefont {Chen}}, \bibinfo {author} {\bibfnamefont
  {S.}~\bibnamefont {Hong}}, \bibinfo {author} {\bibfnamefont {C.}~\bibnamefont
  {Erickson}}, \bibinfo {author} {\bibfnamefont {I.~K.}\ \bibnamefont
  {Drozdov}}, \bibinfo {author} {\bibfnamefont {J.}~\bibnamefont {Chau}},
  \bibinfo {author} {\bibfnamefont {G.}~\bibnamefont {Laun}}, \bibinfo {author}
  {\bibfnamefont {R.}~\bibnamefont {Movassagh}}, \bibinfo {author}
  {\bibfnamefont {A.}~\bibnamefont {Asfaw}}, \bibinfo {author} {\bibfnamefont
  {L.~T. A.~N.}\ \bibnamefont {Brand{\~a}o}}, \bibinfo {author} {\bibfnamefont
  {R.}~\bibnamefont {Peralta}}, \bibinfo {author} {\bibfnamefont
  {D.}~\bibnamefont {Abanin}}, \bibinfo {author} {\bibfnamefont
  {R.}~\bibnamefont {Acharya}}, \bibinfo {author} {\bibfnamefont
  {R.}~\bibnamefont {Allen}}, \bibinfo {author} {\bibfnamefont {T.~I.}\
  \bibnamefont {Andersen}}, \ and\ \bibinfo {author} {\bibnamefont
  {\emph{others}}},\ }\bibinfo {title} {\emph {Phase Transition in {{Random
  Circuit Sampling}}}},\ \href@noop {} {\  (\bibinfo {year} {2023})},\ \Eprint
  {http://arxiv.org/abs/2304.11119} {arXiv:2304.11119}\BibitemShut {NoStop}%
\bibitem [{\citenamefont {Deng}\ \emph {et~al.}(2023)\citenamefont {Deng},
  \citenamefont {Gu}, \citenamefont {Liu}, \citenamefont {Gong}, \citenamefont
  {Su}, \citenamefont {Zhang}, \citenamefont {Tang}, \citenamefont {Jia},
  \citenamefont {Xu}, \citenamefont {Chen}, \citenamefont {Zhong},
  \citenamefont {Qin}, \citenamefont {Wang}, \citenamefont {Peng},
  \citenamefont {Yan}, \citenamefont {Hu}, \citenamefont {Huang}, \citenamefont
  {Li}, \citenamefont {Li}, \citenamefont {Chen},\ and\ \citenamefont
  {\emph{others}}}]{deng_gaussian_2023}%
  \BibitemOpen
  \bibfield  {author} {\bibinfo {author} {\bibfnamefont {Y.-H.}\ \bibnamefont
  {Deng}}, \bibinfo {author} {\bibfnamefont {Y.-C.}\ \bibnamefont {Gu}},
  \bibinfo {author} {\bibfnamefont {H.-L.}\ \bibnamefont {Liu}}, \bibinfo
  {author} {\bibfnamefont {S.-Q.}\ \bibnamefont {Gong}}, \bibinfo {author}
  {\bibfnamefont {H.}~\bibnamefont {Su}}, \bibinfo {author} {\bibfnamefont
  {Z.-J.}\ \bibnamefont {Zhang}}, \bibinfo {author} {\bibfnamefont {H.-Y.}\
  \bibnamefont {Tang}}, \bibinfo {author} {\bibfnamefont {M.-H.}\ \bibnamefont
  {Jia}}, \bibinfo {author} {\bibfnamefont {J.-M.}\ \bibnamefont {Xu}},
  \bibinfo {author} {\bibfnamefont {M.-C.}\ \bibnamefont {Chen}}, \bibinfo
  {author} {\bibfnamefont {H.-S.}\ \bibnamefont {Zhong}}, \bibinfo {author}
  {\bibfnamefont {J.}~\bibnamefont {Qin}}, \bibinfo {author} {\bibfnamefont
  {H.}~\bibnamefont {Wang}}, \bibinfo {author} {\bibfnamefont {L.-C.}\
  \bibnamefont {Peng}}, \bibinfo {author} {\bibfnamefont {J.}~\bibnamefont
  {Yan}}, \bibinfo {author} {\bibfnamefont {Y.}~\bibnamefont {Hu}}, \bibinfo
  {author} {\bibfnamefont {J.}~\bibnamefont {Huang}}, \bibinfo {author}
  {\bibfnamefont {H.}~\bibnamefont {Li}}, \bibinfo {author} {\bibfnamefont
  {Y.}~\bibnamefont {Li}}, \bibinfo {author} {\bibfnamefont {Y.}~\bibnamefont
  {Chen}}, \ and\ \bibinfo {author} {\bibnamefont {\emph{others}}},\ }\bibinfo
  {title} {\emph {Gaussian {{Boson Sampling}} with {{Pseudo-Photon-Number
  Resolving Detectors}} and {{Quantum Computational Advantage}}}},\ \href@noop
  {} {\  (\bibinfo {year} {2023})},\ \Eprint {http://arxiv.org/abs/2304.12240}
  {arXiv:2304.12240}\BibitemShut {NoStop}%
\bibitem [{\citenamefont {Gao}\ \emph {et~al.}(2024)\citenamefont {Gao},
  \citenamefont {Kalinowski}, \citenamefont {Chou}, \citenamefont {Lukin},
  \citenamefont {Barak},\ and\ \citenamefont {Choi}}]{gao_limitations_2021}%
  \BibitemOpen
  \bibfield  {author} {\bibinfo {author} {\bibfnamefont {X.}~\bibnamefont
  {Gao}}, \bibinfo {author} {\bibfnamefont {M.}~\bibnamefont {Kalinowski}},
  \bibinfo {author} {\bibfnamefont {C.-N.}\ \bibnamefont {Chou}}, \bibinfo
  {author} {\bibfnamefont {M.~D.}\ \bibnamefont {Lukin}}, \bibinfo {author}
  {\bibfnamefont {B.}~\bibnamefont {Barak}}, \ and\ \bibinfo {author}
  {\bibfnamefont {S.}~\bibnamefont {Choi}},\ }\bibinfo {title} {\emph
  {Limitations of {{Linear Cross-Entropy}} as a {{Measure}} for {{Quantum
  Advantage}}}},\ \href {\doibase 10.1103/PRXQuantum.5.010334} {\bibfield
  {journal} {\bibinfo  {journal} {PRX Quantum}\ }\textbf {\bibinfo {volume}
  {5}},\ \bibinfo {pages} {010334} (\bibinfo {year} {2024})}\BibitemShut
  {NoStop}%
\bibitem [{\citenamefont {Ware}\ \emph {et~al.}(2023)\citenamefont {Ware},
  \citenamefont {Deshpande}, \citenamefont {Hangleiter}, \citenamefont
  {Niroula}, \citenamefont {Fefferman}, \citenamefont {Gorshkov},\ and\
  \citenamefont {Gullans}}]{ware_sharp_2023}%
  \BibitemOpen
  \bibfield  {author} {\bibinfo {author} {\bibfnamefont {B.}~\bibnamefont
  {Ware}}, \bibinfo {author} {\bibfnamefont {A.}~\bibnamefont {Deshpande}},
  \bibinfo {author} {\bibfnamefont {D.}~\bibnamefont {Hangleiter}}, \bibinfo
  {author} {\bibfnamefont {P.}~\bibnamefont {Niroula}}, \bibinfo {author}
  {\bibfnamefont {B.}~\bibnamefont {Fefferman}}, \bibinfo {author}
  {\bibfnamefont {A.~V.}\ \bibnamefont {Gorshkov}}, \ and\ \bibinfo {author}
  {\bibfnamefont {M.~J.}\ \bibnamefont {Gullans}},\ }\bibinfo {title} {\emph {A
  Sharp Phase Transition in Linear Cross-Entropy Benchmarking}},\ \href@noop {}
  {\  (\bibinfo {year} {2023})},\ \Eprint {http://arxiv.org/abs/2305.04954}
  {arXiv:2305.04954}\BibitemShut {NoStop}%
\bibitem [{\citenamefont {Aaronson}(2007)}]{aaronson_learnability_2007}%
  \BibitemOpen
  \bibfield  {author} {\bibinfo {author} {\bibfnamefont {S.}~\bibnamefont
  {Aaronson}},\ }\bibinfo {title} {\emph {The {{Learnability}} of {{Quantum
  States}}}},\ \href {\doibase 10.1098/rspa.2007.0113} {\bibfield  {journal}
  {\bibinfo  {journal} {Proceedings of the Royal Society A: Mathematical,
  Physical and Engineering Sciences}\ }\textbf {\bibinfo {volume} {463}},\
  \bibinfo {pages} {3089} (\bibinfo {year} {2007})}\BibitemShut {NoStop}%
\bibitem [{\citenamefont {Huang}\ \emph {et~al.}(2020)\citenamefont {Huang},
  \citenamefont {Kueng},\ and\ \citenamefont
  {Preskill}}]{huang_predicting_2020}%
  \BibitemOpen
  \bibfield  {author} {\bibinfo {author} {\bibfnamefont {H.-Y.}\ \bibnamefont
  {Huang}}, \bibinfo {author} {\bibfnamefont {R.}~\bibnamefont {Kueng}}, \ and\
  \bibinfo {author} {\bibfnamefont {J.}~\bibnamefont {Preskill}},\ }\bibinfo
  {title} {\emph {Predicting Many Properties of a Quantum System from Very Few
  Measurements}},\ \href {\doibase 10.1038/s41567-020-0932-7} {\bibfield
  {journal} {\bibinfo  {journal} {Nature Physics}\ }\textbf {\bibinfo {volume}
  {16}},\ \bibinfo {pages} {1050} (\bibinfo {year} {2020})}\BibitemShut
  {NoStop}%
\bibitem [{\citenamefont {Aaronson}\ and\ \citenamefont
  {Arkhipov}(2013)}]{aaronson_computational_2013}%
  \BibitemOpen
  \bibfield  {author} {\bibinfo {author} {\bibfnamefont {S.}~\bibnamefont
  {Aaronson}}\ and\ \bibinfo {author} {\bibfnamefont {A.}~\bibnamefont
  {Arkhipov}},\ }\bibinfo {title} {\emph {The {{Computational Complexity}} of
  {{Linear Optics}}}},\ \href {\doibase 10.4086/toc.2013.v009a004} {\bibfield
  {journal} {\bibinfo  {journal} {Th. Comp.}\ }\textbf {\bibinfo {volume}
  {9}},\ \bibinfo {pages} {143} (\bibinfo {year} {2013})}\BibitemShut {NoStop}%
\bibitem [{\citenamefont {Bremner}\ \emph {et~al.}(2016)\citenamefont
  {Bremner}, \citenamefont {Montanaro},\ and\ \citenamefont
  {Shepherd}}]{bremner_average-case_2016}%
  \BibitemOpen
  \bibfield  {author} {\bibinfo {author} {\bibfnamefont {M.~J.}\ \bibnamefont
  {Bremner}}, \bibinfo {author} {\bibfnamefont {A.}~\bibnamefont {Montanaro}},
  \ and\ \bibinfo {author} {\bibfnamefont {D.~J.}\ \bibnamefont {Shepherd}},\
  }\bibinfo {title} {\emph {Average-{{Case Complexity Versus Approximate
  Simulation}} of {{Commuting Quantum Computations}}}},\ \href {\doibase
  10.1103/PhysRevLett.117.080501} {\bibfield  {journal} {\bibinfo  {journal}
  {Physical Review Letters}\ }\textbf {\bibinfo {volume} {117}},\ \bibinfo
  {pages} {080501} (\bibinfo {year} {2016})}\BibitemShut {NoStop}%
\bibitem [{\citenamefont {Beverland}\ \emph {et~al.}(2020)\citenamefont
  {Beverland}, \citenamefont {Campbell}, \citenamefont {Howard},\ and\
  \citenamefont {Kliuchnikov}}]{beverland_lower_2020}%
  \BibitemOpen
  \bibfield  {author} {\bibinfo {author} {\bibfnamefont {M.}~\bibnamefont
  {Beverland}}, \bibinfo {author} {\bibfnamefont {E.}~\bibnamefont {Campbell}},
  \bibinfo {author} {\bibfnamefont {M.}~\bibnamefont {Howard}}, \ and\ \bibinfo
  {author} {\bibfnamefont {V.}~\bibnamefont {Kliuchnikov}},\ }\bibinfo {title}
  {\emph {Lower Bounds on the Non-{{Clifford}} Resources for Quantum
  Computations}},\ \href {\doibase 10.1088/2058-9565/ab8963} {\bibfield
  {journal} {\bibinfo  {journal} {Quantum Sci. Technol.}\ }\textbf {\bibinfo
  {volume} {5}},\ \bibinfo {pages} {035009} (\bibinfo {year}
  {2020})}\BibitemShut {NoStop}%
\bibitem [{\citenamefont {Montanaro}(2017)}]{montanaro_learning_2017}%
  \BibitemOpen
  \bibfield  {author} {\bibinfo {author} {\bibfnamefont {A.}~\bibnamefont
  {Montanaro}},\ }\bibinfo {title} {\emph {Learning Stabilizer States by
  {{Bell}} Sampling}},\ \href@noop {} {\  (\bibinfo {year} {2017})},\ \Eprint
  {http://arxiv.org/abs/1707.04012} {arXiv:1707.04012}\BibitemShut {NoStop}%
\bibitem [{\citenamefont {Gross}\ \emph {et~al.}(2021)\citenamefont {Gross},
  \citenamefont {Nezami},\ and\ \citenamefont {Walter}}]{gross_schurweyl_2021}%
  \BibitemOpen
  \bibfield  {author} {\bibinfo {author} {\bibfnamefont {D.}~\bibnamefont
  {Gross}}, \bibinfo {author} {\bibfnamefont {S.}~\bibnamefont {Nezami}}, \
  and\ \bibinfo {author} {\bibfnamefont {M.}~\bibnamefont {Walter}},\ }\bibinfo
  {title} {\emph {Schur\textendash{{Weyl Duality}} for the {{Clifford Group}}
  with {{Applications}}: {{Property Testing}}, a {{Robust Hudson Theorem}}, and
  de {{Finetti Representations}}}},\ \href {\doibase
  10.1007/s00220-021-04118-7} {\bibfield  {journal} {\bibinfo  {journal}
  {Commun. Math. Phys.}\ }\textbf {\bibinfo {volume} {385}},\ \bibinfo {pages}
  {1325} (\bibinfo {year} {2021})}\BibitemShut {NoStop}%
\bibitem [{\citenamefont {Haug}\ and\ \citenamefont
  {Kim}(2023)}]{haug_scalable_2023}%
  \BibitemOpen
  \bibfield  {author} {\bibinfo {author} {\bibfnamefont {T.}~\bibnamefont
  {Haug}}\ and\ \bibinfo {author} {\bibfnamefont {M.}~\bibnamefont {Kim}},\
  }\bibinfo {title} {\emph {Scalable {{Measures}} of {{Magic Resource}} for
  {{Quantum Computers}}}},\ \href {\doibase 10.1103/PRXQuantum.4.010301}
  {\bibfield  {journal} {\bibinfo  {journal} {PRX Quantum}\ }\textbf {\bibinfo
  {volume} {4}},\ \bibinfo {pages} {010301} (\bibinfo {year}
  {2023})}\BibitemShut {NoStop}%
\bibitem [{\citenamefont {Haug}\ \emph {et~al.}(2023)\citenamefont {Haug},
  \citenamefont {Lee},\ and\ \citenamefont {Kim}}]{haug_efficient_2023}%
  \BibitemOpen
  \bibfield  {author} {\bibinfo {author} {\bibfnamefont {T.}~\bibnamefont
  {Haug}}, \bibinfo {author} {\bibfnamefont {S.}~\bibnamefont {Lee}}, \ and\
  \bibinfo {author} {\bibfnamefont {M.~S.}\ \bibnamefont {Kim}},\ }\bibinfo
  {title} {\emph {Efficient Stabilizer Entropies for Quantum Computers}},\
  \href@noop {} {\  (\bibinfo {year} {2023})},\ \Eprint
  {http://arxiv.org/abs/2305.19152} {arXiv:2305.19152}\BibitemShut {NoStop}%
\bibitem [{\citenamefont {Huang}\ \emph {et~al.}(2021)\citenamefont {Huang},
  \citenamefont {Kueng},\ and\ \citenamefont
  {Preskill}}]{huang_information-theoretic_2021}%
  \BibitemOpen
  \bibfield  {author} {\bibinfo {author} {\bibfnamefont {H.-Y.}\ \bibnamefont
  {Huang}}, \bibinfo {author} {\bibfnamefont {R.}~\bibnamefont {Kueng}}, \ and\
  \bibinfo {author} {\bibfnamefont {J.}~\bibnamefont {Preskill}},\ }\bibinfo
  {title} {\emph {Information-{{Theoretic Bounds}} on {{Quantum Advantage}} in
  {{Machine Learning}}}},\ \href {\doibase 10.1103/PhysRevLett.126.190505}
  {\bibfield  {journal} {\bibinfo  {journal} {Phys. Rev. Lett.}\ }\textbf
  {\bibinfo {volume} {126}},\ \bibinfo {pages} {190505} (\bibinfo {year}
  {2021})}\BibitemShut {NoStop}%
\bibitem [{sup()}]{suppmat}%
  \BibitemOpen
  \href@noop {} {}\bibinfo {note} {The Supplemental Material includes the
  additional references
  \cite{brandao_local_2016,zhou_emergent_2019,hunter-jones_unitary_2019,barak_spoofing_2021,dalzell_random_2022,dalzell_random_2021,ringbauer_verifiable_2023,flammia_direct_2011,fitzsimons_unconditionally_2017,reichardt_classical_2013,mahadev_classical_2018-1,garnerone_typicality_2010,collins_matrix_2013,fukuda_typical_2019,iosue_page_2023,mnih_empirical_2008,erdos_probabilistic_1965,grewal_efficient_2023,Hastings10}.}\BibitemShut
  {Stop}%
\bibitem [{\citenamefont {Bouland}\ \emph {et~al.}(2019)\citenamefont
  {Bouland}, \citenamefont {Fefferman}, \citenamefont {Nirkhe},\ and\
  \citenamefont {Vazirani}}]{bouland_complexity_2019}%
  \BibitemOpen
  \bibfield  {author} {\bibinfo {author} {\bibfnamefont {A.}~\bibnamefont
  {Bouland}}, \bibinfo {author} {\bibfnamefont {B.}~\bibnamefont {Fefferman}},
  \bibinfo {author} {\bibfnamefont {C.}~\bibnamefont {Nirkhe}}, \ and\ \bibinfo
  {author} {\bibfnamefont {U.}~\bibnamefont {Vazirani}},\ }\bibinfo {title}
  {\emph {On the Complexity and Verification of Quantum Random Circuit
  Sampling}},\ \href {\doibase 10.1038/s41567-018-0318-2} {\bibfield  {journal}
  {\bibinfo  {journal} {Nature Phys}\ }\textbf {\bibinfo {volume} {15}},\
  \bibinfo {pages} {159} (\bibinfo {year} {2019})}\BibitemShut {NoStop}%
\bibitem [{\citenamefont {Movassagh}(2020)}]{movassagh_quantum_2020}%
  \BibitemOpen
  \bibfield  {author} {\bibinfo {author} {\bibfnamefont {R.}~\bibnamefont
  {Movassagh}},\ }\bibinfo {title} {\emph {Quantum Supremacy and Random
  Circuits}},\ \href@noop {} {\  (\bibinfo {year} {2020})},\ \Eprint
  {http://arxiv.org/abs/1909.06210} {arXiv:1909.06210}\BibitemShut {NoStop}%
\bibitem [{\citenamefont {Kondo}\ \emph {et~al.}(2022)\citenamefont {Kondo},
  \citenamefont {Mori},\ and\ \citenamefont {Movassagh}}]{kondo_quantum_2022}%
  \BibitemOpen
  \bibfield  {author} {\bibinfo {author} {\bibfnamefont {Y.}~\bibnamefont
  {Kondo}}, \bibinfo {author} {\bibfnamefont {R.}~\bibnamefont {Mori}}, \ and\
  \bibinfo {author} {\bibfnamefont {R.}~\bibnamefont {Movassagh}},\ }\bibfield
  {title} {\emph {\bibinfo {title} {\emph {Quantum Supremacy and Hardness of
  Estimating Output Probabilities of Quantum Circuits}},\ }}in\ \href {\doibase
  10.1109/FOCS52979.2021.00126} {\emph {\bibinfo {booktitle} {2021 {{IEEE}}
  62nd {{Annual Symposium}} on {{Foundations}} of {{Computer Science}}
  ({{FOCS}})}}}\ (\bibinfo {year} {2022})\ pp.\ \bibinfo {pages}
  {1296--1307}\BibitemShut {NoStop}%
\bibitem [{\citenamefont {Krovi}(2022)}]{krovi_average-case_2022}%
  \BibitemOpen
  \bibfield  {author} {\bibinfo {author} {\bibfnamefont {H.}~\bibnamefont
  {Krovi}},\ }\bibinfo {title} {\emph {Average-Case Hardness of Estimating
  Probabilities of Random Quantum Circuits with a Linear Scaling in the Error
  Exponent}},\ \href@noop {} {\  (\bibinfo {year} {2022})},\ \Eprint
  {http://arxiv.org/abs/2206.05642} {arXiv:2206.05642}\BibitemShut {NoStop}%
\bibitem [{\citenamefont {Gottesman}(2009)}]{gottesman_introduction_2009}%
  \BibitemOpen
  \bibfield  {author} {\bibinfo {author} {\bibfnamefont {D.}~\bibnamefont
  {Gottesman}},\ }\bibinfo {title} {\emph {An {{Introduction}} to {{Quantum
  Error Correction}} and {{Fault-Tolerant Quantum Computation}}}},\ \href@noop
  {} {\  (\bibinfo {year} {2009})},\ \Eprint {http://arxiv.org/abs/0904.2557}
  {arXiv:0904.2557}\BibitemShut {NoStop}%
\bibitem [{\citenamefont {Sommers}\ \emph {et~al.}(2023)\citenamefont
  {Sommers}, \citenamefont {Huse},\ and\ \citenamefont
  {Gullans}}]{sommers_crystalline_2023}%
  \BibitemOpen
  \bibfield  {author} {\bibinfo {author} {\bibfnamefont {G.~M.}\ \bibnamefont
  {Sommers}}, \bibinfo {author} {\bibfnamefont {D.~A.}\ \bibnamefont {Huse}}, \
  and\ \bibinfo {author} {\bibfnamefont {M.~J.}\ \bibnamefont {Gullans}},\
  }\bibinfo {title} {\emph {Crystalline {{Quantum Circuits}}}},\ \href
  {\doibase 10.1103/PRXQuantum.4.030313} {\bibfield  {journal} {\bibinfo
  {journal} {PRX Quantum}\ }\textbf {\bibinfo {volume} {4}},\ \bibinfo {pages}
  {030313} (\bibinfo {year} {2023})}\BibitemShut {NoStop}%
\bibitem [{\citenamefont {Beale}\ \emph {et~al.}(2018)\citenamefont {Beale},
  \citenamefont {Wallman}, \citenamefont {Guti{\'e}rrez}, \citenamefont
  {Brown},\ and\ \citenamefont {Laflamme}}]{beale_quantum_2018}%
  \BibitemOpen
  \bibfield  {author} {\bibinfo {author} {\bibfnamefont {S.~J.}\ \bibnamefont
  {Beale}}, \bibinfo {author} {\bibfnamefont {J.~J.}\ \bibnamefont {Wallman}},
  \bibinfo {author} {\bibfnamefont {M.}~\bibnamefont {Guti{\'e}rrez}}, \bibinfo
  {author} {\bibfnamefont {K.~R.}\ \bibnamefont {Brown}}, \ and\ \bibinfo
  {author} {\bibfnamefont {R.}~\bibnamefont {Laflamme}},\ }\bibinfo {title}
  {\emph {Quantum {{Error Correction Decoheres Noise}}}},\ \href {\doibase
  10.1103/PhysRevLett.121.190501} {\bibfield  {journal} {\bibinfo  {journal}
  {Phys. Rev. Lett.}\ }\textbf {\bibinfo {volume} {121}},\ \bibinfo {pages}
  {190501} (\bibinfo {year} {2018})}\BibitemShut {NoStop}%
\bibitem [{\citenamefont {Huang}\ \emph {et~al.}(2019)\citenamefont {Huang},
  \citenamefont {Doherty},\ and\ \citenamefont
  {Flammia}}]{huang_performance_2019}%
  \BibitemOpen
  \bibfield  {author} {\bibinfo {author} {\bibfnamefont {E.}~\bibnamefont
  {Huang}}, \bibinfo {author} {\bibfnamefont {A.~C.}\ \bibnamefont {Doherty}},
  \ and\ \bibinfo {author} {\bibfnamefont {S.}~\bibnamefont {Flammia}},\
  }\bibinfo {title} {\emph {Performance of Quantum Error Correction with
  Coherent Errors}},\ \href {\doibase 10.1103/PhysRevA.99.022313} {\bibfield
  {journal} {\bibinfo  {journal} {Phys. Rev. A}\ }\textbf {\bibinfo {volume}
  {99}},\ \bibinfo {pages} {022313} (\bibinfo {year} {2019})}\BibitemShut
  {NoStop}%
\bibitem [{\citenamefont {Wallman}\ and\ \citenamefont
  {Emerson}(2016)}]{wallman_noise_2016}%
  \BibitemOpen
  \bibfield  {author} {\bibinfo {author} {\bibfnamefont {J.~J.}\ \bibnamefont
  {Wallman}}\ and\ \bibinfo {author} {\bibfnamefont {J.}~\bibnamefont
  {Emerson}},\ }\bibinfo {title} {\emph {Noise Tailoring for Scalable Quantum
  Computation via Randomized Compiling}},\ \href {\doibase
  10.1103/PhysRevA.94.052325} {\bibfield  {journal} {\bibinfo  {journal} {Phys.
  Rev. A}\ }\textbf {\bibinfo {volume} {94}},\ \bibinfo {pages} {052325}
  (\bibinfo {year} {2016})}\BibitemShut {NoStop}%
\bibitem [{\citenamefont {Winick}\ \emph {et~al.}(2022)\citenamefont {Winick},
  \citenamefont {Wallman}, \citenamefont {Dahlen}, \citenamefont {Hincks},
  \citenamefont {Ospadov},\ and\ \citenamefont
  {Emerson}}]{winick_concepts_2022}%
  \BibitemOpen
  \bibfield  {author} {\bibinfo {author} {\bibfnamefont {A.}~\bibnamefont
  {Winick}}, \bibinfo {author} {\bibfnamefont {J.~J.}\ \bibnamefont {Wallman}},
  \bibinfo {author} {\bibfnamefont {D.}~\bibnamefont {Dahlen}}, \bibinfo
  {author} {\bibfnamefont {I.}~\bibnamefont {Hincks}}, \bibinfo {author}
  {\bibfnamefont {E.}~\bibnamefont {Ospadov}}, \ and\ \bibinfo {author}
  {\bibfnamefont {J.}~\bibnamefont {Emerson}},\ }\bibinfo {title} {\emph
  {Concepts and Conditions for Error Suppression through Randomized
  Compiling}},\ \href@noop {} {\  (\bibinfo {year} {2022})},\ \Eprint
  {http://arxiv.org/abs/2212.07500} {arXiv:2212.07500}\BibitemShut {NoStop}%
\bibitem [{\citenamefont {Choi}\ \emph {et~al.}(2023)\citenamefont {Choi},
  \citenamefont {Shaw}, \citenamefont {Madjarov}, \citenamefont {Xie},
  \citenamefont {Finkelstein}, \citenamefont {Covey}, \citenamefont {Cotler},
  \citenamefont {Mark}, \citenamefont {Huang}, \citenamefont {Kale},
  \citenamefont {Pichler}, \citenamefont {Brand{\~a}o}, \citenamefont {Choi},\
  and\ \citenamefont {Endres}}]{choi_preparing_2023}%
  \BibitemOpen
  \bibfield  {author} {\bibinfo {author} {\bibfnamefont {J.}~\bibnamefont
  {Choi}}, \bibinfo {author} {\bibfnamefont {A.~L.}\ \bibnamefont {Shaw}},
  \bibinfo {author} {\bibfnamefont {I.~S.}\ \bibnamefont {Madjarov}}, \bibinfo
  {author} {\bibfnamefont {X.}~\bibnamefont {Xie}}, \bibinfo {author}
  {\bibfnamefont {R.}~\bibnamefont {Finkelstein}}, \bibinfo {author}
  {\bibfnamefont {J.~P.}\ \bibnamefont {Covey}}, \bibinfo {author}
  {\bibfnamefont {J.~S.}\ \bibnamefont {Cotler}}, \bibinfo {author}
  {\bibfnamefont {D.~K.}\ \bibnamefont {Mark}}, \bibinfo {author}
  {\bibfnamefont {H.-Y.}\ \bibnamefont {Huang}}, \bibinfo {author}
  {\bibfnamefont {A.}~\bibnamefont {Kale}}, \bibinfo {author} {\bibfnamefont
  {H.}~\bibnamefont {Pichler}}, \bibinfo {author} {\bibfnamefont {F.~G. S.~L.}\
  \bibnamefont {Brand{\~a}o}}, \bibinfo {author} {\bibfnamefont
  {S.}~\bibnamefont {Choi}}, \ and\ \bibinfo {author} {\bibfnamefont
  {M.}~\bibnamefont {Endres}},\ }\bibinfo {title} {\emph {Preparing Random
  States and Benchmarking with Many-Body Quantum Chaos}},\ \href {\doibase
  10.1038/s41586-022-05442-1} {\bibfield  {journal} {\bibinfo  {journal}
  {Nature}\ }\textbf {\bibinfo {volume} {613}},\ \bibinfo {pages} {468}
  (\bibinfo {year} {2023})}\BibitemShut {NoStop}%
\bibitem [{\citenamefont {Page}(1993)}]{page_average_1993}%
  \BibitemOpen
  \bibfield  {author} {\bibinfo {author} {\bibfnamefont {D.~N.}\ \bibnamefont
  {Page}},\ }\bibinfo {title} {\emph {Average Entropy of a Subsystem}},\ \href
  {\doibase 10.1103/PhysRevLett.71.1291} {\bibfield  {journal} {\bibinfo
  {journal} {Phys. Rev. Lett.}\ }\textbf {\bibinfo {volume} {71}},\ \bibinfo
  {pages} {1291} (\bibinfo {year} {1993})}\BibitemShut {NoStop}%
\bibitem [{\citenamefont {Nahum}\ \emph {et~al.}(2017)\citenamefont {Nahum},
  \citenamefont {Ruhman}, \citenamefont {Vijay},\ and\ \citenamefont
  {Haah}}]{nahum_quantum_2017}%
  \BibitemOpen
  \bibfield  {author} {\bibinfo {author} {\bibfnamefont {A.}~\bibnamefont
  {Nahum}}, \bibinfo {author} {\bibfnamefont {J.}~\bibnamefont {Ruhman}},
  \bibinfo {author} {\bibfnamefont {S.}~\bibnamefont {Vijay}}, \ and\ \bibinfo
  {author} {\bibfnamefont {J.}~\bibnamefont {Haah}},\ }\bibinfo {title} {\emph
  {Quantum {{Entanglement Growth}} under {{Random Unitary Dynamics}}}},\ \href
  {\doibase 10.1103/PhysRevX.7.031016} {\bibfield  {journal} {\bibinfo
  {journal} {Phys. Rev. X}\ }\textbf {\bibinfo {volume} {7}},\ \bibinfo {pages}
  {031016} (\bibinfo {year} {2017})}\BibitemShut {NoStop}%
\bibitem [{\citenamefont {Mark}\ \emph {et~al.}(2023)\citenamefont {Mark},
  \citenamefont {Choi}, \citenamefont {Shaw}, \citenamefont {Endres},\ and\
  \citenamefont {Choi}}]{Mark23}%
  \BibitemOpen
  \bibfield  {author} {\bibinfo {author} {\bibfnamefont {D.~K.}\ \bibnamefont
  {Mark}}, \bibinfo {author} {\bibfnamefont {J.}~\bibnamefont {Choi}}, \bibinfo
  {author} {\bibfnamefont {A.~L.}\ \bibnamefont {Shaw}}, \bibinfo {author}
  {\bibfnamefont {M.}~\bibnamefont {Endres}}, \ and\ \bibinfo {author}
  {\bibfnamefont {S.}~\bibnamefont {Choi}},\ }\bibinfo {title} {\emph
  {Benchmarking Quantum Simulators Using Ergodic Quantum Dynamics}},\ \href
  {\doibase 10.1103/PhysRevLett.131.110601} {\bibfield  {journal} {\bibinfo
  {journal} {Phys. Rev. Lett.}\ }\textbf {\bibinfo {volume} {131}},\ \bibinfo
  {pages} {110601} (\bibinfo {year} {2023})}\BibitemShut {NoStop}%
\bibitem [{\citenamefont {Shaw}\ \emph {et~al.}(2023)\citenamefont {Shaw},
  \citenamefont {Chen}, \citenamefont {Choi}, \citenamefont {Mark},
  \citenamefont {Scholl}, \citenamefont {Finkelstein}, \citenamefont {Elben},
  \citenamefont {Choi},\ and\ \citenamefont {Endres}}]{Shaw23}%
  \BibitemOpen
  \bibfield  {author} {\bibinfo {author} {\bibfnamefont {A.~L.}\ \bibnamefont
  {Shaw}}, \bibinfo {author} {\bibfnamefont {Z.}~\bibnamefont {Chen}}, \bibinfo
  {author} {\bibfnamefont {J.}~\bibnamefont {Choi}}, \bibinfo {author}
  {\bibfnamefont {D.~K.}\ \bibnamefont {Mark}}, \bibinfo {author}
  {\bibfnamefont {P.}~\bibnamefont {Scholl}}, \bibinfo {author} {\bibfnamefont
  {R.}~\bibnamefont {Finkelstein}}, \bibinfo {author} {\bibfnamefont
  {A.}~\bibnamefont {Elben}}, \bibinfo {author} {\bibfnamefont
  {S.}~\bibnamefont {Choi}}, \ and\ \bibinfo {author} {\bibfnamefont
  {M.}~\bibnamefont {Endres}},\ }\bibinfo {title} {\emph {Benchmarking highly
  entangled states on a 60-atom analog quantum simulator}},\ \href@noop {}
  {\bibfield  {journal} {\bibinfo  {journal} {arXiv:2308.07914}\ } (\bibinfo
  {year} {2023})}\BibitemShut {NoStop}%
\bibitem [{\citenamefont {Lai}\ and\ \citenamefont
  {Cheng}(2022)}]{lai_learning_2022}%
  \BibitemOpen
  \bibfield  {author} {\bibinfo {author} {\bibfnamefont {C.-Y.}\ \bibnamefont
  {Lai}}\ and\ \bibinfo {author} {\bibfnamefont {H.-C.}\ \bibnamefont
  {Cheng}},\ }\bibinfo {title} {\emph {Learning Quantum Circuits of Some
  \${{T}}\$ Gates}},\ \href {\doibase 10.1109/TIT.2022.3151760} {\bibfield
  {journal} {\bibinfo  {journal} {IEEE Trans. Inform. Theory}\ }\textbf
  {\bibinfo {volume} {68}},\ \bibinfo {pages} {3951} (\bibinfo {year}
  {2022})}\BibitemShut {NoStop}%
\bibitem [{\citenamefont {Arunachalam}\ \emph {et~al.}(2022)\citenamefont
  {Arunachalam}, \citenamefont {Bravyi}, \citenamefont {Nirkhe},\ and\
  \citenamefont {O'Gorman}}]{arunachalam_parameterized_2022}%
  \BibitemOpen
  \bibfield  {author} {\bibinfo {author} {\bibfnamefont {S.}~\bibnamefont
  {Arunachalam}}, \bibinfo {author} {\bibfnamefont {S.}~\bibnamefont {Bravyi}},
  \bibinfo {author} {\bibfnamefont {C.}~\bibnamefont {Nirkhe}}, \ and\ \bibinfo
  {author} {\bibfnamefont {B.}~\bibnamefont {O'Gorman}},\ }\bibinfo {title}
  {\emph {The {{Parameterized Complexity}} of {{Quantum Verification}}}},\
  \href@noop {} {\  (\bibinfo {year} {2022})},\ \Eprint
  {http://arxiv.org/abs/2202.08119} {arXiv:2202.08119}\BibitemShut {NoStop}%
\bibitem [{\citenamefont {Leone}\ \emph
  {et~al.}(2023{\natexlab{a}})\citenamefont {Leone}, \citenamefont {Oliviero},
  \citenamefont {Lloyd},\ and\ \citenamefont {Hamma}}]{leone_learning_2023-1}%
  \BibitemOpen
  \bibfield  {author} {\bibinfo {author} {\bibfnamefont {L.}~\bibnamefont
  {Leone}}, \bibinfo {author} {\bibfnamefont {S.~F.~E.}\ \bibnamefont
  {Oliviero}}, \bibinfo {author} {\bibfnamefont {S.}~\bibnamefont {Lloyd}}, \
  and\ \bibinfo {author} {\bibfnamefont {A.}~\bibnamefont {Hamma}},\ }\bibinfo
  {title} {\emph {Learning Efficient Decoders for Quasi-Chaotic Quantum
  Scramblers}},\ \href@noop {} {\  (\bibinfo {year} {2023}{\natexlab{a}})},\
  \Eprint {http://arxiv.org/abs/2212.11338} {arXiv:2212.11338}\BibitemShut
  {NoStop}%
\bibitem [{\citenamefont {Guta}\ \emph {et~al.}(2018)\citenamefont {Guta},
  \citenamefont {Kahn}, \citenamefont {Kueng},\ and\ \citenamefont
  {Tropp}}]{guta_fast_2018}%
  \BibitemOpen
  \bibfield  {author} {\bibinfo {author} {\bibfnamefont {M.}~\bibnamefont
  {Guta}}, \bibinfo {author} {\bibfnamefont {J.}~\bibnamefont {Kahn}}, \bibinfo
  {author} {\bibfnamefont {R.}~\bibnamefont {Kueng}}, \ and\ \bibinfo {author}
  {\bibfnamefont {J.~A.}\ \bibnamefont {Tropp}},\ }\bibinfo {title} {\emph
  {Fast State Tomography with Optimal Error Bounds}},\ \href@noop {} {\
  (\bibinfo {year} {2018})},\ \Eprint {http://arxiv.org/abs/1809.11162}
  {arXiv:1809.11162}\BibitemShut {NoStop}%
\bibitem [{\citenamefont {Bravyi}\ \emph {et~al.}(2019)\citenamefont {Bravyi},
  \citenamefont {Browne}, \citenamefont {Calpin}, \citenamefont {Campbell},
  \citenamefont {Gosset},\ and\ \citenamefont
  {Howard}}]{bravyi_simulation_2019-1}%
  \BibitemOpen
  \bibfield  {author} {\bibinfo {author} {\bibfnamefont {S.}~\bibnamefont
  {Bravyi}}, \bibinfo {author} {\bibfnamefont {D.}~\bibnamefont {Browne}},
  \bibinfo {author} {\bibfnamefont {P.}~\bibnamefont {Calpin}}, \bibinfo
  {author} {\bibfnamefont {E.}~\bibnamefont {Campbell}}, \bibinfo {author}
  {\bibfnamefont {D.}~\bibnamefont {Gosset}}, \ and\ \bibinfo {author}
  {\bibfnamefont {M.}~\bibnamefont {Howard}},\ }\bibinfo {title} {\emph
  {Simulation of Quantum Circuits by Low-Rank Stabilizer Decompositions}},\
  \href {\doibase 10.22331/q-2019-09-02-181} {\bibfield  {journal} {\bibinfo
  {journal} {Quantum}\ }\textbf {\bibinfo {volume} {3}},\ \bibinfo {pages}
  {181} (\bibinfo {year} {2019})}\BibitemShut {NoStop}%
\bibitem [{\citenamefont {Pashayan}\ \emph {et~al.}(2022)\citenamefont
  {Pashayan}, \citenamefont {{Reardon-Smith}}, \citenamefont {Korzekwa},\ and\
  \citenamefont {Bartlett}}]{pashayan_fast_2022}%
  \BibitemOpen
  \bibfield  {author} {\bibinfo {author} {\bibfnamefont {H.}~\bibnamefont
  {Pashayan}}, \bibinfo {author} {\bibfnamefont {O.}~\bibnamefont
  {{Reardon-Smith}}}, \bibinfo {author} {\bibfnamefont {K.}~\bibnamefont
  {Korzekwa}}, \ and\ \bibinfo {author} {\bibfnamefont {S.~D.}\ \bibnamefont
  {Bartlett}},\ }\bibinfo {title} {\emph {Fast {{Estimation}} of {{Outcome
  Probabilities}} for {{Quantum Circuits}}}},\ \href {\doibase
  10.1103/PRXQuantum.3.020361} {\bibfield  {journal} {\bibinfo  {journal} {PRX
  Quantum}\ }\textbf {\bibinfo {volume} {3}},\ \bibinfo {pages} {020361}
  (\bibinfo {year} {2022})}\BibitemShut {NoStop}%
\bibitem [{\citenamefont {Campbell}\ and\ \citenamefont
  {Howard}(2017)}]{campbell_unified_2017}%
  \BibitemOpen
  \bibfield  {author} {\bibinfo {author} {\bibfnamefont {E.~T.}\ \bibnamefont
  {Campbell}}\ and\ \bibinfo {author} {\bibfnamefont {M.}~\bibnamefont
  {Howard}},\ }\bibinfo {title} {\emph {A Unified Framework for Magic State
  Distillation and Multi-Qubit Gate-Synthesis with Reduced Resource Cost}},\
  \href {\doibase 10.1103/PhysRevA.95.022316} {\bibfield  {journal} {\bibinfo
  {journal} {Phys. Rev. A}\ }\textbf {\bibinfo {volume} {95}},\ \bibinfo
  {pages} {022316} (\bibinfo {year} {2017})}\BibitemShut {NoStop}%
\bibitem [{\citenamefont {Barenco}\ \emph {et~al.}(1996)\citenamefont
  {Barenco}, \citenamefont {Berthiaume}, \citenamefont {Deutsch}, \citenamefont
  {Ekert}, \citenamefont {Jozsa},\ and\ \citenamefont
  {Macchiavello}}]{barenco_stabilisation_1996}%
  \BibitemOpen
  \bibfield  {author} {\bibinfo {author} {\bibfnamefont {A.}~\bibnamefont
  {Barenco}}, \bibinfo {author} {\bibfnamefont {A.}~\bibnamefont {Berthiaume}},
  \bibinfo {author} {\bibfnamefont {D.}~\bibnamefont {Deutsch}}, \bibinfo
  {author} {\bibfnamefont {A.}~\bibnamefont {Ekert}}, \bibinfo {author}
  {\bibfnamefont {R.}~\bibnamefont {Jozsa}}, \ and\ \bibinfo {author}
  {\bibfnamefont {C.}~\bibnamefont {Macchiavello}},\ }\bibinfo {title} {\emph
  {Stabilisation of {{Quantum Computations}} by {{Symmetrisation}}}},\
  \href@noop {} {\  (\bibinfo {year} {1996})},\ \Eprint
  {http://arxiv.org/abs/quant-ph/9604028} {arXiv:quant-ph/9604028}\BibitemShut
  {NoStop}%
\bibitem [{\citenamefont {Cotler}\ \emph {et~al.}(2019)\citenamefont {Cotler},
  \citenamefont {Choi}, \citenamefont {Lukin}, \citenamefont {Gharibyan},
  \citenamefont {Grover}, \citenamefont {Tai}, \citenamefont {Rispoli},
  \citenamefont {Schittko}, \citenamefont {Preiss}, \citenamefont {Kaufman},
  \citenamefont {Greiner}, \citenamefont {Pichler},\ and\ \citenamefont
  {Hayden}}]{cotler_quantum_2019}%
  \BibitemOpen
  \bibfield  {author} {\bibinfo {author} {\bibfnamefont {J.}~\bibnamefont
  {Cotler}}, \bibinfo {author} {\bibfnamefont {S.}~\bibnamefont {Choi}},
  \bibinfo {author} {\bibfnamefont {A.}~\bibnamefont {Lukin}}, \bibinfo
  {author} {\bibfnamefont {H.}~\bibnamefont {Gharibyan}}, \bibinfo {author}
  {\bibfnamefont {T.}~\bibnamefont {Grover}}, \bibinfo {author} {\bibfnamefont
  {M.~E.}\ \bibnamefont {Tai}}, \bibinfo {author} {\bibfnamefont
  {M.}~\bibnamefont {Rispoli}}, \bibinfo {author} {\bibfnamefont
  {R.}~\bibnamefont {Schittko}}, \bibinfo {author} {\bibfnamefont {P.~M.}\
  \bibnamefont {Preiss}}, \bibinfo {author} {\bibfnamefont {A.~M.}\
  \bibnamefont {Kaufman}}, \bibinfo {author} {\bibfnamefont {M.}~\bibnamefont
  {Greiner}}, \bibinfo {author} {\bibfnamefont {H.}~\bibnamefont {Pichler}}, \
  and\ \bibinfo {author} {\bibfnamefont {P.}~\bibnamefont {Hayden}},\ }\bibinfo
  {title} {\emph {Quantum {{Virtual Cooling}}}},\ \href {\doibase
  10.1103/PhysRevX.9.031013} {\bibfield  {journal} {\bibinfo  {journal} {Phys.
  Rev. X}\ }\textbf {\bibinfo {volume} {9}},\ \bibinfo {pages} {031013}
  (\bibinfo {year} {2019})}\BibitemShut {NoStop}%
\bibitem [{\citenamefont {Koczor}(2021)}]{koczor_exponential_2021}%
  \BibitemOpen
  \bibfield  {author} {\bibinfo {author} {\bibfnamefont {B.}~\bibnamefont
  {Koczor}},\ }\bibinfo {title} {\emph {Exponential {{Error Suppression}} for
  {{Near-Term Quantum Devices}}}},\ \href {\doibase 10.1103/PhysRevX.11.031057}
  {\bibfield  {journal} {\bibinfo  {journal} {Phys. Rev. X}\ }\textbf {\bibinfo
  {volume} {11}},\ \bibinfo {pages} {031057} (\bibinfo {year}
  {2021})}\BibitemShut {NoStop}%
\bibitem [{\citenamefont {Huggins}\ \emph {et~al.}(2021)\citenamefont
  {Huggins}, \citenamefont {McArdle}, \citenamefont {O'Brien}, \citenamefont
  {Lee}, \citenamefont {Rubin}, \citenamefont {Boixo}, \citenamefont {Whaley},
  \citenamefont {Babbush},\ and\ \citenamefont
  {McClean}}]{huggins_virtual_2021}%
  \BibitemOpen
  \bibfield  {author} {\bibinfo {author} {\bibfnamefont {W.~J.}\ \bibnamefont
  {Huggins}}, \bibinfo {author} {\bibfnamefont {S.}~\bibnamefont {McArdle}},
  \bibinfo {author} {\bibfnamefont {T.~E.}\ \bibnamefont {O'Brien}}, \bibinfo
  {author} {\bibfnamefont {J.}~\bibnamefont {Lee}}, \bibinfo {author}
  {\bibfnamefont {N.~C.}\ \bibnamefont {Rubin}}, \bibinfo {author}
  {\bibfnamefont {S.}~\bibnamefont {Boixo}}, \bibinfo {author} {\bibfnamefont
  {K.~B.}\ \bibnamefont {Whaley}}, \bibinfo {author} {\bibfnamefont
  {R.}~\bibnamefont {Babbush}}, \ and\ \bibinfo {author} {\bibfnamefont
  {J.~R.}\ \bibnamefont {McClean}},\ }\bibinfo {title} {\emph {Virtual
  {{Distillation}} for {{Quantum Error Mitigation}}}},\ \href {\doibase
  10.1103/PhysRevX.11.041036} {\bibfield  {journal} {\bibinfo  {journal} {Phys.
  Rev. X}\ }\textbf {\bibinfo {volume} {11}},\ \bibinfo {pages} {041036}
  (\bibinfo {year} {2021})}\BibitemShut {NoStop}%
\bibitem [{\citenamefont {Bermudez}\ \emph {et~al.}(2017)\citenamefont
  {Bermudez}, \citenamefont {Xu}, \citenamefont {Nigmatullin}, \citenamefont
  {O'Gorman}, \citenamefont {Negnevitsky}, \citenamefont {Schindler},
  \citenamefont {Monz}, \citenamefont {Poschinger}, \citenamefont {Hempel},
  \citenamefont {Home}, \citenamefont {{Schmidt-Kaler}}, \citenamefont
  {Biercuk}, \citenamefont {Blatt}, \citenamefont {Benjamin},\ and\
  \citenamefont {M{\"u}ller}}]{bermudez_assessing_2017}%
  \BibitemOpen
  \bibfield  {author} {\bibinfo {author} {\bibfnamefont {A.}~\bibnamefont
  {Bermudez}}, \bibinfo {author} {\bibfnamefont {X.}~\bibnamefont {Xu}},
  \bibinfo {author} {\bibfnamefont {R.}~\bibnamefont {Nigmatullin}}, \bibinfo
  {author} {\bibfnamefont {J.}~\bibnamefont {O'Gorman}}, \bibinfo {author}
  {\bibfnamefont {V.}~\bibnamefont {Negnevitsky}}, \bibinfo {author}
  {\bibfnamefont {P.}~\bibnamefont {Schindler}}, \bibinfo {author}
  {\bibfnamefont {T.}~\bibnamefont {Monz}}, \bibinfo {author} {\bibfnamefont
  {U.~G.}\ \bibnamefont {Poschinger}}, \bibinfo {author} {\bibfnamefont
  {C.}~\bibnamefont {Hempel}}, \bibinfo {author} {\bibfnamefont
  {J.}~\bibnamefont {Home}}, \bibinfo {author} {\bibfnamefont {F.}~\bibnamefont
  {{Schmidt-Kaler}}}, \bibinfo {author} {\bibfnamefont {M.}~\bibnamefont
  {Biercuk}}, \bibinfo {author} {\bibfnamefont {R.}~\bibnamefont {Blatt}},
  \bibinfo {author} {\bibfnamefont {S.}~\bibnamefont {Benjamin}}, \ and\
  \bibinfo {author} {\bibfnamefont {M.}~\bibnamefont {M{\"u}ller}},\ }\bibinfo
  {title} {\emph {Assessing the {{Progress}} of {{Trapped-Ion Processors
  Towards Fault-Tolerant Quantum Computation}}}},\ \href {\doibase
  10.1103/PhysRevX.7.041061} {\bibfield  {journal} {\bibinfo  {journal} {Phys.
  Rev. X}\ }\textbf {\bibinfo {volume} {7}},\ \bibinfo {pages} {041061}
  (\bibinfo {year} {2017})}\BibitemShut {NoStop}%
\bibitem [{\citenamefont {Bluvstein}\ \emph {et~al.}(2022)\citenamefont
  {Bluvstein}, \citenamefont {Levine}, \citenamefont {Semeghini}, \citenamefont
  {Wang}, \citenamefont {Ebadi}, \citenamefont {Kalinowski}, \citenamefont
  {Keesling}, \citenamefont {Maskara}, \citenamefont {Pichler}, \citenamefont
  {Greiner}, \citenamefont {Vuleti{\'c}},\ and\ \citenamefont
  {Lukin}}]{bluvstein_quantum_2022}%
  \BibitemOpen
  \bibfield  {author} {\bibinfo {author} {\bibfnamefont {D.}~\bibnamefont
  {Bluvstein}}, \bibinfo {author} {\bibfnamefont {H.}~\bibnamefont {Levine}},
  \bibinfo {author} {\bibfnamefont {G.}~\bibnamefont {Semeghini}}, \bibinfo
  {author} {\bibfnamefont {T.~T.}\ \bibnamefont {Wang}}, \bibinfo {author}
  {\bibfnamefont {S.}~\bibnamefont {Ebadi}}, \bibinfo {author} {\bibfnamefont
  {M.}~\bibnamefont {Kalinowski}}, \bibinfo {author} {\bibfnamefont
  {A.}~\bibnamefont {Keesling}}, \bibinfo {author} {\bibfnamefont
  {N.}~\bibnamefont {Maskara}}, \bibinfo {author} {\bibfnamefont
  {H.}~\bibnamefont {Pichler}}, \bibinfo {author} {\bibfnamefont
  {M.}~\bibnamefont {Greiner}}, \bibinfo {author} {\bibfnamefont
  {V.}~\bibnamefont {Vuleti{\'c}}}, \ and\ \bibinfo {author} {\bibfnamefont
  {M.~D.}\ \bibnamefont {Lukin}},\ }\bibinfo {title} {\emph {A Quantum
  Processor Based on Coherent Transport of Entangled Atom Arrays}},\ \href
  {\doibase 10.1038/s41586-022-04592-6} {\bibfield  {journal} {\bibinfo
  {journal} {Nature}\ }\textbf {\bibinfo {volume} {604}},\ \bibinfo {pages}
  {451} (\bibinfo {year} {2022})}\BibitemShut {NoStop}%
\bibitem [{\citenamefont {Cai}\ \emph {et~al.}(2022)\citenamefont {Cai},
  \citenamefont {Siegel},\ and\ \citenamefont {Benjamin}}]{Cai2022}%
  \BibitemOpen
  \bibfield  {author} {\bibinfo {author} {\bibfnamefont {Z.}~\bibnamefont
  {Cai}}, \bibinfo {author} {\bibfnamefont {A.}~\bibnamefont {Siegel}}, \ and\
  \bibinfo {author} {\bibfnamefont {S.}~\bibnamefont {Benjamin}},\ }\bibinfo
  {title} {\emph {{Looped Pipelines Enabling Effective 3D Qubit Lattices in a
  Strictly 2D Device}}},\ \href@noop {} {\  (\bibinfo {year} {2022})},\ \Eprint
  {http://arxiv.org/abs/2203.13123} {arXiv:2203.13123}\BibitemShut {NoStop}%
\bibitem [{\citenamefont {Bluvstein}\ \emph {et~al.}(2024)\citenamefont
  {Bluvstein}, \citenamefont {Evered}, \citenamefont {Geim}, \citenamefont
  {Li}, \citenamefont {Zhou}, \citenamefont {Manovitz}, \citenamefont {Ebadi},
  \citenamefont {Cain}, \citenamefont {Kalinowski}, \citenamefont {Hangleiter},
  \citenamefont {Ataides}, \citenamefont {Maskara}, \citenamefont {Cong},
  \citenamefont {Gao}, \citenamefont {Rodriguez}, \citenamefont {Karolyshyn},
  \citenamefont {Semeghini}, \citenamefont {Gullans}, \citenamefont {Greiner},
  \citenamefont {Vuleti{\'c}},\ and\ \citenamefont
  {Lukin}}]{bluvstein_logical_2023}%
  \BibitemOpen
  \bibfield  {author} {\bibinfo {author} {\bibfnamefont {D.}~\bibnamefont
  {Bluvstein}}, \bibinfo {author} {\bibfnamefont {S.~J.}\ \bibnamefont
  {Evered}}, \bibinfo {author} {\bibfnamefont {A.~A.}\ \bibnamefont {Geim}},
  \bibinfo {author} {\bibfnamefont {S.~H.}\ \bibnamefont {Li}}, \bibinfo
  {author} {\bibfnamefont {H.}~\bibnamefont {Zhou}}, \bibinfo {author}
  {\bibfnamefont {T.}~\bibnamefont {Manovitz}}, \bibinfo {author}
  {\bibfnamefont {S.}~\bibnamefont {Ebadi}}, \bibinfo {author} {\bibfnamefont
  {M.}~\bibnamefont {Cain}}, \bibinfo {author} {\bibfnamefont {M.}~\bibnamefont
  {Kalinowski}}, \bibinfo {author} {\bibfnamefont {D.}~\bibnamefont
  {Hangleiter}}, \bibinfo {author} {\bibfnamefont {J.~P.~B.}\ \bibnamefont
  {Ataides}}, \bibinfo {author} {\bibfnamefont {N.}~\bibnamefont {Maskara}},
  \bibinfo {author} {\bibfnamefont {I.}~\bibnamefont {Cong}}, \bibinfo {author}
  {\bibfnamefont {X.}~\bibnamefont {Gao}}, \bibinfo {author} {\bibfnamefont
  {P.~S.}\ \bibnamefont {Rodriguez}}, \bibinfo {author} {\bibfnamefont
  {T.}~\bibnamefont {Karolyshyn}}, \bibinfo {author} {\bibfnamefont
  {G.}~\bibnamefont {Semeghini}}, \bibinfo {author} {\bibfnamefont {M.~J.}\
  \bibnamefont {Gullans}}, \bibinfo {author} {\bibfnamefont {M.}~\bibnamefont
  {Greiner}}, \bibinfo {author} {\bibfnamefont {V.}~\bibnamefont
  {Vuleti{\'c}}}, \ and\ \bibinfo {author} {\bibfnamefont {M.~D.}\ \bibnamefont
  {Lukin}},\ }\bibinfo {title} {\emph {Logical Quantum Processor Based on
  Reconfigurable Atom Arrays}},\ \href {\doibase 10.1038/s41586-023-06927-3}
  {\bibfield  {journal} {\bibinfo  {journal} {Nature}\ }\textbf {\bibinfo
  {volume} {626}},\ \bibinfo {pages} {58} (\bibinfo {year} {2024})}\BibitemShut
  {NoStop}%
\bibitem [{\citenamefont {Gao}\ and\ \citenamefont
  {Duan}(2018)}]{gao_efficient_2018}%
  \BibitemOpen
  \bibfield  {author} {\bibinfo {author} {\bibfnamefont {X.}~\bibnamefont
  {Gao}}\ and\ \bibinfo {author} {\bibfnamefont {L.}~\bibnamefont {Duan}},\
  }\bibinfo {title} {\emph {Efficient Classical Simulation of Noisy Quantum
  Computation}},\ \href@noop {} {\  (\bibinfo {year} {2018})},\ \Eprint
  {http://arxiv.org/abs/1810.03176} {arXiv:1810.03176}\BibitemShut {NoStop}%
\bibitem [{\citenamefont {Aharonov}\ \emph {et~al.}(2022)\citenamefont
  {Aharonov}, \citenamefont {Gao}, \citenamefont {Landau}, \citenamefont
  {Liu},\ and\ \citenamefont {Vazirani}}]{aharonov_polynomial-time_2022}%
  \BibitemOpen
  \bibfield  {author} {\bibinfo {author} {\bibfnamefont {D.}~\bibnamefont
  {Aharonov}}, \bibinfo {author} {\bibfnamefont {X.}~\bibnamefont {Gao}},
  \bibinfo {author} {\bibfnamefont {Z.}~\bibnamefont {Landau}}, \bibinfo
  {author} {\bibfnamefont {Y.}~\bibnamefont {Liu}}, \ and\ \bibinfo {author}
  {\bibfnamefont {U.}~\bibnamefont {Vazirani}},\ }\bibinfo {title} {\emph {A
  Polynomial-Time Classical Algorithm for Noisy Random Circuit Sampling}},\
  \href@noop {} {\  (\bibinfo {year} {2022})},\ \Eprint
  {http://arxiv.org/abs/2211.03999} {arXiv:2211.03999}\BibitemShut {NoStop}%
\bibitem [{\citenamefont {Grewal}\ \emph
  {et~al.}(2023{\natexlab{a}})\citenamefont {Grewal}, \citenamefont {Iyer},
  \citenamefont {Kretschmer},\ and\ \citenamefont
  {Liang}}]{grewal_efficient_2023-2}%
  \BibitemOpen
  \bibfield  {author} {\bibinfo {author} {\bibfnamefont {S.}~\bibnamefont
  {Grewal}}, \bibinfo {author} {\bibfnamefont {V.}~\bibnamefont {Iyer}},
  \bibinfo {author} {\bibfnamefont {W.}~\bibnamefont {Kretschmer}}, \ and\
  \bibinfo {author} {\bibfnamefont {D.}~\bibnamefont {Liang}},\ }\bibinfo
  {title} {\emph {Efficient {{Learning}} of {{Quantum States Prepared With Few
  Non-Clifford Gates}}}},\ \href@noop {} {\  (\bibinfo {year}
  {2023}{\natexlab{a}})},\ \Eprint {http://arxiv.org/abs/2305.13409}
  {arXiv:2305.13409}\BibitemShut {NoStop}%
\bibitem [{\citenamefont {Leone}\ \emph
  {et~al.}(2023{\natexlab{b}})\citenamefont {Leone}, \citenamefont {Oliviero},\
  and\ \citenamefont {Hamma}}]{leone_learning_2023}%
  \BibitemOpen
  \bibfield  {author} {\bibinfo {author} {\bibfnamefont {L.}~\bibnamefont
  {Leone}}, \bibinfo {author} {\bibfnamefont {S.~F.~E.}\ \bibnamefont
  {Oliviero}}, \ and\ \bibinfo {author} {\bibfnamefont {A.}~\bibnamefont
  {Hamma}},\ }\bibinfo {title} {\emph {Learning T-Doped Stabilizer States}},\
  \href@noop {} {\  (\bibinfo {year} {2023}{\natexlab{b}})},\ \Eprint
  {http://arxiv.org/abs/2305.15398} {arXiv:2305.15398}\BibitemShut {NoStop}%
\bibitem [{\citenamefont {Brand{\~a}o}\ \emph {et~al.}(2016)\citenamefont
  {Brand{\~a}o}, \citenamefont {Harrow},\ and\ \citenamefont
  {Horodecki}}]{brandao_local_2016}%
  \BibitemOpen
  \bibfield  {author} {\bibinfo {author} {\bibfnamefont {F.~G. S.~L.}\
  \bibnamefont {Brand{\~a}o}}, \bibinfo {author} {\bibfnamefont {A.~W.}\
  \bibnamefont {Harrow}}, \ and\ \bibinfo {author} {\bibfnamefont
  {M.}~\bibnamefont {Horodecki}},\ }\bibinfo {title} {\emph {Local {{Random
  Quantum Circuits}} Are {{Approximate Polynomial-Designs}}}},\ \href {\doibase
  10.1007/s00220-016-2706-8} {\bibfield  {journal} {\bibinfo  {journal}
  {Commun. Math. Phys.}\ }\textbf {\bibinfo {volume} {346}},\ \bibinfo {pages}
  {397} (\bibinfo {year} {2016})}\BibitemShut {NoStop}%
\bibitem [{\citenamefont {Zhou}\ and\ \citenamefont
  {Nahum}(2019)}]{zhou_emergent_2019}%
  \BibitemOpen
  \bibfield  {author} {\bibinfo {author} {\bibfnamefont {T.}~\bibnamefont
  {Zhou}}\ and\ \bibinfo {author} {\bibfnamefont {A.}~\bibnamefont {Nahum}},\
  }\bibinfo {title} {\emph {Emergent Statistical Mechanics of Entanglement in
  Random Unitary Circuits}},\ \href {\doibase 10.1103/PhysRevB.99.174205}
  {\bibfield  {journal} {\bibinfo  {journal} {Phys. Rev. B}\ }\textbf {\bibinfo
  {volume} {99}},\ \bibinfo {pages} {174205} (\bibinfo {year}
  {2019})}\BibitemShut {NoStop}%
\bibitem [{\citenamefont {{Hunter-Jones}}(2019)}]{hunter-jones_unitary_2019}%
  \BibitemOpen
  \bibfield  {author} {\bibinfo {author} {\bibfnamefont {N.}~\bibnamefont
  {{Hunter-Jones}}},\ }\bibinfo {title} {\emph {Unitary Designs from
  Statistical Mechanics in Random Quantum Circuits}},\ \href@noop {} {\
  (\bibinfo {year} {2019})},\ \Eprint {http://arxiv.org/abs/1905.12053}
  {arXiv:1905.12053}\BibitemShut {NoStop}%
\bibitem [{\citenamefont {Barak}\ \emph {et~al.}(2021)\citenamefont {Barak},
  \citenamefont {Chou},\ and\ \citenamefont {Gao}}]{barak_spoofing_2021}%
  \BibitemOpen
  \bibfield  {author} {\bibinfo {author} {\bibfnamefont {B.}~\bibnamefont
  {Barak}}, \bibinfo {author} {\bibfnamefont {C.-N.}\ \bibnamefont {Chou}}, \
  and\ \bibinfo {author} {\bibfnamefont {X.}~\bibnamefont {Gao}},\ }\bibfield
  {title} {\emph {\bibinfo {title} {\emph {Spoofing {{Linear Cross-Entropy
  Benchmarking}} in {{Shallow Quantum Circuits}}}},\ }}in\ \href {\doibase
  10.4230/LIPIcs.ITCS.2021.30} {\emph {\bibinfo {booktitle} {12th
  {{Innovations}} in {{Theoretical Computer Science Conference}} ({{ITCS}}
  2021)}}},\ \bibinfo {series} {Leibniz {{International Proceedings}} in
  {{Informatics}} ({{LIPIcs}})}, Vol.\ \bibinfo {volume} {185},\ \bibinfo
  {editor} {edited by\ \bibinfo {editor} {\bibfnamefont {J.~R.}\ \bibnamefont
  {Lee}}}\ (\bibinfo  {publisher} {{Schloss Dagstuhl\textendash Leibniz-Zentrum
  f\"ur Informatik}},\ \bibinfo {address} {{Dagstuhl, Germany}},\ \bibinfo
  {year} {2021})\ pp.\ \bibinfo {pages} {30:1--30:20},\ \Eprint
  {http://arxiv.org/abs/2005.02421} {arXiv:2005.02421}\BibitemShut {NoStop}%
\bibitem [{\citenamefont {Dalzell}\ \emph {et~al.}(2022)\citenamefont
  {Dalzell}, \citenamefont {{Hunter-Jones}},\ and\ \citenamefont
  {Brand{\~a}o}}]{dalzell_random_2022}%
  \BibitemOpen
  \bibfield  {author} {\bibinfo {author} {\bibfnamefont {A.~M.}\ \bibnamefont
  {Dalzell}}, \bibinfo {author} {\bibfnamefont {N.}~\bibnamefont
  {{Hunter-Jones}}}, \ and\ \bibinfo {author} {\bibfnamefont {F.~G. S.~L.}\
  \bibnamefont {Brand{\~a}o}},\ }\bibinfo {title} {\emph {Random {{Quantum
  Circuits Anticoncentrate}} in {{Log Depth}}}},\ \href {\doibase
  10.1103/PRXQuantum.3.010333} {\bibfield  {journal} {\bibinfo  {journal} {PRX
  Quantum}\ }\textbf {\bibinfo {volume} {3}},\ \bibinfo {pages} {010333}
  (\bibinfo {year} {2022})}\BibitemShut {NoStop}%
\bibitem [{\citenamefont {Dalzell}\ \emph {et~al.}(2021)\citenamefont
  {Dalzell}, \citenamefont {{Hunter-Jones}},\ and\ \citenamefont
  {Brand{\~a}o}}]{dalzell_random_2021}%
  \BibitemOpen
  \bibfield  {author} {\bibinfo {author} {\bibfnamefont {A.~M.}\ \bibnamefont
  {Dalzell}}, \bibinfo {author} {\bibfnamefont {N.}~\bibnamefont
  {{Hunter-Jones}}}, \ and\ \bibinfo {author} {\bibfnamefont {F.~G. S.~L.}\
  \bibnamefont {Brand{\~a}o}},\ }\bibinfo {title} {\emph {Random Quantum
  Circuits Transform Local Noise into Global White Noise}},\ \href@noop {} {\
  (\bibinfo {year} {2021})},\ \Eprint {http://arxiv.org/abs/2111.14907}
  {arXiv:2111.14907}\BibitemShut {NoStop}%
\bibitem [{\citenamefont {Ringbauer}\ \emph {et~al.}(2023)\citenamefont
  {Ringbauer}, \citenamefont {Hinsche}, \citenamefont {Feldker}, \citenamefont
  {Faehrmann}, \citenamefont {{Bermejo-Vega}}, \citenamefont {Edmunds},
  \citenamefont {Postler}, \citenamefont {Stricker}, \citenamefont {Marciniak},
  \citenamefont {Meth}, \citenamefont {Pogorelov}, \citenamefont {Blatt},
  \citenamefont {Schindler}, \citenamefont {Eisert}, \citenamefont {Monz},\
  and\ \citenamefont {Hangleiter}}]{ringbauer_verifiable_2023}%
  \BibitemOpen
  \bibfield  {author} {\bibinfo {author} {\bibfnamefont {M.}~\bibnamefont
  {Ringbauer}}, \bibinfo {author} {\bibfnamefont {M.}~\bibnamefont {Hinsche}},
  \bibinfo {author} {\bibfnamefont {T.}~\bibnamefont {Feldker}}, \bibinfo
  {author} {\bibfnamefont {P.~K.}\ \bibnamefont {Faehrmann}}, \bibinfo {author}
  {\bibfnamefont {J.}~\bibnamefont {{Bermejo-Vega}}}, \bibinfo {author}
  {\bibfnamefont {C.}~\bibnamefont {Edmunds}}, \bibinfo {author} {\bibfnamefont
  {L.}~\bibnamefont {Postler}}, \bibinfo {author} {\bibfnamefont
  {R.}~\bibnamefont {Stricker}}, \bibinfo {author} {\bibfnamefont {C.~D.}\
  \bibnamefont {Marciniak}}, \bibinfo {author} {\bibfnamefont {M.}~\bibnamefont
  {Meth}}, \bibinfo {author} {\bibfnamefont {I.}~\bibnamefont {Pogorelov}},
  \bibinfo {author} {\bibfnamefont {R.}~\bibnamefont {Blatt}}, \bibinfo
  {author} {\bibfnamefont {P.}~\bibnamefont {Schindler}}, \bibinfo {author}
  {\bibfnamefont {J.}~\bibnamefont {Eisert}}, \bibinfo {author} {\bibfnamefont
  {T.}~\bibnamefont {Monz}}, \ and\ \bibinfo {author} {\bibfnamefont
  {D.}~\bibnamefont {Hangleiter}},\ }\bibinfo {title} {\emph {Verifiable
  Measurement-Based Quantum Random Sampling with Trapped Ions}},\ \href@noop {}
  {\  (\bibinfo {year} {2023})},\ \Eprint {http://arxiv.org/abs/2307.14424}
  {arXiv:2307.14424}\BibitemShut {NoStop}%
\bibitem [{\citenamefont {Flammia}\ and\ \citenamefont
  {Liu}(2011)}]{flammia_direct_2011}%
  \BibitemOpen
  \bibfield  {author} {\bibinfo {author} {\bibfnamefont {S.~T.}\ \bibnamefont
  {Flammia}}\ and\ \bibinfo {author} {\bibfnamefont {Y.-K.}\ \bibnamefont
  {Liu}},\ }\bibinfo {title} {\emph {Direct {{Fidelity Estimation}} from {{Few
  Pauli Measurements}}}},\ \href {\doibase 10.1103/PhysRevLett.106.230501}
  {\bibfield  {journal} {\bibinfo  {journal} {Phys. Rev. Lett.}\ }\textbf
  {\bibinfo {volume} {106}},\ \bibinfo {pages} {230501} (\bibinfo {year}
  {2011})}\BibitemShut {NoStop}%
\bibitem [{\citenamefont {Fitzsimons}\ and\ \citenamefont
  {Kashefi}(2017)}]{fitzsimons_unconditionally_2017}%
  \BibitemOpen
  \bibfield  {author} {\bibinfo {author} {\bibfnamefont {J.~F.}\ \bibnamefont
  {Fitzsimons}}\ and\ \bibinfo {author} {\bibfnamefont {E.}~\bibnamefont
  {Kashefi}},\ }\bibinfo {title} {\emph {Unconditionally Verifiable Blind
  Quantum Computation}},\ \href {\doibase 10.1103/PhysRevA.96.012303}
  {\bibfield  {journal} {\bibinfo  {journal} {Phys. Rev. A}\ }\textbf {\bibinfo
  {volume} {96}},\ \bibinfo {pages} {012303} (\bibinfo {year}
  {2017})}\BibitemShut {NoStop}%
\bibitem [{\citenamefont {Reichardt}\ \emph {et~al.}(2013)\citenamefont
  {Reichardt}, \citenamefont {Unger},\ and\ \citenamefont
  {Vazirani}}]{reichardt_classical_2013}%
  \BibitemOpen
  \bibfield  {author} {\bibinfo {author} {\bibfnamefont {B.~W.}\ \bibnamefont
  {Reichardt}}, \bibinfo {author} {\bibfnamefont {F.}~\bibnamefont {Unger}}, \
  and\ \bibinfo {author} {\bibfnamefont {U.}~\bibnamefont {Vazirani}},\
  }\bibinfo {title} {\emph {Classical Command of Quantum Systems}},\ \href
  {\doibase 10.1038/nature12035} {\bibfield  {journal} {\bibinfo  {journal}
  {Nature}\ }\textbf {\bibinfo {volume} {496}},\ \bibinfo {pages} {456}
  (\bibinfo {year} {2013})}\BibitemShut {NoStop}%
\bibitem [{\citenamefont {Mahadev}(2018)}]{mahadev_classical_2018-1}%
  \BibitemOpen
  \bibfield  {author} {\bibinfo {author} {\bibfnamefont {U.}~\bibnamefont
  {Mahadev}},\ }\bibfield  {title} {\emph {\bibinfo {title} {\emph {Classical
  {{Verification}} of {{Quantum Computations}}}},\ }}in\ \href {\doibase
  10.1109/FOCS.2018.00033} {\emph {\bibinfo {booktitle} {2018 {{IEEE}} 59th
  {{Annual Symposium}} on {{Foundations}} of {{Computer Science}}
  ({{FOCS}})}}}\ (\bibinfo {year} {2018})\ pp.\ \bibinfo {pages}
  {259--267}\BibitemShut {NoStop}%
\bibitem [{\citenamefont {Garnerone}\ \emph {et~al.}(2010)\citenamefont
  {Garnerone}, \citenamefont {{de Oliveira}},\ and\ \citenamefont
  {Zanardi}}]{garnerone_typicality_2010}%
  \BibitemOpen
  \bibfield  {author} {\bibinfo {author} {\bibfnamefont {S.}~\bibnamefont
  {Garnerone}}, \bibinfo {author} {\bibfnamefont {T.~R.}\ \bibnamefont {{de
  Oliveira}}}, \ and\ \bibinfo {author} {\bibfnamefont {P.}~\bibnamefont
  {Zanardi}},\ }\bibinfo {title} {\emph {Typicality in Random Matrix Product
  States}},\ \href {\doibase 10.1103/PhysRevA.81.032336} {\bibfield  {journal}
  {\bibinfo  {journal} {Phys. Rev. A}\ }\textbf {\bibinfo {volume} {81}},\
  \bibinfo {pages} {032336} (\bibinfo {year} {2010})}\BibitemShut {NoStop}%
\bibitem [{\citenamefont {Collins}\ \emph {et~al.}(2013)\citenamefont
  {Collins}, \citenamefont {{Gonzalez-Guillen}},\ and\ \citenamefont
  {{Perez-Garcia}}}]{collins_matrix_2013}%
  \BibitemOpen
  \bibfield  {author} {\bibinfo {author} {\bibfnamefont {B.}~\bibnamefont
  {Collins}}, \bibinfo {author} {\bibfnamefont {C.~E.}\ \bibnamefont
  {{Gonzalez-Guillen}}}, \ and\ \bibinfo {author} {\bibfnamefont
  {D.}~\bibnamefont {{Perez-Garcia}}},\ }\bibinfo {title} {\emph {Matrix
  {{Product States}}, {{Random Matrix Theory}} and the {{Principle}} of
  {{Maximum Entropy}}}},\ \href {\doibase 10.1007/s00220-013-1718-x} {\bibfield
   {journal} {\bibinfo  {journal} {Communications in Mathematical Physics}\
  }\textbf {\bibinfo {volume} {320}},\ \bibinfo {pages} {663} (\bibinfo {year}
  {2013})}\BibitemShut {NoStop}%
\bibitem [{\citenamefont {Fukuda}\ and\ \citenamefont
  {Koenig}(2019)}]{fukuda_typical_2019}%
  \BibitemOpen
  \bibfield  {author} {\bibinfo {author} {\bibfnamefont {M.}~\bibnamefont
  {Fukuda}}\ and\ \bibinfo {author} {\bibfnamefont {R.}~\bibnamefont
  {Koenig}},\ }\bibinfo {title} {\emph {Typical Entanglement for {{Gaussian}}
  States}},\ \href {\doibase 10.1063/1.5119950} {\bibfield  {journal} {\bibinfo
   {journal} {Journal of Mathematical Physics}\ }\textbf {\bibinfo {volume}
  {60}},\ \bibinfo {pages} {112203} (\bibinfo {year} {2019})}\BibitemShut
  {NoStop}%
\bibitem [{\citenamefont {Iosue}\ \emph {et~al.}(2023)\citenamefont {Iosue},
  \citenamefont {Ehrenberg}, \citenamefont {Hangleiter}, \citenamefont
  {Deshpande},\ and\ \citenamefont {Gorshkov}}]{iosue_page_2023}%
  \BibitemOpen
  \bibfield  {author} {\bibinfo {author} {\bibfnamefont {J.~T.}\ \bibnamefont
  {Iosue}}, \bibinfo {author} {\bibfnamefont {A.}~\bibnamefont {Ehrenberg}},
  \bibinfo {author} {\bibfnamefont {D.}~\bibnamefont {Hangleiter}}, \bibinfo
  {author} {\bibfnamefont {A.}~\bibnamefont {Deshpande}}, \ and\ \bibinfo
  {author} {\bibfnamefont {A.~V.}\ \bibnamefont {Gorshkov}},\ }\bibinfo {title}
  {\emph {Page Curves and Typical Entanglement in Linear Optics}},\ \href
  {\doibase 10.22331/q-2023-05-23-1017} {\bibfield  {journal} {\bibinfo
  {journal} {Quantum}\ }\textbf {\bibinfo {volume} {7}},\ \bibinfo {pages}
  {1017} (\bibinfo {year} {2023})}\BibitemShut {NoStop}%
\bibitem [{\citenamefont {Mnih}\ \emph {et~al.}(2008)\citenamefont {Mnih},
  \citenamefont {Szepesv{\'a}ri},\ and\ \citenamefont
  {Audibert}}]{mnih_empirical_2008}%
  \BibitemOpen
  \bibfield  {author} {\bibinfo {author} {\bibfnamefont {V.}~\bibnamefont
  {Mnih}}, \bibinfo {author} {\bibfnamefont {C.}~\bibnamefont
  {Szepesv{\'a}ri}}, \ and\ \bibinfo {author} {\bibfnamefont {J.-Y.}\
  \bibnamefont {Audibert}},\ }\bibfield  {title} {\emph {\bibinfo {title}
  {\emph {Empirical {{Bernstein}} Stopping}},\ }}in\ \href {\doibase
  10.1145/1390156.1390241} {\emph {\bibinfo {booktitle} {Proceedings of the
  25th International Conference on {{Machine}} Learning}}},\ \bibinfo {series
  and number} {{{ICML}} '08}\ (\bibinfo  {publisher} {{Association for
  Computing Machinery}},\ \bibinfo {address} {{Helsinki, Finland}},\ \bibinfo
  {year} {2008})\ pp.\ \bibinfo {pages} {672--679}\BibitemShut {NoStop}%
\bibitem [{\citenamefont {Erd{\"o}s}\ and\ \citenamefont
  {R{\'e}nyi}(1965)}]{erdos_probabilistic_1965}%
  \BibitemOpen
  \bibfield  {author} {\bibinfo {author} {\bibfnamefont {P.}~\bibnamefont
  {Erd{\"o}s}}\ and\ \bibinfo {author} {\bibfnamefont {A.}~\bibnamefont
  {R{\'e}nyi}},\ }\bibinfo {title} {\emph {Probabilistic Methods in Group
  Theory}},\ \href {\doibase 10.1007/BF02806383} {\bibfield  {journal}
  {\bibinfo  {journal} {J. Anal. Math.}\ }\textbf {\bibinfo {volume} {14}},\
  \bibinfo {pages} {127} (\bibinfo {year} {1965})}\BibitemShut {NoStop}%
\bibitem [{\citenamefont {Grewal}\ \emph
  {et~al.}(2023{\natexlab{b}})\citenamefont {Grewal}, \citenamefont {Iyer},
  \citenamefont {Kretschmer},\ and\ \citenamefont
  {Liang}}]{grewal_efficient_2023}%
  \BibitemOpen
  \bibfield  {author} {\bibinfo {author} {\bibfnamefont {S.}~\bibnamefont
  {Grewal}}, \bibinfo {author} {\bibfnamefont {V.}~\bibnamefont {Iyer}},
  \bibinfo {author} {\bibfnamefont {W.}~\bibnamefont {Kretschmer}}, \ and\
  \bibinfo {author} {\bibfnamefont {D.}~\bibnamefont {Liang}},\ }\bibinfo
  {title} {\emph {Efficient {{Learning}} of {{Quantum States Prepared With Few
  Non-Clifford Gates II}}: {{Single-Copy Measurements}}}},\ \href@noop {} {\
  (\bibinfo {year} {2023}{\natexlab{b}})},\ \Eprint
  {http://arxiv.org/abs/2308.07175} {arXiv:2308.07175}\BibitemShut {NoStop}%
\bibitem [{\citenamefont {Hastings}(2010)}]{Hastings10}%
  \BibitemOpen
  \bibfield  {author} {\bibinfo {author} {\bibfnamefont {M.~B.}\ \bibnamefont
  {Hastings}},\ }\bibinfo {title} {\emph {Locality in Quantum Systems}},\
  \href@noop {} {\  (\bibinfo {year} {2010})},\ \Eprint
  {http://arxiv.org/abs/1008.5137} {arXiv:1008.5137 [math-ph]}\BibitemShut
  {NoStop}%
\end{thebibliography}%


\begin{thebibliography}{34}%
\makeatletter
\providecommand \@ifxundefined [1]{%
 \@ifx{#1\undefined}
}%
\providecommand \@ifnum [1]{%
 \ifnum #1\expandafter \@firstoftwo
 \else \expandafter \@secondoftwo
 \fi
}%
\providecommand \@ifx [1]{%
 \ifx #1\expandafter \@firstoftwo
 \else \expandafter \@secondoftwo
 \fi
}%
\providecommand \natexlab [1]{#1}%
\providecommand \enquote  [1]{``#1''}%
\providecommand \bibnamefont  [1]{#1}%
\providecommand \bibfnamefont [1]{#1}%
\providecommand \citenamefont [1]{#1}%
\providecommand \href@noop [0]{\@secondoftwo}%
\providecommand \href [0]{\begingroup \@sanitize@url \@href}%
\providecommand \@href[1]{\@@startlink{#1}\@@href}%
\providecommand \@@href[1]{\endgroup#1\@@endlink}%
\providecommand \@sanitize@url [0]{\catcode `\\12\catcode `\$12\catcode
  `\&12\catcode `\#12\catcode `\^12\catcode `\_12\catcode `\%12\relax}%
\providecommand \@@startlink[1]{}%
\providecommand \@@endlink[0]{}%
\providecommand \url  [0]{\begingroup\@sanitize@url \@url }%
\providecommand \@url [1]{\endgroup\@href {#1}{\urlprefix }}%
\providecommand \urlprefix  [0]{URL }%
\providecommand \Eprint [0]{\href }%
\providecommand \doibase [0]{http://dx.doi.org/}%
\providecommand \selectlanguage [0]{\@gobble}%
\providecommand \bibinfo  [0]{\@secondoftwo}%
\providecommand \bibfield  [0]{\@secondoftwo}%
\providecommand \translation [1]{[#1]}%
\providecommand \BibitemOpen [0]{}%
\providecommand \bibitemStop [0]{}%
\providecommand \bibitemNoStop [0]{.\EOS\space}%
\providecommand \EOS [0]{\spacefactor3000\relax}%
\providecommand \BibitemShut  [1]{\csname bibitem#1\endcsname}%
\let\auto@bib@innerbib\@empty
\bibitem [{\citenamefont {Hangleiter}\ and\ \citenamefont
  {Eisert}(2023)}]{hangleiter_computational_2023-1}%
  \BibitemOpen
  \bibfield  {author} {\bibinfo {author} {\bibfnamefont {D.}~\bibnamefont
  {Hangleiter}}\ and\ \bibinfo {author} {\bibfnamefont {J.}~\bibnamefont
  {Eisert}},\ }\bibinfo {title} {\emph {Computational Advantage of Quantum
  Random Sampling}},\ \href {\doibase 10.1103/RevModPhys.95.035001} {\bibfield
  {journal} {\bibinfo  {journal} {Rev. Mod. Phys.}\ }\textbf {\bibinfo {volume}
  {95}},\ \bibinfo {pages} {035001} (\bibinfo {year} {2023})}\BibitemShut
  {NoStop}%
\bibitem [{\citenamefont {Bremner}\ \emph {et~al.}(2016)\citenamefont
  {Bremner}, \citenamefont {Montanaro},\ and\ \citenamefont
  {Shepherd}}]{bremner_average-case_2016}%
  \BibitemOpen
  \bibfield  {author} {\bibinfo {author} {\bibfnamefont {M.~J.}\ \bibnamefont
  {Bremner}}, \bibinfo {author} {\bibfnamefont {A.}~\bibnamefont {Montanaro}},
  \ and\ \bibinfo {author} {\bibfnamefont {D.~J.}\ \bibnamefont {Shepherd}},\
  }\bibinfo {title} {\emph {Average-{{Case Complexity Versus Approximate
  Simulation}} of {{Commuting Quantum Computations}}}},\ \href {\doibase
  10.1103/PhysRevLett.117.080501} {\bibfield  {journal} {\bibinfo  {journal}
  {Physical Review Letters}\ }\textbf {\bibinfo {volume} {117}},\ \bibinfo
  {pages} {080501} (\bibinfo {year} {2016})}\BibitemShut {NoStop}%
\bibitem [{\citenamefont {Krovi}(2022)}]{krovi_average-case_2022}%
  \BibitemOpen
  \bibfield  {author} {\bibinfo {author} {\bibfnamefont {H.}~\bibnamefont
  {Krovi}},\ }\bibinfo {title} {\emph {Average-Case Hardness of Estimating
  Probabilities of Random Quantum Circuits with a Linear Scaling in the Error
  Exponent}},\ \href@noop {} {\  (\bibinfo {year} {2022})},\ \Eprint
  {http://arxiv.org/abs/2206.05642} {arXiv:2206.05642}\BibitemShut {NoStop}%
\bibitem [{\citenamefont {Zhu}\ \emph {et~al.}(2016)\citenamefont {Zhu},
  \citenamefont {Kueng}, \citenamefont {Grassl},\ and\ \citenamefont
  {Gross}}]{zhu_clifford_2016}%
  \BibitemOpen
  \bibfield  {author} {\bibinfo {author} {\bibfnamefont {H.}~\bibnamefont
  {Zhu}}, \bibinfo {author} {\bibfnamefont {R.}~\bibnamefont {Kueng}}, \bibinfo
  {author} {\bibfnamefont {M.}~\bibnamefont {Grassl}}, \ and\ \bibinfo {author}
  {\bibfnamefont {D.}~\bibnamefont {Gross}},\ }\bibinfo {title} {\emph {The
  {{Clifford}} Group Fails Gracefully to Be a Unitary 4-Design}},\ \href@noop
  {} {\  (\bibinfo {year} {2016})},\ \Eprint {http://arxiv.org/abs/1609.08172}
  {arXiv:1609.08172}\BibitemShut {NoStop}%
\bibitem [{\citenamefont {Brand{\~a}o}\ \emph {et~al.}(2016)\citenamefont
  {Brand{\~a}o}, \citenamefont {Harrow},\ and\ \citenamefont
  {Horodecki}}]{brandao_local_2016}%
  \BibitemOpen
  \bibfield  {author} {\bibinfo {author} {\bibfnamefont {F.~G. S.~L.}\
  \bibnamefont {Brand{\~a}o}}, \bibinfo {author} {\bibfnamefont {A.~W.}\
  \bibnamefont {Harrow}}, \ and\ \bibinfo {author} {\bibfnamefont
  {M.}~\bibnamefont {Horodecki}},\ }\bibinfo {title} {\emph {Local {{Random
  Quantum Circuits}} Are {{Approximate Polynomial-Designs}}}},\ \href {\doibase
  10.1007/s00220-016-2706-8} {\bibfield  {journal} {\bibinfo  {journal}
  {Commun. Math. Phys.}\ }\textbf {\bibinfo {volume} {346}},\ \bibinfo {pages}
  {397} (\bibinfo {year} {2016})}\BibitemShut {NoStop}%
\bibitem [{\citenamefont {Zhou}\ and\ \citenamefont
  {Nahum}(2019)}]{zhou_emergent_2019}%
  \BibitemOpen
  \bibfield  {author} {\bibinfo {author} {\bibfnamefont {T.}~\bibnamefont
  {Zhou}}\ and\ \bibinfo {author} {\bibfnamefont {A.}~\bibnamefont {Nahum}},\
  }\bibinfo {title} {\emph {Emergent Statistical Mechanics of Entanglement in
  Random Unitary Circuits}},\ \href {\doibase 10.1103/PhysRevB.99.174205}
  {\bibfield  {journal} {\bibinfo  {journal} {Phys. Rev. B}\ }\textbf {\bibinfo
  {volume} {99}},\ \bibinfo {pages} {174205} (\bibinfo {year}
  {2019})}\BibitemShut {NoStop}%
\bibitem [{\citenamefont {{Hunter-Jones}}(2019)}]{hunter-jones_unitary_2019}%
  \BibitemOpen
  \bibfield  {author} {\bibinfo {author} {\bibfnamefont {N.}~\bibnamefont
  {{Hunter-Jones}}},\ }\bibinfo {title} {\emph {Unitary Designs from
  Statistical Mechanics in Random Quantum Circuits}},\ \href@noop {} {\
  (\bibinfo {year} {2019})},\ \Eprint {http://arxiv.org/abs/1905.12053}
  {arXiv:1905.12053}\BibitemShut {NoStop}%
\bibitem [{\citenamefont {Barak}\ \emph {et~al.}(2021)\citenamefont {Barak},
  \citenamefont {Chou},\ and\ \citenamefont {Gao}}]{barak_spoofing_2021}%
  \BibitemOpen
  \bibfield  {author} {\bibinfo {author} {\bibfnamefont {B.}~\bibnamefont
  {Barak}}, \bibinfo {author} {\bibfnamefont {C.-N.}\ \bibnamefont {Chou}}, \
  and\ \bibinfo {author} {\bibfnamefont {X.}~\bibnamefont {Gao}},\ }\bibfield
  {title} {\emph {\bibinfo {title} {\emph {Spoofing {{Linear Cross-Entropy
  Benchmarking}} in {{Shallow Quantum Circuits}}}},\ }}in\ \href {\doibase
  10.4230/LIPIcs.ITCS.2021.30} {\emph {\bibinfo {booktitle} {12th
  {{Innovations}} in {{Theoretical Computer Science Conference}} ({{ITCS}}
  2021)}}},\ \bibinfo {series} {Leibniz {{International Proceedings}} in
  {{Informatics}} ({{LIPIcs}})}, Vol.\ \bibinfo {volume} {185},\ \bibinfo
  {editor} {edited by\ \bibinfo {editor} {\bibfnamefont {J.~R.}\ \bibnamefont
  {Lee}}}\ (\bibinfo  {publisher} {{Schloss Dagstuhl\textendash Leibniz-Zentrum
  f\"ur Informatik}},\ \bibinfo {address} {{Dagstuhl, Germany}},\ \bibinfo
  {year} {2021})\ pp.\ \bibinfo {pages} {30:1--30:20},\ \Eprint
  {http://arxiv.org/abs/2005.02421} {arXiv:2005.02421}\BibitemShut {NoStop}%
\bibitem [{\citenamefont {Dalzell}\ \emph {et~al.}(2022)\citenamefont
  {Dalzell}, \citenamefont {{Hunter-Jones}},\ and\ \citenamefont
  {Brand{\~a}o}}]{dalzell_random_2022}%
  \BibitemOpen
  \bibfield  {author} {\bibinfo {author} {\bibfnamefont {A.~M.}\ \bibnamefont
  {Dalzell}}, \bibinfo {author} {\bibfnamefont {N.}~\bibnamefont
  {{Hunter-Jones}}}, \ and\ \bibinfo {author} {\bibfnamefont {F.~G. S.~L.}\
  \bibnamefont {Brand{\~a}o}},\ }\bibinfo {title} {\emph {Random {{Quantum
  Circuits Anticoncentrate}} in {{Log Depth}}}},\ \href {\doibase
  10.1103/PRXQuantum.3.010333} {\bibfield  {journal} {\bibinfo  {journal} {PRX
  Quantum}\ }\textbf {\bibinfo {volume} {3}},\ \bibinfo {pages} {010333}
  (\bibinfo {year} {2022})}\BibitemShut {NoStop}%
\bibitem [{\citenamefont {Wallman}\ and\ \citenamefont
  {Emerson}(2016)}]{wallman_noise_2016}%
  \BibitemOpen
  \bibfield  {author} {\bibinfo {author} {\bibfnamefont {J.~J.}\ \bibnamefont
  {Wallman}}\ and\ \bibinfo {author} {\bibfnamefont {J.}~\bibnamefont
  {Emerson}},\ }\bibinfo {title} {\emph {Noise Tailoring for Scalable Quantum
  Computation via Randomized Compiling}},\ \href {\doibase
  10.1103/PhysRevA.94.052325} {\bibfield  {journal} {\bibinfo  {journal} {Phys.
  Rev. A}\ }\textbf {\bibinfo {volume} {94}},\ \bibinfo {pages} {052325}
  (\bibinfo {year} {2016})}\BibitemShut {NoStop}%
\bibitem [{\citenamefont {Ware}\ \emph {et~al.}(2023)\citenamefont {Ware},
  \citenamefont {Deshpande}, \citenamefont {Hangleiter}, \citenamefont
  {Niroula}, \citenamefont {Fefferman}, \citenamefont {Gorshkov},\ and\
  \citenamefont {Gullans}}]{ware_sharp_2023}%
  \BibitemOpen
  \bibfield  {author} {\bibinfo {author} {\bibfnamefont {B.}~\bibnamefont
  {Ware}}, \bibinfo {author} {\bibfnamefont {A.}~\bibnamefont {Deshpande}},
  \bibinfo {author} {\bibfnamefont {D.}~\bibnamefont {Hangleiter}}, \bibinfo
  {author} {\bibfnamefont {P.}~\bibnamefont {Niroula}}, \bibinfo {author}
  {\bibfnamefont {B.}~\bibnamefont {Fefferman}}, \bibinfo {author}
  {\bibfnamefont {A.~V.}\ \bibnamefont {Gorshkov}}, \ and\ \bibinfo {author}
  {\bibfnamefont {M.~J.}\ \bibnamefont {Gullans}},\ }\bibinfo {title} {\emph {A
  Sharp Phase Transition in Linear Cross-Entropy Benchmarking}},\ \href@noop {}
  {\  (\bibinfo {year} {2023})},\ \Eprint {http://arxiv.org/abs/2305.04954}
  {arXiv:2305.04954}\BibitemShut {NoStop}%
\bibitem [{\citenamefont {Morvan}\ \emph {et~al.}(2023)\citenamefont {Morvan},
  \citenamefont {Villalonga}, \citenamefont {Mi}, \citenamefont {Mandr{\`a}},
  \citenamefont {Bengtsson}, \citenamefont {Klimov}, \citenamefont {Chen},
  \citenamefont {Hong}, \citenamefont {Erickson}, \citenamefont {Drozdov},
  \citenamefont {Chau}, \citenamefont {Laun}, \citenamefont {Movassagh},
  \citenamefont {Asfaw}, \citenamefont {Brand{\~a}o}, \citenamefont {Peralta},
  \citenamefont {Abanin}, \citenamefont {Acharya}, \citenamefont {Allen},
  \citenamefont {Andersen},\ and\ \citenamefont
  {\emph{others}}}]{morvan_phase_2023}%
  \BibitemOpen
  \bibfield  {author} {\bibinfo {author} {\bibfnamefont {A.}~\bibnamefont
  {Morvan}}, \bibinfo {author} {\bibfnamefont {B.}~\bibnamefont {Villalonga}},
  \bibinfo {author} {\bibfnamefont {X.}~\bibnamefont {Mi}}, \bibinfo {author}
  {\bibfnamefont {S.}~\bibnamefont {Mandr{\`a}}}, \bibinfo {author}
  {\bibfnamefont {A.}~\bibnamefont {Bengtsson}}, \bibinfo {author}
  {\bibfnamefont {P.~V.}\ \bibnamefont {Klimov}}, \bibinfo {author}
  {\bibfnamefont {Z.}~\bibnamefont {Chen}}, \bibinfo {author} {\bibfnamefont
  {S.}~\bibnamefont {Hong}}, \bibinfo {author} {\bibfnamefont {C.}~\bibnamefont
  {Erickson}}, \bibinfo {author} {\bibfnamefont {I.~K.}\ \bibnamefont
  {Drozdov}}, \bibinfo {author} {\bibfnamefont {J.}~\bibnamefont {Chau}},
  \bibinfo {author} {\bibfnamefont {G.}~\bibnamefont {Laun}}, \bibinfo {author}
  {\bibfnamefont {R.}~\bibnamefont {Movassagh}}, \bibinfo {author}
  {\bibfnamefont {A.}~\bibnamefont {Asfaw}}, \bibinfo {author} {\bibfnamefont
  {L.~T. A.~N.}\ \bibnamefont {Brand{\~a}o}}, \bibinfo {author} {\bibfnamefont
  {R.}~\bibnamefont {Peralta}}, \bibinfo {author} {\bibfnamefont
  {D.}~\bibnamefont {Abanin}}, \bibinfo {author} {\bibfnamefont
  {R.}~\bibnamefont {Acharya}}, \bibinfo {author} {\bibfnamefont
  {R.}~\bibnamefont {Allen}}, \bibinfo {author} {\bibfnamefont {T.~I.}\
  \bibnamefont {Andersen}}, \ and\ \bibinfo {author} {\bibnamefont
  {\emph{others}}},\ }\bibinfo {title} {\emph {Phase Transition in {{Random
  Circuit Sampling}}}},\ \href@noop {} {\  (\bibinfo {year} {2023})},\ \Eprint
  {http://arxiv.org/abs/2304.11119} {arXiv:2304.11119}\BibitemShut {NoStop}%
\bibitem [{\citenamefont {Sommers}\ \emph {et~al.}(2023)\citenamefont
  {Sommers}, \citenamefont {Huse},\ and\ \citenamefont
  {Gullans}}]{sommers_crystalline_2023}%
  \BibitemOpen
  \bibfield  {author} {\bibinfo {author} {\bibfnamefont {G.~M.}\ \bibnamefont
  {Sommers}}, \bibinfo {author} {\bibfnamefont {D.~A.}\ \bibnamefont {Huse}}, \
  and\ \bibinfo {author} {\bibfnamefont {M.~J.}\ \bibnamefont {Gullans}},\
  }\bibinfo {title} {\emph {Crystalline {{Quantum Circuits}}}},\ \href
  {\doibase 10.1103/PRXQuantum.4.030313} {\bibfield  {journal} {\bibinfo
  {journal} {PRX Quantum}\ }\textbf {\bibinfo {volume} {4}},\ \bibinfo {pages}
  {030313} (\bibinfo {year} {2023})}\BibitemShut {NoStop}%
\bibitem [{\citenamefont {Dalzell}\ \emph {et~al.}(2021)\citenamefont
  {Dalzell}, \citenamefont {{Hunter-Jones}},\ and\ \citenamefont
  {Brand{\~a}o}}]{dalzell_random_2021}%
  \BibitemOpen
  \bibfield  {author} {\bibinfo {author} {\bibfnamefont {A.~M.}\ \bibnamefont
  {Dalzell}}, \bibinfo {author} {\bibfnamefont {N.}~\bibnamefont
  {{Hunter-Jones}}}, \ and\ \bibinfo {author} {\bibfnamefont {F.~G. S.~L.}\
  \bibnamefont {Brand{\~a}o}},\ }\bibinfo {title} {\emph {Random Quantum
  Circuits Transform Local Noise into Global White Noise}},\ \href@noop {} {\
  (\bibinfo {year} {2021})},\ \Eprint {http://arxiv.org/abs/2111.14907}
  {arXiv:2111.14907}\BibitemShut {NoStop}%
\bibitem [{\citenamefont {Choi}\ \emph {et~al.}(2023)\citenamefont {Choi},
  \citenamefont {Shaw}, \citenamefont {Madjarov}, \citenamefont {Xie},
  \citenamefont {Finkelstein}, \citenamefont {Covey}, \citenamefont {Cotler},
  \citenamefont {Mark}, \citenamefont {Huang}, \citenamefont {Kale},
  \citenamefont {Pichler}, \citenamefont {Brand{\~a}o}, \citenamefont {Choi},\
  and\ \citenamefont {Endres}}]{choi_preparing_2023}%
  \BibitemOpen
  \bibfield  {author} {\bibinfo {author} {\bibfnamefont {J.}~\bibnamefont
  {Choi}}, \bibinfo {author} {\bibfnamefont {A.~L.}\ \bibnamefont {Shaw}},
  \bibinfo {author} {\bibfnamefont {I.~S.}\ \bibnamefont {Madjarov}}, \bibinfo
  {author} {\bibfnamefont {X.}~\bibnamefont {Xie}}, \bibinfo {author}
  {\bibfnamefont {R.}~\bibnamefont {Finkelstein}}, \bibinfo {author}
  {\bibfnamefont {J.~P.}\ \bibnamefont {Covey}}, \bibinfo {author}
  {\bibfnamefont {J.~S.}\ \bibnamefont {Cotler}}, \bibinfo {author}
  {\bibfnamefont {D.~K.}\ \bibnamefont {Mark}}, \bibinfo {author}
  {\bibfnamefont {H.-Y.}\ \bibnamefont {Huang}}, \bibinfo {author}
  {\bibfnamefont {A.}~\bibnamefont {Kale}}, \bibinfo {author} {\bibfnamefont
  {H.}~\bibnamefont {Pichler}}, \bibinfo {author} {\bibfnamefont {F.~G. S.~L.}\
  \bibnamefont {Brand{\~a}o}}, \bibinfo {author} {\bibfnamefont
  {S.}~\bibnamefont {Choi}}, \ and\ \bibinfo {author} {\bibfnamefont
  {M.}~\bibnamefont {Endres}},\ }\bibinfo {title} {\emph {Preparing Random
  States and Benchmarking with Many-Body Quantum Chaos}},\ \href {\doibase
  10.1038/s41586-022-05442-1} {\bibfield  {journal} {\bibinfo  {journal}
  {Nature}\ }\textbf {\bibinfo {volume} {613}},\ \bibinfo {pages} {468}
  (\bibinfo {year} {2023})}\BibitemShut {NoStop}%
\bibitem [{\citenamefont {Ringbauer}\ \emph {et~al.}(2023)\citenamefont
  {Ringbauer}, \citenamefont {Hinsche}, \citenamefont {Feldker}, \citenamefont
  {Faehrmann}, \citenamefont {{Bermejo-Vega}}, \citenamefont {Edmunds},
  \citenamefont {Postler}, \citenamefont {Stricker}, \citenamefont {Marciniak},
  \citenamefont {Meth}, \citenamefont {Pogorelov}, \citenamefont {Blatt},
  \citenamefont {Schindler}, \citenamefont {Eisert}, \citenamefont {Monz},\
  and\ \citenamefont {Hangleiter}}]{ringbauer_verifiable_2023}%
  \BibitemOpen
  \bibfield  {author} {\bibinfo {author} {\bibfnamefont {M.}~\bibnamefont
  {Ringbauer}}, \bibinfo {author} {\bibfnamefont {M.}~\bibnamefont {Hinsche}},
  \bibinfo {author} {\bibfnamefont {T.}~\bibnamefont {Feldker}}, \bibinfo
  {author} {\bibfnamefont {P.~K.}\ \bibnamefont {Faehrmann}}, \bibinfo {author}
  {\bibfnamefont {J.}~\bibnamefont {{Bermejo-Vega}}}, \bibinfo {author}
  {\bibfnamefont {C.}~\bibnamefont {Edmunds}}, \bibinfo {author} {\bibfnamefont
  {L.}~\bibnamefont {Postler}}, \bibinfo {author} {\bibfnamefont
  {R.}~\bibnamefont {Stricker}}, \bibinfo {author} {\bibfnamefont {C.~D.}\
  \bibnamefont {Marciniak}}, \bibinfo {author} {\bibfnamefont {M.}~\bibnamefont
  {Meth}}, \bibinfo {author} {\bibfnamefont {I.}~\bibnamefont {Pogorelov}},
  \bibinfo {author} {\bibfnamefont {R.}~\bibnamefont {Blatt}}, \bibinfo
  {author} {\bibfnamefont {P.}~\bibnamefont {Schindler}}, \bibinfo {author}
  {\bibfnamefont {J.}~\bibnamefont {Eisert}}, \bibinfo {author} {\bibfnamefont
  {T.}~\bibnamefont {Monz}}, \ and\ \bibinfo {author} {\bibfnamefont
  {D.}~\bibnamefont {Hangleiter}},\ }\bibinfo {title} {\emph {Verifiable
  Measurement-Based Quantum Random Sampling with Trapped Ions}},\ \href@noop {}
  {\  (\bibinfo {year} {2023})},\ \Eprint {http://arxiv.org/abs/2307.14424}
  {arXiv:2307.14424}\BibitemShut {NoStop}%
\bibitem [{\citenamefont {Flammia}\ and\ \citenamefont
  {Liu}(2011)}]{flammia_direct_2011}%
  \BibitemOpen
  \bibfield  {author} {\bibinfo {author} {\bibfnamefont {S.~T.}\ \bibnamefont
  {Flammia}}\ and\ \bibinfo {author} {\bibfnamefont {Y.-K.}\ \bibnamefont
  {Liu}},\ }\bibinfo {title} {\emph {Direct {{Fidelity Estimation}} from {{Few
  Pauli Measurements}}}},\ \href {\doibase 10.1103/PhysRevLett.106.230501}
  {\bibfield  {journal} {\bibinfo  {journal} {Phys. Rev. Lett.}\ }\textbf
  {\bibinfo {volume} {106}},\ \bibinfo {pages} {230501} (\bibinfo {year}
  {2011})}\BibitemShut {NoStop}%
\bibitem [{\citenamefont {Fitzsimons}\ and\ \citenamefont
  {Kashefi}(2017)}]{fitzsimons_unconditionally_2017}%
  \BibitemOpen
  \bibfield  {author} {\bibinfo {author} {\bibfnamefont {J.~F.}\ \bibnamefont
  {Fitzsimons}}\ and\ \bibinfo {author} {\bibfnamefont {E.}~\bibnamefont
  {Kashefi}},\ }\bibinfo {title} {\emph {Unconditionally Verifiable Blind
  Quantum Computation}},\ \href {\doibase 10.1103/PhysRevA.96.012303}
  {\bibfield  {journal} {\bibinfo  {journal} {Phys. Rev. A}\ }\textbf {\bibinfo
  {volume} {96}},\ \bibinfo {pages} {012303} (\bibinfo {year}
  {2017})}\BibitemShut {NoStop}%
\bibitem [{\citenamefont {Reichardt}\ \emph {et~al.}(2013)\citenamefont
  {Reichardt}, \citenamefont {Unger},\ and\ \citenamefont
  {Vazirani}}]{reichardt_classical_2013}%
  \BibitemOpen
  \bibfield  {author} {\bibinfo {author} {\bibfnamefont {B.~W.}\ \bibnamefont
  {Reichardt}}, \bibinfo {author} {\bibfnamefont {F.}~\bibnamefont {Unger}}, \
  and\ \bibinfo {author} {\bibfnamefont {U.}~\bibnamefont {Vazirani}},\
  }\bibinfo {title} {\emph {Classical Command of Quantum Systems}},\ \href
  {\doibase 10.1038/nature12035} {\bibfield  {journal} {\bibinfo  {journal}
  {Nature}\ }\textbf {\bibinfo {volume} {496}},\ \bibinfo {pages} {456}
  (\bibinfo {year} {2013})}\BibitemShut {NoStop}%
\bibitem [{\citenamefont {Mahadev}(2018)}]{mahadev_classical_2018-1}%
  \BibitemOpen
  \bibfield  {author} {\bibinfo {author} {\bibfnamefont {U.}~\bibnamefont
  {Mahadev}},\ }\bibfield  {title} {\emph {\bibinfo {title} {\emph {Classical
  {{Verification}} of {{Quantum Computations}}}},\ }}in\ \href {\doibase
  10.1109/FOCS.2018.00033} {\emph {\bibinfo {booktitle} {2018 {{IEEE}} 59th
  {{Annual Symposium}} on {{Foundations}} of {{Computer Science}}
  ({{FOCS}})}}}\ (\bibinfo {year} {2018})\ pp.\ \bibinfo {pages}
  {259--267}\BibitemShut {NoStop}%
\bibitem [{\citenamefont {Page}(1993)}]{page_average_1993}%
  \BibitemOpen
  \bibfield  {author} {\bibinfo {author} {\bibfnamefont {D.~N.}\ \bibnamefont
  {Page}},\ }\bibinfo {title} {\emph {Average Entropy of a Subsystem}},\ \href
  {\doibase 10.1103/PhysRevLett.71.1291} {\bibfield  {journal} {\bibinfo
  {journal} {Phys. Rev. Lett.}\ }\textbf {\bibinfo {volume} {71}},\ \bibinfo
  {pages} {1291} (\bibinfo {year} {1993})}\BibitemShut {NoStop}%
\bibitem [{\citenamefont {Garnerone}\ \emph {et~al.}(2010)\citenamefont
  {Garnerone}, \citenamefont {{de Oliveira}},\ and\ \citenamefont
  {Zanardi}}]{garnerone_typicality_2010}%
  \BibitemOpen
  \bibfield  {author} {\bibinfo {author} {\bibfnamefont {S.}~\bibnamefont
  {Garnerone}}, \bibinfo {author} {\bibfnamefont {T.~R.}\ \bibnamefont {{de
  Oliveira}}}, \ and\ \bibinfo {author} {\bibfnamefont {P.}~\bibnamefont
  {Zanardi}},\ }\bibinfo {title} {\emph {Typicality in Random Matrix Product
  States}},\ \href {\doibase 10.1103/PhysRevA.81.032336} {\bibfield  {journal}
  {\bibinfo  {journal} {Phys. Rev. A}\ }\textbf {\bibinfo {volume} {81}},\
  \bibinfo {pages} {032336} (\bibinfo {year} {2010})}\BibitemShut {NoStop}%
\bibitem [{\citenamefont {Collins}\ \emph {et~al.}(2013)\citenamefont
  {Collins}, \citenamefont {{Gonzalez-Guillen}},\ and\ \citenamefont
  {{Perez-Garcia}}}]{collins_matrix_2013}%
  \BibitemOpen
  \bibfield  {author} {\bibinfo {author} {\bibfnamefont {B.}~\bibnamefont
  {Collins}}, \bibinfo {author} {\bibfnamefont {C.~E.}\ \bibnamefont
  {{Gonzalez-Guillen}}}, \ and\ \bibinfo {author} {\bibfnamefont
  {D.}~\bibnamefont {{Perez-Garcia}}},\ }\bibinfo {title} {\emph {Matrix
  {{Product States}}, {{Random Matrix Theory}} and the {{Principle}} of
  {{Maximum Entropy}}}},\ \href {\doibase 10.1007/s00220-013-1718-x} {\bibfield
   {journal} {\bibinfo  {journal} {Communications in Mathematical Physics}\
  }\textbf {\bibinfo {volume} {320}},\ \bibinfo {pages} {663} (\bibinfo {year}
  {2013})}\BibitemShut {NoStop}%
\bibitem [{\citenamefont {Fukuda}\ and\ \citenamefont
  {Koenig}(2019)}]{fukuda_typical_2019}%
  \BibitemOpen
  \bibfield  {author} {\bibinfo {author} {\bibfnamefont {M.}~\bibnamefont
  {Fukuda}}\ and\ \bibinfo {author} {\bibfnamefont {R.}~\bibnamefont
  {Koenig}},\ }\bibinfo {title} {\emph {Typical Entanglement for {{Gaussian}}
  States}},\ \href {\doibase 10.1063/1.5119950} {\bibfield  {journal} {\bibinfo
   {journal} {Journal of Mathematical Physics}\ }\textbf {\bibinfo {volume}
  {60}},\ \bibinfo {pages} {112203} (\bibinfo {year} {2019})}\BibitemShut
  {NoStop}%
\bibitem [{\citenamefont {Iosue}\ \emph {et~al.}(2023)\citenamefont {Iosue},
  \citenamefont {Ehrenberg}, \citenamefont {Hangleiter}, \citenamefont
  {Deshpande},\ and\ \citenamefont {Gorshkov}}]{iosue_page_2023}%
  \BibitemOpen
  \bibfield  {author} {\bibinfo {author} {\bibfnamefont {J.~T.}\ \bibnamefont
  {Iosue}}, \bibinfo {author} {\bibfnamefont {A.}~\bibnamefont {Ehrenberg}},
  \bibinfo {author} {\bibfnamefont {D.}~\bibnamefont {Hangleiter}}, \bibinfo
  {author} {\bibfnamefont {A.}~\bibnamefont {Deshpande}}, \ and\ \bibinfo
  {author} {\bibfnamefont {A.~V.}\ \bibnamefont {Gorshkov}},\ }\bibinfo {title}
  {\emph {Page Curves and Typical Entanglement in Linear Optics}},\ \href
  {\doibase 10.22331/q-2023-05-23-1017} {\bibfield  {journal} {\bibinfo
  {journal} {Quantum}\ }\textbf {\bibinfo {volume} {7}},\ \bibinfo {pages}
  {1017} (\bibinfo {year} {2023})}\BibitemShut {NoStop}%
\bibitem [{\citenamefont {Montanaro}(2017)}]{montanaro_learning_2017}%
  \BibitemOpen
  \bibfield  {author} {\bibinfo {author} {\bibfnamefont {A.}~\bibnamefont
  {Montanaro}},\ }\bibinfo {title} {\emph {Learning Stabilizer States by
  {{Bell}} Sampling}},\ \href@noop {} {\  (\bibinfo {year} {2017})},\ \Eprint
  {http://arxiv.org/abs/1707.04012} {arXiv:1707.04012}\BibitemShut {NoStop}%
\bibitem [{\citenamefont {Gu{\c t}{\u a}}\ \emph {et~al.}(2020)\citenamefont
  {Gu{\c t}{\u a}}, \citenamefont {Kahn}, \citenamefont {Kueng},\ and\
  \citenamefont {Tropp}}]{guta_fast_2020}%
  \BibitemOpen
  \bibfield  {author} {\bibinfo {author} {\bibfnamefont {M.}~\bibnamefont
  {Gu{\c t}{\u a}}}, \bibinfo {author} {\bibfnamefont {J.}~\bibnamefont
  {Kahn}}, \bibinfo {author} {\bibfnamefont {R.}~\bibnamefont {Kueng}}, \ and\
  \bibinfo {author} {\bibfnamefont {J.~A.}\ \bibnamefont {Tropp}},\ }\bibinfo
  {title} {\emph {Fast State Tomography with Optimal Error Bounds}},\ \href
  {\doibase 10.1088/1751-8121/ab8111} {\bibfield  {journal} {\bibinfo
  {journal} {J. Phys. A: Math. Theor.}\ }\textbf {\bibinfo {volume} {53}},\
  \bibinfo {pages} {204001} (\bibinfo {year} {2020})}\BibitemShut {NoStop}%
\bibitem [{\citenamefont {Huang}\ \emph {et~al.}(2021)\citenamefont {Huang},
  \citenamefont {Kueng},\ and\ \citenamefont
  {Preskill}}]{huang_information-theoretic_2021}%
  \BibitemOpen
  \bibfield  {author} {\bibinfo {author} {\bibfnamefont {H.-Y.}\ \bibnamefont
  {Huang}}, \bibinfo {author} {\bibfnamefont {R.}~\bibnamefont {Kueng}}, \ and\
  \bibinfo {author} {\bibfnamefont {J.}~\bibnamefont {Preskill}},\ }\bibinfo
  {title} {\emph {Information-{{Theoretic Bounds}} on {{Quantum Advantage}} in
  {{Machine Learning}}}},\ \href {\doibase 10.1103/PhysRevLett.126.190505}
  {\bibfield  {journal} {\bibinfo  {journal} {Phys. Rev. Lett.}\ }\textbf
  {\bibinfo {volume} {126}},\ \bibinfo {pages} {190505} (\bibinfo {year}
  {2021})}\BibitemShut {NoStop}%
\bibitem [{\citenamefont {Mnih}\ \emph {et~al.}(2008)\citenamefont {Mnih},
  \citenamefont {Szepesv{\'a}ri},\ and\ \citenamefont
  {Audibert}}]{mnih_empirical_2008}%
  \BibitemOpen
  \bibfield  {author} {\bibinfo {author} {\bibfnamefont {V.}~\bibnamefont
  {Mnih}}, \bibinfo {author} {\bibfnamefont {C.}~\bibnamefont
  {Szepesv{\'a}ri}}, \ and\ \bibinfo {author} {\bibfnamefont {J.-Y.}\
  \bibnamefont {Audibert}},\ }\bibfield  {title} {\emph {\bibinfo {title}
  {\emph {Empirical {{Bernstein}} Stopping}},\ }}in\ \href {\doibase
  10.1145/1390156.1390241} {\emph {\bibinfo {booktitle} {Proceedings of the
  25th International Conference on {{Machine}} Learning}}},\ \bibinfo {series
  and number} {{{ICML}} '08}\ (\bibinfo  {publisher} {{Association for
  Computing Machinery}},\ \bibinfo {address} {{Helsinki, Finland}},\ \bibinfo
  {year} {2008})\ pp.\ \bibinfo {pages} {672--679}\BibitemShut {NoStop}%
\bibitem [{\citenamefont {Erd{\"o}s}\ and\ \citenamefont
  {R{\'e}nyi}(1965)}]{erdos_probabilistic_1965}%
  \BibitemOpen
  \bibfield  {author} {\bibinfo {author} {\bibfnamefont {P.}~\bibnamefont
  {Erd{\"o}s}}\ and\ \bibinfo {author} {\bibfnamefont {A.}~\bibnamefont
  {R{\'e}nyi}},\ }\bibinfo {title} {\emph {Probabilistic Methods in Group
  Theory}},\ \href {\doibase 10.1007/BF02806383} {\bibfield  {journal}
  {\bibinfo  {journal} {J. Anal. Math.}\ }\textbf {\bibinfo {volume} {14}},\
  \bibinfo {pages} {127} (\bibinfo {year} {1965})}\BibitemShut {NoStop}%
\bibitem [{\citenamefont {Grewal}\ \emph {et~al.}(2023)\citenamefont {Grewal},
  \citenamefont {Iyer}, \citenamefont {Kretschmer},\ and\ \citenamefont
  {Liang}}]{grewal_efficient_2023}%
  \BibitemOpen
  \bibfield  {author} {\bibinfo {author} {\bibfnamefont {S.}~\bibnamefont
  {Grewal}}, \bibinfo {author} {\bibfnamefont {V.}~\bibnamefont {Iyer}},
  \bibinfo {author} {\bibfnamefont {W.}~\bibnamefont {Kretschmer}}, \ and\
  \bibinfo {author} {\bibfnamefont {D.}~\bibnamefont {Liang}},\ }\bibinfo
  {title} {\emph {Efficient {{Learning}} of {{Quantum States Prepared With Few
  Non-Clifford Gates II}}: {{Single-Copy Measurements}}}},\ \href@noop {} {\
  (\bibinfo {year} {2023})},\ \Eprint {http://arxiv.org/abs/2308.07175}
  {arXiv:2308.07175}\BibitemShut {NoStop}%
\bibitem [{\citenamefont {Hastings}(2010)}]{Hastings10}%
  \BibitemOpen
  \bibfield  {author} {\bibinfo {author} {\bibfnamefont {M.~B.}\ \bibnamefont
  {Hastings}},\ }\bibinfo {title} {\emph {Locality in Quantum Systems}},\
  \href@noop {} {\  (\bibinfo {year} {2010})},\ \Eprint
  {http://arxiv.org/abs/1008.5137} {arXiv:1008.5137 [math-ph]}\BibitemShut
  {NoStop}%
\bibitem [{\citenamefont {Gao}\ and\ \citenamefont
  {Duan}(2018)}]{gao_efficient_2018}%
  \BibitemOpen
  \bibfield  {author} {\bibinfo {author} {\bibfnamefont {X.}~\bibnamefont
  {Gao}}\ and\ \bibinfo {author} {\bibfnamefont {L.}~\bibnamefont {Duan}},\
  }\bibinfo {title} {\emph {Efficient Classical Simulation of Noisy Quantum
  Computation}},\ \href@noop {} {\  (\bibinfo {year} {2018})},\ \Eprint
  {http://arxiv.org/abs/1810.03176} {arXiv:1810.03176}\BibitemShut {NoStop}%
\bibitem [{\citenamefont {Aharonov}\ \emph {et~al.}(2022)\citenamefont
  {Aharonov}, \citenamefont {Gao}, \citenamefont {Landau}, \citenamefont
  {Liu},\ and\ \citenamefont {Vazirani}}]{aharonov_polynomial-time_2022}%
  \BibitemOpen
  \bibfield  {author} {\bibinfo {author} {\bibfnamefont {D.}~\bibnamefont
  {Aharonov}}, \bibinfo {author} {\bibfnamefont {X.}~\bibnamefont {Gao}},
  \bibinfo {author} {\bibfnamefont {Z.}~\bibnamefont {Landau}}, \bibinfo
  {author} {\bibfnamefont {Y.}~\bibnamefont {Liu}}, \ and\ \bibinfo {author}
  {\bibfnamefont {U.}~\bibnamefont {Vazirani}},\ }\bibinfo {title} {\emph {A
  Polynomial-Time Classical Algorithm for Noisy Random Circuit Sampling}},\
  \href@noop {} {\  (\bibinfo {year} {2022})},\ \Eprint
  {http://arxiv.org/abs/2211.03999} {arXiv:2211.03999}\BibitemShut {NoStop}%
\end{thebibliography}%
\end{document}